\documentclass{article}

 \usepackage[main, final]{neurips_2026}

\usepackage[utf8]{inputenc} 
\usepackage[T1]{fontenc}    
\usepackage{hyperref}       
\usepackage{url}            
\usepackage{booktabs}       
\usepackage{amsfonts}       
\usepackage{nicefrac}       
\usepackage{microtype}      
\usepackage{xcolor}         

\usepackage{subcaption}
\usepackage{algorithmic}
\usepackage[ruled]{algorithm2e} 

\usepackage{amsmath}
\usepackage{amssymb}
\usepackage{mathtools}
\usepackage{amsthm}

\usepackage{placeins} 

\usepackage{tikz}
\usetikzlibrary{arrows.meta}
\setcitestyle{authoryear,round}

\usepackage[capitalize,noabbrev]{cleveref}

\theoremstyle{plain}
\newtheorem{theorem}{Theorem}[section]

\newtheorem{lemma}[theorem]{Lemma}

\theoremstyle{definition}

\newtheorem{assumption}[theorem]{Assumption}
\theoremstyle{remark}
\newtheorem{remark}[theorem]{Remark}

\usepackage[textsize=tiny]{todonotes}

\usepackage{subcaption}

\usepackage{color}
\definecolor{darkred}{RGB}{100,0,0}
\definecolor{darkgreen}{RGB}{0,100,0}
\definecolor{darkblue}{RGB}{0,0,150}

\usepackage{mathtools}

\usepackage{bbm} 
\newcommand{\ind}{\mathbbm{1}}

\newcommand{\PP}{\mathbb{P}}
\newcommand{\E}{\mathbb{E}}

\newcommand{\Ac}{\mathcal{A}}
\newcommand{\Dc}{\mathcal{D}}

\newcommand{\Ic}{\mathcal{I}}

\newcommand{\Cc}{{\mathcal{C}}}

\newcommand{\Dctr}{\mathcal{D}_{\text{tr}}}
\newcommand{\Dccal}{\mathcal{D}_{\text{cal}}}

\newcommand{\Iccal}{\mathcal{I}_{\text{cal}}}
\newcommand{\sttau}{s_\tau}

\newcommand{\consta}{\kappa}

\newcommand{\ttau}{\tau} 

\newcommand{\AIPCW}{\text{AIPCW}}
\newcommand{\DR}{\text{DR}}

\newcommand{\bigCI}{\mathrel{\text{\scalebox{1.07}{$\perp\mkern-10mu\perp$}}}}

\usepackage{soul}

\title{History-Aware Conformal Prediction Sets for Censored Time-to-Event Outcomes}

\author{%
  Yuyao Wang 
  \\
  Department of Biostatistics \\
  Brown University \\
  Providence, RI 02912 \\
  \texttt{yuyao\_wang1@brown.edu} \\
  \And
  Alexander W. Levis \\
  Department of Biostatistics, Epidemiology, and Informatics \\
  University of Pennsylvania \\
  Philadelphia, PA 19104 \\
  \texttt{alexander.levis@pennmedicine.upenn.edu} \\
  \AND
  Shu Yang \\
  Department of Statistics \\
  North Carolina State University \\
  Raleigh, NC 27695 \\
  \texttt{syang24@ncsu.edu} \\
  \And
  Larry Han 
  \\
  Department of Biostatistics \\
  Brown University \\
  Providence, RI 02912 \\
  \texttt{larry\_han@brown.edu} \\
}

\begin{document}

\maketitle

\begin{abstract}
  Existing conformal prediction methods for time-to-event outcomes leverage only baseline covariates, producing prediction intervals that are insufficiently informative to facilitate decision making. We propose History-Aware Prediction Sets (HAPS), a conformal framework that constructs prediction sets for individual event times using covariate histories observed up to a decision time, targeting coverage among individuals who have survived to this time. HAPS handles right censoring adjusted for time-varying confounders via inverse probability of censoring weighting. When the censoring weights are consistently estimated, it achieves PAAC (probably asymptotically approximately correct) coverage among survivors. We further propose two doubly robust extensions of HAPS to weaken reliance on consistent estimation of the censoring distribution. In simulations, HAPS and its extensions reduce median prediction interval length by up to 75\% relative to baseline comparators while maintaining close to nominal coverage. On two public benchmark data sets, HAPS reduces the median interval length by up to 60\% for predictions at year 5, compared to the baseline comparators.
\end{abstract}

\section{Introduction}
Many real-world decisions involve predicting individual event times for those who are event-free at a prediction time $\tau$ (i.e., with event time $T>\tau$), using covariate histories observed up to $\tau$. 
Modern machine learning methods can flexibly incorporate high-dimensional covariates to produce accurate point predictions of event times. However, in high-stakes settings, point predictions alone are insufficient; rigorous uncertainty quantification is needed for reliable decision-making \citep{banerji2023clinical}.

Conformal prediction \citep{vovk2005algorithmic} is a general and model-agnostic framework for uncertainty quantification that can be applied with off-the-shelf prediction algorithms and has recently been extended to time-to-event outcomes under right censoring.
Existing work has focused on constructing lower prediction bounds (LPBs)
\citep{candes2023conformalized, gui2024conformalized, sesia2025doubly, davidov2025conformalized, farina2025doubly, si2025training}, with more recent efforts constructing mixed-type and two-sided prediction intervals \citep{holmes2024two, qin2025conformal, yi2025survival}. However, these methods are static in the sense that they perform calibration at the time origin using only baseline covariates, focusing on marginal coverage guarantees over the entire population. Existing methods do not address the dynamic task of constructing prediction intervals at a decision time $\tau$ for individuals who have survived to $\tau$, using their covariate history. The dynamic setting is more difficult because of the updated prediction target and the calibration problem. First, the appropriate notion of coverage must be updated to condition on survivors at $\tau$, which we refer to as \emph{survivor-conditional coverage}. Second, right censoring coarsens not only the event time but also the covariate trajectory updated over time.

We propose a conformal framework that constructs two-sided prediction intervals for individual event times among survivors at a prespecified decision time $\tau$, based on time-varying covariate history observed up to $\tau$. Methodologically, we formulate dynamic conformal calibration as estimation of an optimal scalar parameter that indexes a nested family of candidate sets at time $\tau$. 
Under right censoring, we derive an observed-data estimating equation for the calibration parameter using inverse probability of censoring weighting (IPCW) \citep{robins1992recovery, rotnitzky2005inverse}, which maps the censoring-free calibration condition to observed data without requiring correct modeling of the event-time distribution. 
We then integrate these estimating equations with split conformal prediction and establish asymptotic survivor-conditional coverage at time $\tau$. We further develop two doubly robust extensions of HAPS: one based on a doubly robust post-processing procedure and another based on augmented IPCW (AIPCW), to reduce reliance on consistent estimation of the censoring distribution.

Empirically, our simulations show that, in the presence of time-varying confounders for censoring, baseline conformal methods that use only baseline covariates to adjust for censoring can have biased coverage. 
By leveraging time-varying covariates, HAPS maintains close to nominal coverage while producing substantially narrower two-sided prediction intervals than baseline methods.
On two public benchmark data sets from the \texttt{survival} R package 
(details provided in Section \ref{sec:application}), HAPS reduces the median interval length by 9--60\% for predictions at year 5, relative to baseline methods.

\subsection{Related work}

\paragraph{Conformal prediction for right censored data.}

Early developments on conformal prediction for time-to-event outcomes constructed LPBs under \emph{type-I censoring} \citep{candes2023conformalized, gui2024conformalized}. 
Subsequent work has addressed the more common right censoring setting, by censoring imputation approaches that reduce the problem to the type-I censoring setting \citep{sesia2025doubly}, weighting approaches to handle right censoring \citep{davidov2025conformalized}, and doubly robust methods based on semiparametric efficiency theory for right censoring \citep{farina2025doubly, si2025training}. 
Recently, \citet{holmes2024two} proposed mixed-type prediction sets based on censoring status prediction,
and two-sided prediction intervals have been proposed using inverse probability of censoring weighting \citep{yi2025survival} and bootstrap-based methods under working models \citep{qin2025conformal}. 
A detailed summary of the above literature is provided in Appendix \ref{app:literature}.

Despite these advances, existing conformal survival methods share a common limitation in that the target is a static prediction problem defined at the time origin using only baseline covariates. In many settings, baseline covariates only explain a small proportion of variation for time-to-event outcomes, which can be distal, leading to prediction sets that are too wide to be useful in practice \citep{henderson2001accuracy}. 
Figure \ref{fig:toy_example} illustrates this limitation using two benchmark simulation setups from \citet{gui2024conformalized} and \citet{farina2025doubly}. Here, even when oracle quantiles of the conditional event time distribution are used, the resulting 90\% LPBs are close to zero.


\begin{figure}[ht]
    \centering
    \includegraphics[width=0.33\linewidth]{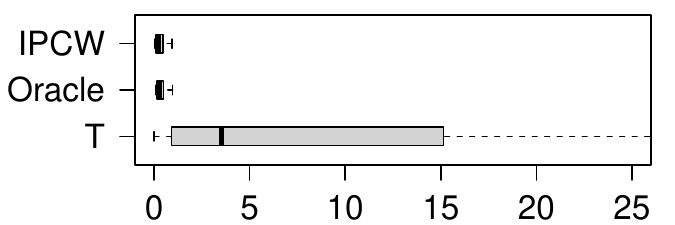}
    \hspace{1em}
    \includegraphics[width=0.33\linewidth]{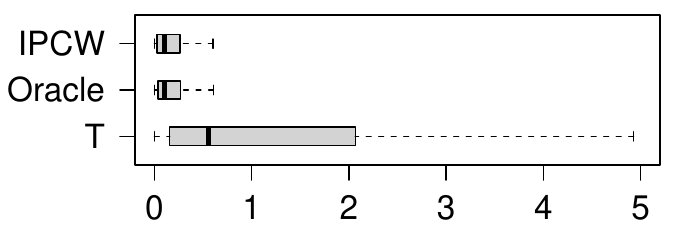}
    \caption{Intrinsic limitation of static approaches to construct LPBs. Boxplots of 90\% LPBs from the IPCW approach in \citet{farina2025doubly} with random survival forests, the oracle 0.1 quantile of the conditional event time distribution given baseline covariate, and the true individual event times $T$, under two benchmark simulation settings: (left) simulation setup 1 in \citet{gui2024conformalized}, and (right) simulation setup 1 in \citet{farina2025doubly}.
    }
    \label{fig:toy_example}
\end{figure}

\paragraph{Dynamic survival prediction with time-varying covariates.}
Independent of conformal prediction, dynamic survival prediction using time-varying covariates (see \citet{rizopoulos2017dynamic} for a review) has been an important area of research in survival analysis. 
Recent advances in dynamic survival prediction algorithms include approaches based on super learner \citep{tanner2021dynamic, rizopoulos2024optimizing}, random survival forests \citep{pickett2021random}, neural networks \citep{lee2019dynamic, lin2022deep, rhodes2023dynamic}, and support vector machines \citep{xie2024support}. However, these methods are designed to provide accurate point predictions and lack rigorous uncertainty quantification, limiting their uptake in practice.

\paragraph{Relationship with adaptive and sequential conformal prediction.}
Updating prediction sets as new information arrives is conceptually related to adaptive or sequential conformal prediction for time series and online learning \citep{gibbs2021adaptive, gibbs2024conformal, batra2023conformal, xu2023sequential, xu2023conformal}. However, our setting differs in that prediction sets evolve because the survivor population at $\tau$ and the available covariate history change over time, not because new outcomes are observed. Accordingly, our goal is to achieve valid coverage among survivors at each prediction time $\tau$, rather than over a stream of outcomes as in time series or online learning settings.


\section{Problem setup and assumptions}\label{sec:preliminary}

Let $T$ denote the event time of interest and $C$ the right censoring time. 
Let $Z_t$ denote the covariates measured at time $t\geq 0$, and let $\bar Z_t = \{Z_s: 0 \leq s \leq t\}$ denote the covariate history up to time $t$. 
In particular, $Z_0$ represents baseline covariates measured at the time origin.
We observe $X = \min(T,C)$ and the event indicator $\Delta = \ind(T < C)$, together with the covariate history observed up to time $X$.
Thus, the observed data for a single individual is $O = (X, \Delta, \bar Z_X)$.
We assume an i.i.d.\ sample $\Dc = \{O_i\}_{i=1}^n$ drawn from the joint distribution of $O$.

\paragraph{Prediction target and coverage criterion.}
For a selected prediction time $\tau \ge 0$, our goal is to construct a prediction set for the event time $T_{n+1}$ of a new test individual, drawn from the same distribution as the training sample, who has survived to time $\tau$, based on the covariate history $\bar Z_{\tau,n+1}$ observed up to that time.
Specifically, given a desired coverage level $1-\alpha \in (0,1)$, we seek a prediction set $\hat\Cc(\bar Z_{\tau,n+1})$ satisfying the survivor-conditional coverage guarantee:
\begin{align}
    \PP\!\left\{ T_{n+1} \in \hat\Cc(\bar Z_{\tau,n+1}) \mid T_{n+1} > \tau \right\} \ge 1-\alpha.
    \label{eq:coverage_exact}
\end{align}
The prediction set $\hat\Cc(\bar Z_{\tau,n+1})$ can be one-sided (where the upper or lower limit is infinity) or two-sided, depending on the downstream application context. 
LPBs provide conservative guarantees on how long an individual is likely to remain event-free, while upper prediction bounds (UPBs) characterize the most optimistic yet plausible event timing; for example, whether the potential benefit of an intensive treatment is large enough to justify its cost or burden. 
In the following, we focus on constructing two-sided prediction intervals, with the understanding that the framework can also be applied to construct one-sided prediction bounds.



\begin{remark}
In contrast to existing conformal survival methods that rely only on baseline covariates $Z_0$, our prediction set $\hat\Cc(\bar Z_{\tau,n+1})$ explicitly incorporates the covariate history up to the prediction time $\tau$.
Allowing $\tau$ to vary allows prediction sets to be updated as new covariate information arrives, resulting in a dynamic procedure that refines uncertainty quantification over time.
\end{remark}

We now state the assumptions under which our procedure is developed.
Let
$\lambda_C(t |~ \cdot ~)
= \lim_{h \to 0^+} \PP(t \le C < t+h \mid \cdot ~,\, C \ge t 
)/h$
denote the conditional hazard function of the censoring time.

\begin{assumption}[Conditional independent censoring given covariate history]\label{ass:cen_ind2}
For all $t \ge 0$, 
$\lambda_C(t | \bar Z_T, T) = \lambda_C(t | \bar Z_t)$.
\end{assumption}

Assumption~\ref{ass:cen_ind2} states that, among individuals who are still at risk (i.e., have not experienced the event or been censored), the instantaneous risk of censoring may depend on past covariate history but not on future covariates or the event time. 
{
This assumption is commonly considered for censoring in the survival literature with time-varying covariates
\citep[e.g.,][]{rotnitzky2005inverse,lok2018estimation,yang2020semiparametric}.
}


\begin{assumption}[Positivity]\label{ass:strict_positivity}
{(i) $\PP(T>\tau)>0$;}
 (ii) there exists $\eta > 0$ such that {$\PP(C>T \mid T, \bar Z_T) > \eta$ almost surely}.
\end{assumption}
{
Assumption~\ref{ass:strict_positivity}(i) requires that survival beyond $\tau$ is possible; in practice, one can choose $\tau$ so that enough individuals remain at risk at $\tau$.
Assumption~\ref{ass:strict_positivity}(ii) requires that, for any realized covariate history, the event time is observed with non-negligible probability. This condition ensures identifiability of the upper tail of the event-time distribution and is standard in time-to-event analyses under right-censoring.
}

\section{Dynamic conformal prediction}

We first develop the dynamic conformal prediction framework in the absence of right censoring to highlight the main calibration idea. We will then show how to extend the formulation to the observed data with right censoring.

\subsection{Under no right censoring}\label{sec:method_noC}

For a fixed prediction time $\tau \ge 0$, let $\bar{\mathcal{Z}}_\tau$ denote the support of $\bar Z_\tau$ conditional on $T>\tau$. Suppose we are given a prediction algorithm
$\Ac$ {for the conditional distribution of $T$ given $(\bar Z_\tau, T>\tau)$.}
We consider a class of nested candidate prediction sets
$\{\Cc_{\tau,\theta}(\bar z_\tau;\Ac) : \theta \in \Theta\},
$
indexed by a scalar parameter $\theta \in \Theta \subseteq \mathbb{R}$ that is increasing in $\theta$, 
where the index space $\Theta$ controls how rich the candidate class is. We will discuss below how $\Theta$ is chosen with an illustrative example on the candidate class.

Our goal is to obtain the prediction set with the smallest size in this class that achieves the desired survivor-conditional coverage. Formally, let $\theta^*$ be the smallest $\theta\in\Theta$ such that 
\begin{align}
    \PP\left\{\left. T\in \Cc_{\tau,\theta}(\bar Z_\tau;\Ac) ~\right|~ T>\tau \right\} \geq 1-\alpha.  \label{eq:coverage_Cc_theta}
\end{align}
Then the corresponding prediction set $\Cc_{\tau,\theta^*}$ is the smallest set in the candidate class that satisfies the coverage requirement among survivors at time $\tau$. Existence of $\theta^*$ follows from the monotonicity of the left-hand side of \eqref{eq:coverage_Cc_theta} in $\theta$, provided the class contains sufficiently large sets. Note that $\theta^*$ implicitly depends on the prediction algorithm $\Ac$, the form of the candidate class $\{\Cc_{\tau,\theta}\}_{\theta\in\Theta}$, and the conditional distribution of $(T,\bar Z_{\tau})\mid (T>\tau)$.

\paragraph{Construction of the candidate class.}


For time-to-event prediction, estimates for the conditional quantile of $T$ can be obtained using off-the-shelf software, so a natural way to construct the candidate class 
$\{\Cc_{\tau,\theta}\}_{\theta\in\Theta}$
is through quantile-based intervals. In particular, when placing equal emphasis on lower- and upper-tail miscoverage, one may consider 
\begin{align}
     \Cc_{\tau,\theta}(\bar z_\tau; \Ac) = (\hat{q}_{1-\theta}(\bar z_\tau), \hat{q}_{\theta}(\bar z_\tau)), \label{eq:Cc_condi_quantile}
\end{align}
where $\theta\in[0.5, 1]$, and $\hat{q}_{\theta}(\bar z_\tau)$ denotes the estimated $\theta$-quantile of $T\mid (\bar Z_\tau = \bar z_\tau,T>\tau)$ according to $\Ac$. This candidate class is used in our simulations and data applications.
In cases where lower- and upper-tail miscoverage are not equally undesirable, asymmetric quantile pairs can be used instead. Additionally, the choice of quantile levels determines the location of the nested candidate intervals and thus affects the length of the final prediction interval. To further improve efficiency, the position of the candidate intervals can be selected in a data-adaptive manner.
More generally, one can form the candidate class based on conformity scores, with details provided in Appendix \ref{app:conformity_score}.

\paragraph{Reformulation of the calibration condition.}

Note that the calibration parameter $\theta^*$ is unknown and needs to be estimated from data. To motivate our estimator and its extension in right censored settings, we rewrite the coverage condition \eqref{eq:coverage_Cc_theta} in terms of an estimating inequality. 
Note that \eqref{eq:coverage_Cc_theta} can be equivalently written as: 
\begin{align}
    \E\left[\left. \ind\{T\in \Cc_{\tau,\theta}(\bar Z_\tau; \Ac)\} \right| T>\tau \right] \geq 1-\alpha, \label{eq:EE_1}
\end{align}
which can be rewritten in unconditional form as $\E\{U^*(\theta,\Ac)\} \geq 0$, where
\begin{align}
    U^*(\theta;\Ac) 
    &=  \ind(T>\tau) \left[\ind\{T\in \Cc_{\tau,\theta}(\bar Z_\tau; \Ac)\} - (1-\alpha) \right]. \label{eq:D(theta)}
\end{align}
Therefore, $\theta^*$ can be characterized as
\begin{align}
    \theta^* = \inf \{\theta\in\Theta: \E\{U^*(\theta;\Ac)\} \geq 0\}. \label{eq:EE_ineq_D}
\end{align}
When equality holds at $\theta^*$, the function $U^*(\theta;\Ac)$ serves as a nonparametric estimating function for $\theta^*$ in the censoring-free data.



\subsection{Under right censoring}\label{sec:method_C}

We now turn to the setting under right censoring. In this case, neither the event time $T$ nor the covariate history $\bar Z_\tau$ is always fully observed, so the censoring-free calibration procedure developed in Section~\ref{sec:method_noC} no longer applies directly. Nevertheless, the key insight from the previous subsection still applies: if we can learn the calibration parameter $\theta^*$, then a valid prediction set can be constructed as $\hat\Cc = \Cc_{\tau,\hat\theta}$. We first consider mapping the censoring-free data estimating function of $\theta^*$ to the observed data using IPCW and then consider two doubly robust variants.

\paragraph{Identification of $\theta^*$.}

Let 
$G(t|\bar Z_t) 
= \PP(C > t \mid \bar Z_t)$
denote the conditional survival probability of censoring given covariate history. Under assumption \ref{ass:cen_ind2}, $G(t|\bar Z_t) = \exp\{ - \int_0^t \lambda_C(u |\bar Z_u)\,du \}$.
Applying IPCW \citep{robins1992recovery, rotnitzky2005inverse} to the estimating function of $\theta^*$ in \eqref{eq:D(theta)} results in the following observed data estimating functions: 
\begin{align}
    U(\theta; G,\Ac) 
    & =  \frac{\Delta \ind(X >\tau)}{G(X|\bar Z_X)} \left[\ind\{X \in \Cc_{\tau,\theta}(\bar Z_\tau; \Ac)\} - (1-\alpha) \right], \label{eq:IPCW}
\end{align}
Under Assumption \ref{ass:cen_ind2}, $G(T|\bar Z_T) = \PP(C>T \mid T,\bar Z_T)$, and the IPCW factor $\Delta / G(X|\bar Z_X) = \ind({T<C}) / G(T|\bar Z_T) = \ind(T<C)/\PP(T<C\mid T,\bar Z_T)$ reweights uncensored individuals to represent the hypothetical population under no right censoring, thereby correcting the bias induced by censoring. 

\begin{lemma}[]\label{lem:identification_IPCW}
    Under Assumptions \ref{ass:cen_ind2} and \ref{ass:strict_positivity}, $\E\left\{ U(\theta;G,\Ac) \right\} = \E\{U^*(\theta;\Ac)\}$
    for all $\theta\in\Theta$. Therefore, $\theta^* = \inf \{\theta\in\Theta: \E\{ U(\theta;G,\Ac) \} \geq 0\}$.
\end{lemma}

Lemma~\ref{lem:identification_IPCW} shows that the IPCW estimating function preserves the calibration condition from the censoring-free setting. In particular, $\theta^*$ can be identified as the minimum value of $\theta$ for which the weighted estimating equation is nonnegative in expectation. A proof is provided in Appendix~\ref{app:proof_identification}.

\paragraph{Estimation of $\theta^*$ and the prediction set.}

We next integrate the calibration condition of Lemma \ref{lem:identification_IPCW} to construct prediction sets that target the survivor-conditional coverage guarantee in \eqref{eq:coverage_Cc_theta}. 
For illustration, we focus on split conformal prediction, but the same framework can be applied to other conformal prediction algorithms (see \citet{angelopoulos2024theoretical} for a comprehensive review).
The resulting dynamic split conformal prediction procedure is summarized in Algorithm \ref{alg:split_dynamicCP}.


\begin{algorithm}[tbh!]
\caption{History-Aware Prediction Sets (HAPS)}
\label{alg:split_dynamicCP}
\begin{algorithmic}
\STATE \textbf{Input:} observed data $\Dc = \{O_i\}_{i=1}^n$; coverage level $1-\alpha$; prediction time $\tau$; prediction model $\Ac$ for $T\mid (\bar Z_\tau, T>\tau)$; a rule of constructing the candidate class after fitting the prediction model; censoring model for $G$; covariate history $\bar Z_{\tau,n+1}$ from a new test individual.
\begin{enumerate}
    \item Randomly split $\Dc$ into a training set $\Dctr$ and a calibration set $\Dccal$ with almost equal size.
    \item Fit the prediction model $\hat\Ac$ on $\Dctr$, and construct nested candidate sets $\{\Cc_{\tau,\theta}\}_{\theta \in \Theta}$.
    \item On $\Dctr$, estimate the nuisance parameters $G$, and denote the estimate as $\hat G$. 
    \item On $\Dccal$, compute 
    $\hat\theta
    = \inf\left\{
    \theta\in\Theta: 
    \sum_{i \in \Dccal}
    U_{i}(\theta;\hat{G},\hat{\Ac}) \ge 0
    \right\}.$
\end{enumerate}
\STATE \textbf{Output:} prediction set $\hat\Cc(\bar Z_{\tau,n+1}) = \Cc_{\tau,\hat\theta}(\bar Z_{\tau,n+1} ; \hat\Ac)$.
\end{algorithmic}
\end{algorithm}

In practice, $\hat\theta$ in Algorithm \ref{alg:split_dynamicCP} can be computed using a grid search on $\Theta$ restricted to a discrete and finite set. For example, with the quantile-based candidate class in \eqref{eq:Cc_condi_quantile}, we take $\Theta$ to be a uniform grid on $[0.5, 1]$ with increment $0.01$ in both the simulation and data applications. 
The nuisance parameter $G$ can be estimated using standard survival modeling techniques that accommodate dependent right censoring with time-varying covariates, with detailed discussion provided in Appendix \ref{app:G_est}.



\subsection{Survivor-conditional coverage }\label{sec:theory}

We now present the theoretical guarantee for HAPS. 
In particular, we show that the constructed prediction sets achieve survivor-conditional coverage at a given prediction time $\tau$, up to terms that vanish as the calibration sample size increases and the estimator for the censoring distribution is consistent.

We consider the following $L_2$-norm of the censoring distribution estimation error: 
$\|\hat G - G\|_2 = \E_{X,\bar Z_X}\{|\hat G(X\mid \bar Z_X) - G(X\mid \bar Z_X)|^2\}^{1/2}$, where the expectation is taken with respect to $(X,\bar Z_X)$ conditional on the data used to estimate $\hat G$.
Let $\sttau = \PP(T>\tau)$ denote the marginal survival probability at time $\tau$, and $|\Dccal|$ denote the calibration sample size. 


\begin{theorem}\label{thm:IPCW_coverage}
    Under Assumptions~\ref{ass:cen_ind2} and ~\ref{ass:strict_positivity}, and assume that $\hat G(X_{n+1}|\bar Z_{X_{n+1}, n+1}) > \eta$ almost surely, and that the class of candidate sets contains large enough sets so that Algorithm \ref{alg:split_dynamicCP} admits a feasible solution, then for any $\epsilon \in(0,1)$, there exists a constant $\consta>0$ such that, with probability at least $1-\epsilon$ over $\Dc$, 
    \begin{align}
        & \PP(T_{n+1} \in \hat\Cc(\bar Z_{\tau, n+1})\mid T_{n+1}>\tau, \Dc) \nonumber \\
        &\qquad \geq 1-\alpha - \sttau^{-1}\eta^{-2} \|\hat G - G\|_2  
        - \sttau^{-1}\eta^{-1} \left( \sqrt{\frac{1}{2}\log \frac{1}{\epsilon}} + \consta \right) \frac{1}{\sqrt{|\Dccal|}}. \label{eq:coverage_lower_bound}
    \end{align}
\end{theorem}
Theorem~\ref{thm:IPCW_coverage} shows that {HAPS achieves \emph{probably asymptotically approximately correct (PAAC)} coverage among survivors at time $\tau$ when $G$ is consistently estimated. We adopt the term ``PAAC" from \citet{qiu2023prediction}, which is referred to as ``asymptotic PAC" in \citet{farina2025doubly}.}
The gap between the lower bound in \eqref{eq:coverage_lower_bound} and the nominal level $1-\alpha$ decomposes into two components: (i) an estimation error term due to learning the censoring distribution $G$, and (ii) a calibration error term of order $|\Dccal|^{-1/2}$ from split conformal calibration.
Consequently, when $\hat G$ is consistent and the calibration sample size grows, the coverage lower bound converges to the nominal level $1-\alpha$. 
The proof of Theorem~\ref{thm:IPCW_coverage} is provided in Appendix~\ref{app:proof_coverage}. 

When a (semi)-parametric model is used to estimate $G$, consistency of $\hat G$ can be achieved under correct model specification. In practice, one may also use flexible nonparametric or machine learning methods to estimate $G$, which can reduce reliance on parametric model assumptions. To further reduce sensitivity to potential misspecification of the censoring model, we next present two doubly robust extensions of HAPS.


\section{Doubly robust extensions of HAPS}
 
\paragraph{Doubly robust post-processing.}
Recall that $\hat{q}_{\theta}(\bar z_\tau)$ denotes the estimated $\theta$-quantile of $T\mid (\bar Z_\tau = \bar z_\tau,T>\tau)$ under the prediction algorithm $\Ac$. 
Given the prediction interval $\hat\Cc(\bar Z_{\tau,n+1})$ from Algorithm \ref{alg:split_dynamicCP}, we consider the post-processing set 
\[
\hat\Cc_{\DR}(\bar Z_{\tau,n+1}) = \hat\Cc(\bar Z_{\tau,n+1})\cup \left(\hat q_{\alpha/2}(\bar Z_{\tau,n+1}),~ \hat q_{1-\alpha/2}(\bar Z_{\tau,n+1}) \right).
\]
We call this extension HAPS-DR.
The prediction set $\hat \Cc_{\DR}$ is doubly robust in the sense that it achieves PAAC coverage among survivors when either $\hat G$ is consistent and achieves asymptotic coverage among survivors when the quantile estimators $\hat q_{\alpha/2}$ and $\hat q_{1-\alpha/2}$ are consistent. 
When $\hat G$ is consistent, PAAC coverage follows from Theorem~\ref{thm:IPCW_coverage}, since $\hat \Cc_{\DR}$ contains the original HAPS prediction set $\hat\Cc$. On the other hand, when $\hat q_{\alpha/2}$ and $\hat q_{1-\alpha/2}$ consistently estimate the corresponding conditional quantiles of 
$T \mid (\bar Z_\tau, T>\tau)$, the quantile interval
$\left(\hat q_{\alpha/2}(\bar Z_{\tau,n+1}),~ \hat q_{1-\alpha/2}(\bar Z_{\tau,n+1}) \right)$ 
achieves asymptotic $1-\alpha$ coverage among survivors, and hence so does its union with $\hat\Cc$.

\paragraph{AIPCW.}
Unlike HAPS-DR, this extension applies the general AIPCW approach to derive a doubly robust estimating function for $\theta^*$, which leads to prediction intervals with doubly robust coverage. AIPCW \citep{robins1992recovery,rotnitzky2005inverse} is a general approach for doubly robust censoring adjustment by augmenting IPCW estimating equations with outcome modeling.

Before presenting the AIPCW estimating function,
we first introduce the counting process and martingale notation for $C$.
Let $N_C(t) = \ind(X \leq t, \Delta = 0)$, and 
$M_C(t) = N_C(t) - \int_0^t \ind(X \geq u) \lambda_C(u|\bar Z_u) du$.
The AIPCW estimating function for $\theta^*$ is: 
\begin{align}
    U_{\AIPCW}(\theta; G, h_{\tau}, \Ac) 
    & = U(\theta; G,\Ac) + \int_0^{X} h_\ttau(u,\bar Z_u;\theta,\Ac) \frac{dM_C(u)}{G(u|\bar Z_u)}, \label{eq:U_AIPCW}
\end{align}
where 
$h_\ttau(u,\bar Z_u; \theta,\Ac)
= \E\left(\left. \ind(T>\ttau) \left[\ind\{T\in \Cc_{\ttau,\theta}(\bar Z_\ttau;\Ac)\} - (1-\alpha) \right] \right| \bar Z_u, T\geq u\right)$.

Compared to \eqref{eq:IPCW}, $U_{\AIPCW}$ involves an additional nuisance function 
$h_{\tau}$, which is a conditional expectation involving the joint distribution of event time and covariate history. Estimating this conditional expectation can be challenging when covariate histories are observed in continuous time.
In Appendix~\ref{app:h_est}, we discuss these challenges and provide nonparametric identification and estimation of $h_{\tau}$. 
We construct the prediction set by augmenting Step~3 of 
Algorithm~\ref{alg:split_dynamicCP} with an estimator of $h_\tau$, and replacing 
$U$ by $U_{\AIPCW}$ in Step~4. 
We call this extension HAPS-A.

HAPS-A is
is rate doubly robust due to the mixed-bias property \citep{rotnitzky2021characterization, luo2025doubly}: 
the effect of nuisance estimation on the coverage error enters only through the product of the estimation errors for $G$ and $h_\tau$. 
Thus, $G$ and $h_\tau$ may be estimated using flexible machine learning methods, which often converge more slowly than $n^{-1/2}$ individually, while still allowing the coverage error to vanish at a $n^{-1/2}$ rate, provided standard regularity conditions hold and that the product error rate of $G$ and $h_{\tau}$ is faster than $n^{-1/2}$.

We now formally state the coverage guarantee for HAPS-A. 
Recall that $\lambda_C(t\mid\bar Z_t)$ denotes the conditional hazard function of $C$ given covariate history, and let $\Lambda_C(t\mid\bar Z_t) = \int_0^t\lambda_C(s\mid\bar Z_s)\,ds$ denote the conditional cumulative hazard function of $C$.
To simplify the presentation and build intuition about the product errors, we assume that $\hat G$ admits a nonnegative
censoring hazard function $\hat\lambda_C$, with corresponding
cumulative hazard
$\hat\Lambda_C(t\mid\bar Z_t)
= \int_0^t\hat\lambda_C(u\mid\bar Z_u)\,du
= -\log\hat G(t\mid\bar Z_t)$. 
The coverage result extends to more general settings using integral products between the estimation errors of $h_{\tau}$ and $\Lambda_C$. Integral products are encountered in rate double robustness for time-to-event data where both sets of nuisance parameters involve functions of time \citep{ying2023cautionary, wang2024doubly, westling2024inference, luo2025doubly}.

\begin{theorem}[HAPS-A coverage]\label{thm:AIPCW_coverage}
Under Assumptions~\ref{ass:cen_ind2} and~\ref{ass:strict_positivity}
hold, and assume that $\hat G(X_{n+1}\mid\bar Z_{X_{n+1},n+1})>\eta$ almost surely,
and that $\hat h_\tau$ takes values in $[-(1-\alpha),\alpha]$.
In addition, assume that $\Theta$ is finite and that, almost surely, there exists
$\theta\in\Theta$ such that
$
\sum_{i\in\Iccal}
U_{\AIPCW,i}(\theta;\hat G,\hat h_\tau,\hat\Ac)\geq 0.
$
Then for any $\epsilon\in(0,1)$, there exists a constant $\consta>0$
such that, with probability at least $1-\epsilon$ over $\Dc$,
\begin{align*}
    & \PP\{T_{n+1}\in\hat\Cc_{\AIPCW}(\bar Z_{\tau,n+1})\mid T_{n+1}>\tau,\Dc\} \\
    &\qquad \geq 1-\alpha - \sttau^{-1}\eta^{-1}\sup_{\theta\in\Theta}\left[\min\{r_1(\theta),\, r_2(\theta)+(1-\eta)\,r_3(\theta)\}\right] \\
    &\qquad\quad - \sttau^{-1}c_\eta\left(\sqrt{\frac{1}{2}\log\frac{1}{\epsilon}}+\consta\right)\frac{1}{\sqrt{|\Dccal|}},
\end{align*}
where $c_\eta = 2\max(\alpha,1-\alpha)(2\eta^{-1}-1)$, and $r_1$, $r_2$, $r_3$, defined in \eqref{eq:r1}--\eqref{eq:r3} in Appendix~\ref{app:proof_AIPCW}, are integrals of products between the estimation errors for $h_\tau$ and conditional censoring distribution.
\end{theorem}
Theorem~\ref{thm:AIPCW_coverage} shows that the coverage gap of HAPS-A is bounded by two terms: the first term involves product errors of the nuisance estimators ($r_1$, $r_2$, $r_3$), and the second term is of order $O(|\Dccal|^{-1/2})$, which goes to zero as the calibration sample size gets larger. The proof is provided in Appendix~\ref{app:proof_AIPCW}.

\section{Numerical experiments}\label{sec:simulation}

We consider the data generating mechanisms (DGMs) described in Appendix \ref{app:simu_DGM}, which vary in their event-time and censoring mechanisms and have censoring rates between 32\%--54\%.
We construct two-sided prediction intervals for survivors at $\tau \in\{0,3,6\}$ using HAPS and its two extensions (HAPS-DR and HAPS-A). 
As comparators, we consider a baseline IPCW method that uses only baseline covariates for both prediction and censoring adjustment (details in Appendix \ref{app:IPCW0}), denoted by `IPCW'. 
In addition, we also include two natural adaptations of this baseline method to the dynamic prediction setting: `trunc', which takes the intersection of the IPCW interval with $[\tau,\infty)$; and `LM', which applies the baseline IPCW method to the landmark subgroup at time $\tau$, i.e., individuals with $X>\tau$. Both baseline adaptations produce prediction intervals contained in $[\tau,\infty)$, the support of event time among survivors at $\tau$.

We consider fitting prediction models for HAPS and its two extensions (HAPS-DR and HAPS-A) using a neural-network dynamic survival model adapted from 
Dynamic-DeepHit  \citep{lee2019dynamic}, denoted by `DDH',  which leverages the covariate history up to the prediction time; landmarking-based Cox proportional hazards models (`Cox'); or random survival forests (`RSF'). For nuisance parameter estimation, we consider both semiparametric models and machine learning methods. Details of the prediction and nuisance models are provided in Appendix \ref{app:simu_prediction_model}. 

We evaluate the performance of the prediction intervals over 200 Monte Carlo replications. 
In each replication, we generate an observed sample of size 1000 to construct the prediction 
intervals, and an independent evaluation sample of size 500 without right censoring, which 
allows direct assessment of coverage.

\paragraph{Results.}
We report the empirical coverage, interval length, and lower and upper endpoints of the two-sided prediction intervals for each method over the 200 Monte Carlo replications. 
Figure \ref{fig:main_linWB1_micC2DGM1_rho03} summarizes the empirical coverage and interval length of the 90\% prediction intervals under two simulation setups, with boxplots computed among survivors at each $\tau$. 
Setup A generates data from DGM1, where the censoring distribution follows a proportional hazards model depending on time-varying covariates, and uses a Cox prediction model together with a correctly specified proportional hazards model for estimating the censoring distribution. 
Setup B generates data from DGM2, where the censoring distribution is a mixture of uniform and Weibull distributions depending on time-varying covariates, and uses a DDH prediction model together with machine learning methods for nuisance function estimation. 
Details of the simulation setups and implementation details, as well as additional summaries for the lower and upper endpoints of the two-sided intervals are provided in  Appendix~\ref{app:simu}.

We observe that HAPS and HAPS-A achieve close to 90\% coverage in both setups for survivors at all three $\tau$, 
while HAPS-DR tends to overcover. The prediction intervals for all three methods become tighter as $\tau$ increases, reflecting the efficiency gain from using more covariate information for prediction.
We also observe that HAPS and HAPS-A have similar median interval length, while HAPS-DR is typically wider, 
reflecting its conservativeness due to the doubly robust post-processing step.

On the other hand, baseline methods (`IPCW', `trunc', `LM') can show large coverage gaps. In both simulation setups, `IPCW' and `trunc' overcover at $\tau = 3$ and $\tau = 6$, while `LM' undercovers in Setup B at all three prediction times. Even at $\tau = 0$, where the coverage target is the full population, the baseline methods exhibit coverage gaps in both settings because adjustment for right censoring using only baseline covariates is insufficient.
Table \ref{tab:setupAB_length_ratio} summarizes the ratios of median prediction interval length for each version of HAPS relative to the baseline methods. Relative to the baseline methods, HAPS and its extensions reduce median interval length by \textbf{5\%--23\%} at $\tau = 3$ and \textbf{34--75\%} at $\tau = 6$. 
To disentangle the two roles of covariate history in the empirical improvements, Appendix~\ref{app:add-baselines} considers additional baselines. The results suggest that using covariate history for handling censoring in calibration ensures close to nominal coverage, and incorporating covariate history in the prediction model leads to shorter intervals. Appendix~\ref{app:small-n} reports additional simulation results with sample sizes that mimic those in the two applications.

\begin{figure}[ht!]
    \centering

    \begin{subfigure}{0.49\linewidth}
        \centering
        \includegraphics[width=\linewidth]{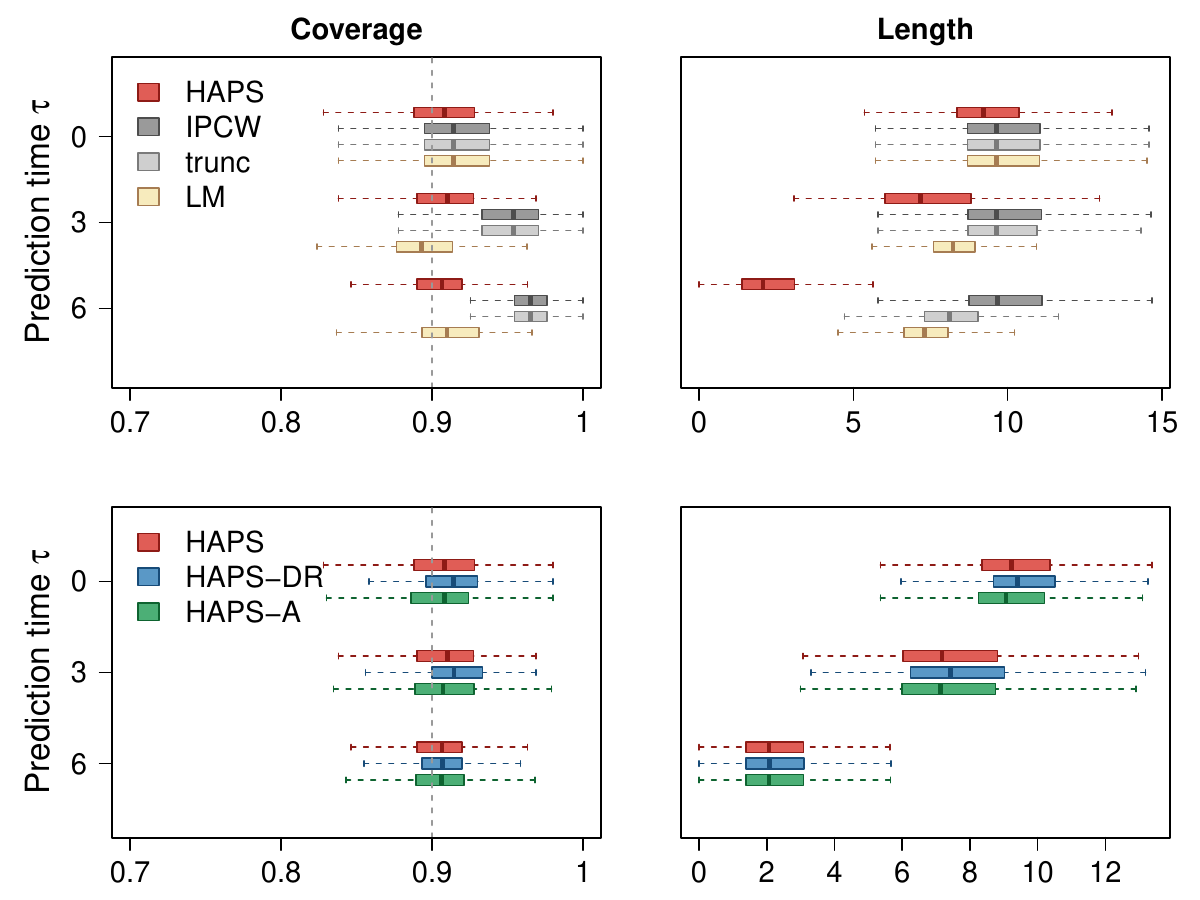}
        \caption{Setup A: DGM1 with Cox prediction model.}
        \label{fig:setup1}
    \end{subfigure}
    \hfill
    \begin{subfigure}{0.49\linewidth}
        \centering
        \includegraphics[width=\linewidth]{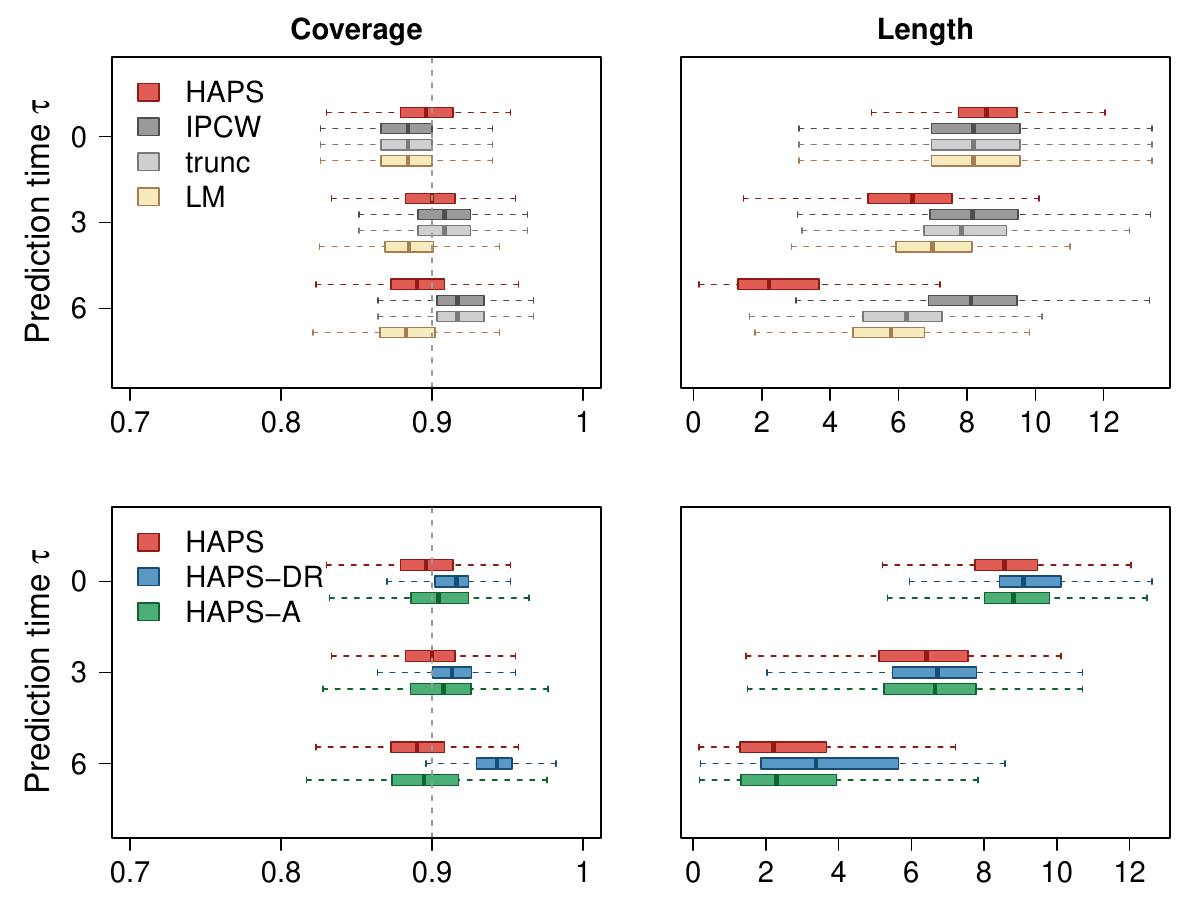}
        \caption{Setup B: DGM2 with DDH prediction model.}
        \label{fig:icml_simu1_dcp_rho03}
    \end{subfigure}


    \caption{Summary of 90\% prediction intervals across 200 Monte Carlo replications comparing HAPS with the  baseline methods (top) and extensions of HAPS (bottom).
    }
    \label{fig:main_linWB1_micC2DGM1_rho03}
\end{figure}


\begin{table}[!ht]
\centering
\caption{Ratio of median interval length for HAPS and its extensions relative to baseline methods.}
\label{tab:setupAB_length_ratio}
{\fontsize{8}{9}\selectfont
\begin{tabular}{llccccccccc}
\toprule
&  & \multicolumn{3}{c}{$\tau=3$} & \multicolumn{3}{c}{$\tau=6$} \\
\cmidrule(lr){3-5} \cmidrule(lr){6-8} 
Setup & Method & IPCW & trunc & LM & IPCW & trunc & LM \\
\midrule
A & HAPS & 0.79 & 0.79 & 0.91 & 0.25 & 0.29 & 0.32 \\
DGM1:Cox & HAPS-DR & 0.81 & 0.81 & 0.95 & 0.25 & 0.29 & 0.33 \\
& HAPS-A & \textbf{0.77} & \textbf{0.77} & \textbf{0.90} & 0.25 & 0.29 & 0.32 \\

\midrule
B & HAPS & \textbf{0.78} & \textbf{0.81} & \textbf{0.91} & \textbf{0.32} & \textbf{0.43} & \textbf{0.46} \\
DGM2:DDH & HAPS-DR & 0.82 & 0.84 & 0.95 & 0.46 & 0.61 & 0.66 \\
& HAPS-A & 0.80 & 0.83 & 0.93 & 0.33 & 0.44 & 0.48 \\
\bottomrule
\end{tabular}
}
\end{table}

Compared with HAPS, the advantage of HAPS-A and HAPS-DR is robustness to censoring estimation errors. 
Appendix \ref{app:HAPS_robustness} presents additional simulations to assess this robustness, including settings with misspecified censoring models and settings with flexible machine learning methods to estimate the censoring distribution.
Appendix \ref{app:simu_impact_pred_accuracy} further studies the impact of prediction model accuracy on prediction intervals.
We observe that HAPS can undercover when the censoring 
model is misspecified. HAPS-DR is more robust to censoring model misspecification, although at the cost of more conservative prediction intervals, especially in settings with small sample sizes or poorly trained prediction models. 
When machine learning methods are used for censoring estimation, HAPS shows robust coverage under different censoring mechanisms, especially with larger sample sizes.
In practice, HAPS is often preferred for its simplicity, whereas HAPS-DR is preferred when robustness is the primary concern. Although HAPS-A is also doubly robust, it requires careful estimation of additional nuisance parameters for the augmentation term, which can be computationally demanding (Appendix~\ref{app:HAPSA_cost}).

\section{Applications}\label{sec:application}

We apply the proposed methods to two publicly available benchmark datasets from the \texttt{survival} R package: the colon cancer dataset `\texttt{colon}' and the primary biliary cholangitis (PBC) dataset `\texttt{pbcseq}'. For both datasets, we focus on predicting overall survival time. The colon cancer dataset contains 929 patients and has a censoring rate of approximately 51.3\%. 
The PBC dataset contains 312 patients with primary biliary cholangitis and has a censoring rate of approximately 55.1\%.
Details of the datasets and the covariates included in the analysis are summarized in Appendix~\ref{app:dataset_details}.

\paragraph{Results.}

We construct two-sided prediction intervals using the DDH prediction model and a proportional hazards model for censoring. 
For evaluation, we repeat 200 random splits of each dataset into an 80\% sample for interval construction  and a 20\% test sample. 
Unlike in the simulations, empirical survivor-conditional coverage cannot be computed directly on the test sample because event times are not fully observed under right censoring. 
We therefore report IPCW estimates of survivor-conditional coverage, 
with details provided in Appendix~\ref{app:justification_IPCW_coverage}. 
These estimates are not the true coverage or the empirical coverage computed with the true event times; they approximate the true coverage only when the censoring weights are close to the truth. Appendix~\ref{app:app-eval-sens} analyzes the sensitivity of the coverage estimates to the censoring models used for evaluation.

Figure~\ref{fig:icmlapp_colon_alpha0.2} summarizes the 80\% prediction intervals among survivors at $\tau=0,3,5$ years for both data sets, with additional 
results for other prediction times provided in Appendix~\ref{app:application}. 
Across the two data applications, the prediction intervals from HAPS and HAPS-DR are similar and attain IPCW-estimated coverage close to 80\%.
Their median interval lengths decrease as $\tau$ increases, reflecting the additional information provided by covariate histories observed up to later prediction times.
Table~\ref{tab:app_length_ratio} in the appendix summarizes the ratios of median prediction interval lengths of HAPS and HAPS-DR relative to the baseline methods (`IPCW', `trunc',  `LM'). 
Compared to the baseline methods, 
HAPS and HAPS-DR reduce median interval length by \textbf{9\%--22\%} in the PBC data and \textbf{13\%--60\%} in the colon cancer data for predictions at year 5.

\begin{figure}[!ht]
    \centering

    \begin{subfigure}{0.49\linewidth}
        \centering
        \includegraphics[width=1\linewidth]{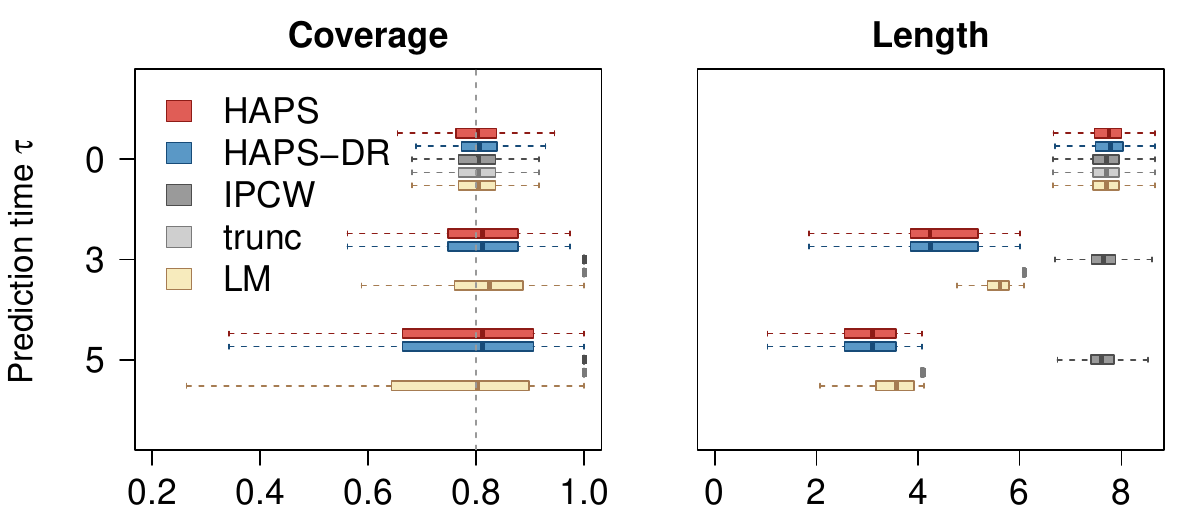}
        \caption{Colon cancer data}
        \label{fig:app_colon_0_3_5}
    \end{subfigure}
    \hfill
    \begin{subfigure}{0.49\linewidth}
        \centering
        \includegraphics[width=1\linewidth]{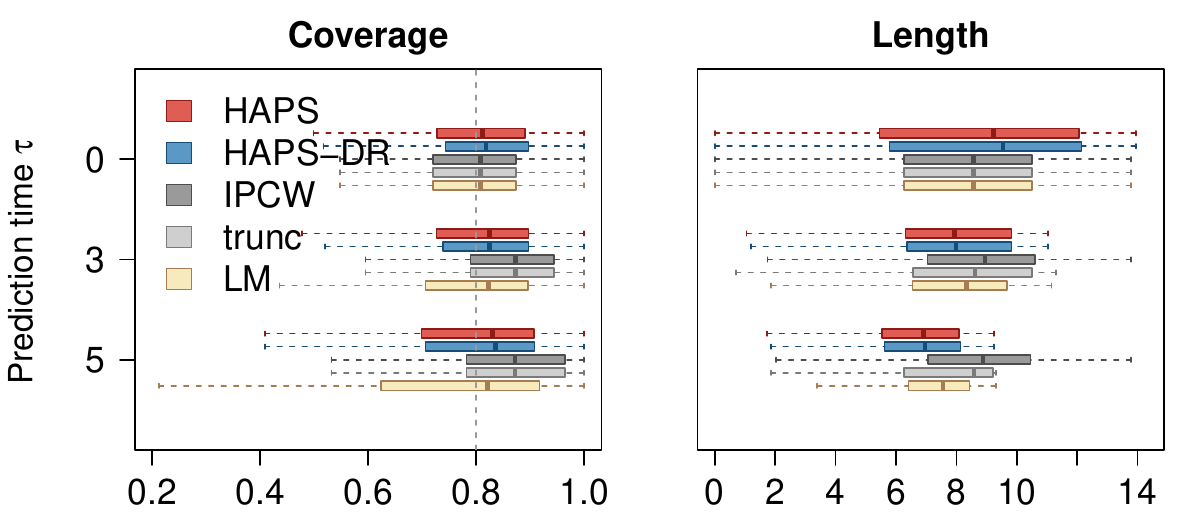}
        \caption{PBC data}
        \label{fig:app_colon_pbc_0_3_5}
    \end{subfigure}

    \caption{Summary of 80\% prediction intervals across 200 random splits for data applications; coverage is IPCW-estimated (Appendix~\ref{app:justification_IPCW_coverage}).}
    \label{fig:icmlapp_colon_alpha0.2}
\end{figure}

Computational resources used for the simulation and data-application experiments are summarized in Appendix~\ref{app:compute}.

\section{Discussion}

We develop History-Aware Prediction Sets (HAPS), a conformal framework for constructing prediction intervals for individual event times using covariate histories observed up to a prediction time $\tau$. Across simulations and data applications, HAPS produces substantially shorter prediction intervals than baseline methods while maintaining close-to-nominal coverage, showing its potential to aid decision-making in settings where updated covariate information is collected over time. The code for implementing the method and reproducing the simulation results is available at: \url{https://github.com/wangyuyao98/HAPS_neurips2026}.

An interesting future direction is to extend HAPS to competing risks and semi-competing risks settings, leveraging advances in modern dynamic prediction algorithms such as Dynamic-DeepHit \citep{lee2019dynamic}.
Another promising direction is to extend HAPS to left-truncated and right-censored data, which arise in studies with delayed entry, by leveraging inverse probability of truncation weighting or doubly robust methods to account for covariate-dependent left truncation \citep{vakulenko2022nonparametric, wang2024doubly, wang2024liberating, han2024truncated}. A third extension is to multi-source and federated settings. Building on multi-source conformal inference under distribution shift \citep{liu2024multisource}, adaptive borrowing across sites could improve calibration efficiency while targeting coverage among survivors in a specified target population. Related advances in federated causal survival analysis \citep{liu2026privacy} motivate implementations that avoid sharing individual-level records. 

\paragraph{Limitations.}
A limitation of HAPS and its extensions is that their validity relies on conditionally independent censoring given the observed time-varying covariate history, as well as positivity. 
Violations of these assumptions may lead to prediction intervals with invalid coverage. 
With extreme weights, and especially for small $\alpha$, the algorithm may return trivial intervals  (Appendix~\ref{app:app-alpha}).
While these conditions are weaker than those imposed by existing conformal methods for right-censored time-to-event outcomes, practitioners should still conduct diagnostics in practice, especially for high-stakes applications. Sensitivity analyses can help to evaluate the impact of potential violations.

An additional limitation of HAPS is its sensitivity to errors in estimating the censoring distribution; misspecification of the censoring model can lead to coverage gaps. To address this limitation, we propose two doubly robust extensions, HAPS-DR and HAPS-A, which improve robustness but introduce additional trade-offs. HAPS-DR can be conservative due to its post-processing step, whereas HAPS-A requires estimating additional nuisance parameters and may be computationally challenging or unstable in small samples or under positivity violations. In practice, HAPS-DR may be preferred when robustness is prioritized, while HAPS may be preferable when simplicity and efficiency are primary concerns.


\paragraph{Broader impacts.}
HAPS provides a model-agnostic framework for uncertainty quantification that can be wrapped around existing prediction algorithms, taking a step toward more reliable use of modern dynamic survival algorithms in high-stakes settings such as clinical care. As shown in our simulations and applications, HAPS and its extensions produce substantially shorter prediction intervals than baseline methods, improving the practical utility of conformal prediction. Potential negative impacts may arise if practitioners over-rely on the resulting intervals without assessing the plausibility of the required assumptions or the sensitivity of conclusions to assumption violations. 

\begin{ack}
We thank the anonymous reviewers and the area chair for their constructive comments.
This research was supported in part by a Patient-Centered Outcomes Research Institute (PCORI) Award (ME-2025C1-44648-IC), the National Institute of Allergy and Infectious Diseases of the National Institutes of Health (K01AI193070; L.H.), and the U.S. National Science Foundation (SES-2242776; S.Y.).
This research also uses services provided by the Open Science Grid (OSG) Consortium \citep{osg07, osg09, https://doi.org/10.21231/906p-4d78, https://doi.org/10.21231/0kvz-ve57}, which is supported by National Science Foundation awards \#2030508 and \#1836650.
The content, including any opinions, findings, conclusions, or recommendations, is solely the responsibility of the authors and does not necessarily represent the views of PCORI, its Board of Governors or Methodology Committee, or the official views of the National Institutes of Health or the U.S. National Science Foundation. The authors declare no competing interests.
\end{ack}



\bibliographystyle{plainnat}
\bibliography{bib/conformal_surrogates, bib/osg}








\newpage
\appendix


\section{Summary of the conformal prediction literature for right censored time-to-event outcomes}\label{app:literature}

Early developments on conformal prediction for time-to-event outcomes considered \emph{type-I censoring}, where each individual’s potential censoring time is always observed regardless of whether censoring occurs. 
Under this setting, \citet{candes2023conformalized} developed calibrated lower prediction bounds by discarding individuals with early censoring times and correcting for the induced covariate shift using weighted conformal prediction \citep{tibshirani2019conformal}. \citet{gui2024conformalized} extended this framework by allowing the censoring-time threshold to be selected in a covariate-dependent and data-adaptive manner, yielding tighter lower bounds while preserving validity. 

Subsequent work has addressed the more common right censoring mechanism, where only the minimum of the event time and censoring time is observed. 
Under this censoring setting, several approaches have been developed for constructing the lower prediction bounds including methods that reduce the problem to the type-I censoring setting through imputation of missing censoring times \citep{sesia2025doubly}, approaches that incorporate weights to account for right censoring \citep{davidov2025conformalized}, as well as methods that leverage semiparametric efficiency theory to handle censoring in doubly robust and efficient manners \citep{farina2025doubly, si2025training}. 

Beyond lower prediction bounds, several recent works developed two-sided prediction intervals for time-to-event outcomes under right censoring. 
\citet{holmes2024two} developed mixed-type prediction intervals using a classification-based approach: individuals are first classified into non-censored and censored groups, after which two-sided prediction intervals are constructed for individuals that are classified as non-censored with enough confidence; lower prediction bounds are constructed otherwise. 
\citet{qin2025conformal} developed a bootstrap-based approach under working models for the conditional event time distribution given covariates. 
\citet{yi2025survival} developed a weighted conformal method that accounts for right censoring using inverse probability of censoring weighting (IPCW), with weights estimated from a localized Kaplan Meier estimator of the conditional censoring distribution given covariates.

\section{General construction of the candidate prediction sets based on conformity scores}\label{app:conformity_score}

To form the candidate class 
$\{\Cc_{\tau,\theta}\}_{\theta\in\Theta}$,
one general strategy is to use a conformity score
$R_\tau(\bar z_\tau, t;\Ac)$ which measures how well a candidate event time $t$ conforms to the prediction based on covariate history $\bar z_\tau$ (smaller values indicating better conformity); see \citet{angelopoulos2024theoretical} for a comprehensive review. Specifically, the candidate class can be formed as
$\Cc_{\tau,\theta}(\bar z_\tau; \Ac) = \{t: R_\tau(\bar z_\tau, t; \Ac)\leq \theta\}$,
for $\theta\in\Theta$, where $\Theta$ contains all possible values of the conformity score.

In practice, to solve for $\hat\theta$ in Algorithm \ref{alg:split_dynamicCP}, one may perform a grid search on a discretized version of  $\Theta$ defined by a grid over the empirical range of the conformity score evaluated on the observed data. In particular, including the maximum observed conformity score in the grid guarantees that the empirical version of \eqref{eq:coverage_Cc_theta} is satisfied for at least one $\theta\in\Theta$, ensuring the existence of a solution.

There is a trade-off between grid resolution and computational cost: finer grids may yield more informative prediction sets but incur higher search cost. In practice, domain knowledge and multi-stage search can be used to restrict the range of $\Theta$ and improve computational efficiency.

\section{Nuisance parameter estimation}

\subsection{Estimation for $G$ with general time-varying covariate histories} \label{app:G_est}

The nuisance parameter $G$ can be estimated using existing regression methods for right censored time-to-event outcomes that accommodate dependent right censoring with time-varying covariates. For example, one can fit a proportional hazards model for censoring, using the \texttt{coxph} function (with the counting process formulation) in the `\texttt{survival}' R package and treating $1-\Delta$ as the event indicator. When stronger assumptions such as independent censoring or conditional independent censoring given baseline covariates hold, simpler estimators such as the Kaplan--Meier estimator or survival regression estimators that incorporate baseline covariates may be used instead.

\subsection{Estimation for $G$ and $h_{\tau}$ with covariates measured at discrete visit times}\label{app:h_est}

As noted in the main paper, estimating the function $h_{\tau}$ with general 
time-varying covariate histories is statistically challenging, because 
$h_{\tau}$ is defined through conditional expectations involving both the event 
time and the covariate history. When the covariate history is observed in 
continuous time, nonparametric estimation of $h_{\tau}$ is generally infeasible 
without imposing strong modeling assumptions on the joint distribution of the 
event time and the covariate process.

Here, we consider the setting where time-varying covariates are measured at 
prespecified times $0 = t_1<\cdots<t_K$. This setting is highly relevant in 
real-world applications where covariates are collected at planned visit 
times. For notational convenience, let $t_{K+1} = \infty$. For each 
$k=1,\ldots,K$, we define the interval-specific contributions of the event time 
and censoring time over $(t_k,t_{k+1}]$ as
\[
T_k = (T\wedge t_{k+1} - t_k)\ind(T>t_k),
\qquad
C_k = (C\wedge t_{k+1} - t_k)\ind(C>t_k).
\]
Let $X_k = \min(T_k,C_k)$ and 
$\Delta_k = \ind(0<T_k<C_k)$. With this notation, the original random variables can be written as
\[
T = \sum_{k=1}^K T_k, 
\qquad 
C = \sum_{k=1}^K C_k, 
\qquad 
X = \sum_{k=1}^K X_k, 
\qquad 
\Delta = \sum_{k=1}^K \Delta_k.
\]
Figure~\ref{fig:Tk_Ck} illustrates the data structure and the definitions of 
$T_k$ and $C_k$.

\vspace{0.5em}
\begin{figure}[ht!]
    \centering
    \begin{tikzpicture}[x=1cm, y=1cm, thick,
    event/.style = {circle, fill=black, inner sep=0pt, minimum size=4pt},
    cens/.style  = {circle, draw, fill=white, inner sep=0pt, minimum size=4pt},
    tag/.style   = {right=6pt},
    seg/.style   = {above}]

  \def\tA{0}  \def\tB{2.5}  \def\tC{5.7}  \def\tK{9.2}  \def\xend{10.2}
  \def\tick{0.12}

  \draw (\tA,4.3) -- (2.1,4.3) node[event] {} node[tag] {$T$};
  \draw (\tA,4.3-\tick) -- (\tA,4.3+\tick);
  \node[seg] at ({(\tA+2.1)/2},4.3) {$T_1$};

  \draw (\tA,3.3) -- (4.15,3.3) node[cens] {} node[tag] {$C$};
  \foreach \x in {\tA,\tB} \draw (\x,3.3-\tick) -- (\x,3.3+\tick);
  \node[seg] at ({(\tA+\tB)/2},3.3)  {$C_1 = t_2 - t_1$};
  \node[seg] at ({(\tB+4.15)/2},3.3) {$C_2$};

  \draw (\tA,2.3) -- (8.65,2.3) node[event] {} node[tag] {$T$};
  \foreach \x in {\tA,\tB,\tC} \draw (\x,2.3-\tick) -- (\x,2.3+\tick);
  \node[seg] at ({(\tA+\tB)/2},2.3)  {$T_1 = t_2 - t_1$};
  \node[seg] at ({(\tB+\tC)/2},2.3)  {$T_2 = t_3 - t_2$};
  \node[seg] at ({(\tC+8.65)/2},2.3) {$T_3$};

  \draw (\tA,1.3) -- (8.0,1.3) node[cens] {} node[tag] {$C$};
  \foreach \x in {\tA,\tB,\tC} \draw (\x,1.3-\tick) -- (\x,1.3+\tick);
  \node[seg] at ({(\tA+\tB)/2},1.3) {$C_1 = t_2 - t_1$};
  \node[seg] at ({(\tB+\tC)/2},1.3) {$C_2 = t_3 - t_2$};
  \node[seg] at ({(\tC+8.0)/2},1.3) {$C_3$};

  \draw[-{Straight Barb[length=2.2mm, width=2.5mm]}] (\tA,0) -- (\xend,0);
  \foreach \x/\lab in {\tA/{$t_1 = 0$}, \tB/{$t_2$}, \tC/{$t_3$}, \tK/{$t_K$}} {
    \draw (\x,-\tick) -- (\x,\tick);
    \node[below] at (\x,-\tick) {\lab};
  }
  \node[below] at ({(\tC+\tK)/2},-\tick) {$\cdots$};
  \node[below] at (\xend,-\tick) {time};

  \node[event] at (11.5,4.3) {};  \node[tag] at (11.5,4.3) {event};
  \node[cens]  at (11.5,3.8) {};  \node[tag] at (11.5,3.8) {censoring};
\end{tikzpicture}

    \caption{Illustration for the data structure and the definitions of $T_k$ and $C_k$.}
    \label{fig:Tk_Ck}
\end{figure}


Since the time-varying covariate process $Z_t$ changes values only at 
$t_1,\ldots,t_K$, we use the simplified notation
$L_k = Z_{t_k}$ and $\bar L_k = (L_1,\ldots,L_k)$ to denote the covariate
history observed up to time $t_k$. We focus on prediction at 
$\ttau = t_\nu$ for a fixed $\nu \in \{1,\ldots,K\}$. Since for any 
$u\in[t_k,t_{k+1})$, $\bar Z_u$ contains the same information as 
$\bar L_k$, we use these two notations interchangeably in this section and
write $\Cc_{\ttau,\theta}(\bar L_\nu;\Ac)$ for the candidate prediction sets.

Under this setting, Assumption~\ref{ass:cen_ind2} can be equivalently expressed
as follows.

\begin{assumption}[Sequential conditional 
independent censoring]\label{ass:cen_discrete}
    For each $k=1,\ldots,K$,
    \[
    C_k \bigCI (T_k,\ldots,T_K,L_{k+1},\ldots,L_K)
    \mid (T>t_k,\bar L_k).
    \]
\end{assumption}

We assume that $T$ and $C$ are absolutely continuous random variables and
introduce the following notation for their interval-specific conditional
survival functions. For $k=1,\ldots,K$ and $t\in(t_k,\infty)$, denote
\begin{align*}
    S_k(t|\bar L_k) 
    &= \PP(T>t\mid \bar L_k, X>t_k),  \\
    G_k(t|\bar L_k) 
    &= \PP(C>t\mid \bar L_k, X>t_k).
\end{align*}

\subsubsection{Piecewise estimation for $G$.}\label{app:G_est_piecewise}

In addition to the estimation approach described in Appendix~\ref{app:G_est}, which fits a global time-to-event model with time-varying covariates, an alternative approach is to exploit the following factorization of $G$ in terms of the interval-specific censoring survival functions $G_k$:
for $t\in(t_k,t_{k+1}]$,
\begin{align}
    G(t|\bar Z_t) 
    = \left\{\prod_{j=1}^{k-1} G_j(t_{j+1}|\bar L_j) \right\} 
    \cdot G_k(t|\bar L_k), 
    \label{eq:G_Gk}
\end{align}
where the empty product is defined as one. One can then fit separate models for each $G_k(t|\bar L_k)$ on the interval $t\in(t_k,t_{k+1}]$, and estimate $G$ by plugging the resulting estimates of $G_k$ into \eqref{eq:G_Gk}.

The conditional censoring survival curve $G_k(t|\bar L_k)$, for $t\in (t_k,t_{k+1}]$, can be estimated by fitting a survival regression model for the censoring time $C$ conditional on $\bar L_k$, using subjects who remain at risk at $t_k$ and imposing administrative censoring at $t_{k+1}$. 
Specifically, the model is fitted using data
\[
\left\{\left(X_i\wedge t_{k+1}, \tilde\Delta_i, \bar L_{k,i} \right): X_i>t_k\right\},
\]
where 
\[
\tilde\Delta_i = (1-\Delta_i)\ind\{X_i\in(t_k,t_{k+1}]\}
\]
is the event indicator for censoring.
Existing software for right-censored 
time-to-event outcomes can be used to estimate $G_k$, including proportional hazards models implemented in the `\texttt{survival}' R package, random survival forests implemented in the `\texttt{randomForestSRC}' R package, and the boosting trees implemented in the `\texttt{xgboost}' R package, among others.

\subsubsection{Estimation of $h_{\tau}$}

For ease of exposition, we focus on the quantile-based candidate prediction 
sets used in the numerical experiments of this paper:
\[
\Cc_{\tau,\theta}(\bar L_\nu;\Ac)
=
\left(q_{1-\theta}(\bar L_\nu;\Ac), q_{\theta}(\bar L_\nu;\Ac)\right),
\]
where $\theta\in[0.5, 1]$ and $q_{\theta}(\bar L_\nu;\Ac)$ denotes the conditional $\theta$-quantile 
of $T\mid (T>t_\nu,\bar L_\nu)$ computed from the prediction model $\Ac$. 
The same estimation approach can also be extended to other candidate prediction sets.

For any real numbers $a$ and $b$, we use $a\vee b$ and $a\wedge b$ to denote 
the maximum and minimum of $a$ and $b$, respectively. For any 
$u\in(t_k,t_{k+1}]$, we can show that (with the proof provided at the end of this subsection)
\begin{align}
    & \quad h_\ttau(u,\bar Z_u; \theta,\Ac) \nonumber \\
    & =
    \E\left(
    \left.
    \ind(T>\ttau)
    \left[
    \ind\{T\in \Cc_{\ttau,\theta}(\bar L_\nu;\Ac)\} - (1-\alpha)
    \right]
    \right| \bar L_k, T\geq u
    \right) \nonumber \\
    & =
    \left\{ 
    \begin{array}{ll}
        \displaystyle
        \frac{\xi_k(\bar L_k;\theta,\Ac,\alpha)}
        {S_k(u|\bar L_k)},  
        & \text{if } k < \nu, \text{ i.e., } u< \ttau, \\[2ex]
        \displaystyle
        \frac{
        S_k\!\left(u\vee q_{1-\theta}(\bar L_{\nu};\Ac)\mid\bar L_{k}\right)
        -
        S_k\!\left(u\vee q_{\theta}(\bar L_{\nu};\Ac)\mid\bar L_k\right)
        }
        {S_k(u|\bar L_k)}
        - (1-\alpha),  
        & \text{if } k \geq \nu, \text{ i.e., } u\geq \ttau,
    \end{array}
    \right. 
    \label{eq:h_tau}
\end{align}
where
\begin{align*}
    \xi_k(\bar L_k;\theta,\Ac,\alpha) 
    & =
    \E\left(
    \left.
    \ind(T>\ttau)
    \left[
    \ind\{T\in \Cc_{\ttau,\theta}(\bar L_\nu;\Ac)\} - (1-\alpha)
    \right]
    \right| \bar L_k, X\geq t_k
    \right).
\end{align*}

Therefore, it remains to estimate $\xi_k(\bar L_k)$ and 
$S_k(t|\bar L_k)$ for $t\in[t_k,\infty)$ and $k=1,\ldots,K$. The function 
$h_{\ttau}$ can then be estimated by plugging these estimates into 
\eqref{eq:h_tau}.

\paragraph{Estimation of $S_k$.}

We first consider estimating $S_k(t|\bar L_k)$ for $t\in[t_k,\infty)$. 
A natural approach is to fit a survival regression model for $T$ conditional on 
$\bar L_k$ using individuals in the landmark risk set at $t_k$, i.e., those with 
$X>t_k$. However, this approach is generally biased for $k<K$, because censoring 
after $t_k$ may depend on future covariates that are associated with the event 
time. Specifically, under Assumption~\ref{ass:cen_discrete}, $C$ may remain 
dependent on $T$ conditional on $\{\bar L_k, X>t_k\}$ for $k<K$.

To account for this residual dependence between the censoring time and the event 
time, we use inverse probability of censoring weighting. 
Denote $\bar t_{j,t}=t_{j+1}\wedge t$. We use the convention that
\[
G_j(s|\bar L_j)=1 \quad \text{for } s\leq t_j .
\]
Then, for $t>t_k$, define
\[
H_k(t)=\prod_{j=k}^K G_j(\bar t_{j,t}|\bar L_j).
\]

Under the sequential conditional 
independent censoring assumption (Assumption \ref{ass:cen_discrete}), $H_k(t)$ can be interpreted as the conditional 
probability of remaining uncensored through time $t$ in the landmark population 
at $t_k$, given the covariate history up to time $t$:
\[
H_k(t)
=
\PP(C>t\mid \bar Z_t, X>t_k).
\]

Then, for each $k=1,\ldots,K-1$, the survival curve $S_k(t|\bar L_k)$ can be 
estimated by fitting a weighted regression model for the event time $T$ 
conditional on $\bar L_k$ among individuals with $X>t_k$, using the uncensored 
observations with IPCW weights $\Delta/H_k(X)$. Equivalently, one may fit the 
model using observations with $X>t_k$ and $\Delta=1$, weighted by 
$1/H_k(X)$.

For the last interval, no future covariates remain. Thus, under 
Assumption~\ref{ass:cen_discrete}, $S_K(t|\bar L_K)$ for $t>t_K$ can be 
estimated by fitting a standard right-censored survival regression model for 
$T$ conditional on $\bar L_K$ among individuals with $X>t_K$.

In practice, this weighted estimation step can be implemented using existing software for survival outcomes, as long as the chosen method accommodates observation-level weights. For example, weighted Cox proportional hazards models can be fitted using the \texttt{survival} R package, and random survival forests and boosted survival models can be implemented using packages such as \texttt{randomForestSRC} and \texttt{xgboost}, respectively. 
The choice of model may be guided by the effective sample size in each landmark risk set.

\paragraph{Estimation of $\xi_k$.}

If there were no right censoring, $\xi_k$ could be estimated by regressing
\begin{align}
\ind(T>\ttau)
\left[
\ind\{T\in \Cc_{\ttau,\theta}(\bar L_\nu;\Ac)\} - (1-\alpha)
\right] \label{eq:y_tilde}
\end{align}
on $\bar L_k$ among individuals at risk at $t_k$. With right censoring, \eqref{eq:y_tilde} is not always observed, because $T$ and the covariate history $\bar L_\nu$ may be censored. We therefore again use
inverse probability of censoring weights to handle censoring.

Specifically, for each $\theta$, $\xi_k(\bar L_k;\theta,\Ac,\alpha)$ can be
estimated by fitting a weighted regression of the pseudo outcome
\begin{align}
\ind(X>\ttau)
\left[
\ind\{X\in \Cc_{\ttau,\theta}(\bar L_\nu;\Ac)\} - (1-\alpha)
\right] \label{eq:y_tilde_X}
\end{align}
on $\bar L_k$, using uncensored individuals who are at risk at $t_k$, i.e.,
those with $X>t_k$ and $\Delta=1$, with IPCW weights $\Delta/H_k(X)$.
For individuals with observed event times before $\ttau$, the pseudo outcome \eqref{eq:y_tilde_X} is zero.

We observe that the pseudo-outcome in \eqref{eq:y_tilde_X} takes only three 
possible values: $0$, $\alpha$, and $-(1-\alpha)$. Therefore, $\xi_k$ can be 
estimated either by fitting a weighted regression model directly to this 
pseudo-outcome, or by fitting a weighted multiclass classification model for 
the three possible outcome categories and then computing the conditional mean 
using the estimated class probabilities. 
Existing regression or classification software that accommodates 
observation-level weights can be used for estimating $\xi_k$. For example, 
flexible models can be implemented using boosting trees in the 
\texttt{xgboost} R package or random forests in the \texttt{randomForestSRC} 
R package. More generally, any regression or machine learning method that 
supports case weights may be used to estimate the conditional mean function 
$\xi_k(\bar L_k;\theta,\Ac,\alpha)$,
and the choice of model may be guided by the effective sample size in each landmark risk set.

\paragraph{Remark.}
As an alternative to the approach described above, $S_k(t|\cdot)$ at a fixed 
$t\in(t_k,\infty)$ and $\xi_k$ can also be nonparametrically identified and 
estimated using an iterative procedure similar to that of 
\citet{qiu2025multiply}. However, in our setting, we require the entire 
conditional survival curve $S_k(t|\cdot)$ over $t\in(t_k,\infty)$. Applying the 
iterative procedure to this task would be computationally prohibitive, as it 
would require fitting separate models for each time point $t$.

\paragraph{Proof of Equation~\eqref{eq:h_tau}.}

Let
\[
Y_{\theta}
=
\ind(T>\ttau)
\left[
\ind\{T\in \Cc_{\ttau,\theta}(\bar L_\nu;\Ac)\} - (1-\alpha)
\right].
\]
For $u\in(t_k,t_{k+1}]$, we have 
\begin{align*}
    h_\ttau(u,\bar Z_u; \theta,\Ac)
    &=
    \E\left(Y_{\theta}\mid \bar L_k, T\geq u\right)  \\
    &=
    \frac{
    \E\left\{\ind(T\geq u)Y_{\theta}
    \mid \bar L_k, T>t_k\right\}
    }{
    \PP(T\geq u\mid \bar L_k, T>t_k)
    }.
\end{align*}
Under the independent censoring assumption and the absolute continuity of $T$,
conditioning on $T>t_k$ can equivalently be replaced by conditioning on the
observed landmark risk set $X>t_k$. Thus,
\[
\PP(T\geq u\mid \bar L_k,T>t_k)
=
\PP(T\geq u\mid \bar L_k,X>t_k)
=
S_k(u|\bar L_k),
\]
where the distinction between $>$ and $\geq$ is immaterial by absolute
continuity.

(1) First consider $k<\nu$, so that $u<\ttau$. Recall that $\tau = t_{\nu}$. Since $T>\ttau$ implies
$T\geq u$, we have
\begin{align*}
    h_\ttau(u,\bar Z_u; \theta,\Ac)
    &=
    \frac{
    \E\left[
    \ind(T>\ttau)
    \left\{
    \ind\{T\in \Cc_{\ttau,\theta}(\bar L_\nu;\Ac)\}-(1-\alpha)
    \right\}
    \mid \bar L_k, X>t_k
    \right]
    }{
    S_k(u|\bar L_k)
    } \\
    &=
    \frac{\xi_k(\bar L_k;\theta,\Ac,\alpha)}
    {S_k(u|\bar L_k)}.
\end{align*}

Next consider $k\geq \nu$, so that $u\geq \ttau$. In this case,
$\bar L_\nu$ is contained in $\bar L_k$, and hence
$q_{1-\theta}(\bar L_\nu;\Ac)$ and $q_{\theta}(\bar L_\nu;\Ac)$ are known
conditional on $\bar L_k$ and given $\Ac$. For the quantile-based candidate set
\[
\Cc_{\ttau,\theta}(\bar L_\nu;\Ac)
=
\left(q_{1-\theta}(\bar L_\nu;\Ac),
q_{\theta}(\bar L_\nu;\Ac)\right), \quad \theta\in[0.5, 1], 
\]
we have
\begin{align*}
    h_\ttau(u,\bar Z_u; \theta,\Ac)
    &=
    \frac{
    \E\left[
    \ind(T\geq u)
    \left\{
    \ind\left(q_{1-\theta}(\bar L_\nu;\Ac)<T<
    q_{\theta}(\bar L_\nu;\Ac)\right)
    -(1-\alpha)
    \right\}
    \mid \bar L_k, X>t_k
    \right]
    }{
    S_k(u|\bar L_k)
    } \\
    &=
    \frac{
    \PP\left(
    T>u\vee q_{1-\theta}(\bar L_\nu;\Ac),
    \,
    T<q_{\theta}(\bar L_\nu;\Ac)
    \mid \bar L_k, X>t_k
    \right)
    }{
    S_k(u|\bar L_k)
    }
    -(1-\alpha) \\
    &=
    \frac{
    S_k\left(u\vee q_{1-\theta}(\bar L_\nu;\Ac)\mid \bar L_k\right)
    -
    S_k\left(u\vee q_{\theta}(\bar L_\nu;\Ac)\mid \bar L_k\right)
    }{
    S_k(u|\bar L_k)
    }
    -(1-\alpha).
\end{align*}

Combining the two cases, Equation \eqref{eq:h_tau} follows.

\FloatBarrier
\subsection{Computational cost of HAPS-A}\label{app:HAPSA_cost}

The most computationally expensive step of HAPS-A is fitting the nuisance models for the augmentation term. Table~\ref{tab:HAPSA_models} summarizes the nuisance models needed for HAPS and HAPS-A. We see that the number of nuisance models needed for HAPS-A is of order $O(K|\Theta|)$, where $|\Theta|$ denotes the size of the grid $\Theta$. Therefore, the computation time can be long when the number of visit times $K$ is large or when the grid is fine (so $|\Theta|$ is large).

To reduce the computation cost, one may apply the localization approach of \citet{kallus2024localized}, which has also been used in conformal prediction \citep{gao2024role}. Specifically, the localized version computes the augmentation term only at the initial HAPS estimator $\hat\theta$ and solves the estimating inequality with it. Then only $K-1$ models for $\xi_k$ need to be fitted at $\hat\theta$.

Table~\ref{tab:HAPSA_time} summarizes the computation time for HAPS, HAPS-A, and localized HAPS-A under different numbers of visit times $K$. We observe that localized HAPS-A is much faster than HAPS-A. For HAPS-A, the computation time is approximately linear in $K$.

\begin{table}[ht]
\centering
\small
\caption{Nuisance models needed for HAPS and HAPS-A at a single prediction time $\tau$.}
\label{tab:HAPSA_models}
\begin{tabular}{lp{0.66\linewidth}}
\toprule
Method & Nuisance models needed at a single $\tau$ \\
\midrule
HAPS             & 1 prediction model + 1 censoring model (or $K$ censoring models using the piecewise approach in Appendix~\ref{app:G_est_piecewise}) \\
HAPS-A           & Models in HAPS + $K$ models for $S_k$ + $(K-1)|\Theta|$ models in total for $\xi_k$ at each $\theta \in \Theta$ \\
Localized HAPS-A & Models in HAPS + $K$ models for $S_k$ + $(K-1)$ models in total for $\xi_k$ at $\hat\theta$ \\
\bottomrule
\end{tabular}
\end{table}

\begin{table}[ht]
\centering
\small
\caption{Computation time in seconds on a MacBook Air (M3 chip) for 3 replicates, under simulation Setup A, at $\tau = 3$, and with $\Theta$ being a uniform grid on $[0.5, 1]$ with increment 0.0025 ($|\Theta| = 201$).
}
\label{tab:HAPSA_time}
\begin{tabular}{llccc}
\toprule
$n$ & Method & $K = 3$ & $K = 5$ & $K = 9$ \\
\midrule
300  & HAPS             & 0.03 & 0.04 & 0.05 \\
     & HAPS-A           & 2.9  & 4.2  & 7.2  \\
     & Localized HAPS-A & 0.21 & 0.13 & 0.19 \\
\midrule
1000 & HAPS             & 0.14 & 0.12 & 0.14 \\
     & HAPS-A           & 18.4 & 28.0 & 45.5 \\
     & Localized HAPS-A & 0.49 & 0.53 & 0.75 \\
\bottomrule
\end{tabular}
\end{table}
\FloatBarrier

\section{Proof for identification of $\theta^*$}\label{app:proof_identification}

\begin{proof}[Proof of Lemma \ref{lem:identification_IPCW}]
    Recall $G(t|\bar Z_t) = \exp\!\left\{ - \int_0^t \lambda_C(u \mid \bar Z_u)\,du \right\}$.
    Under assumption \ref{ass:cen_ind2}, we have
    \begin{align}
    G(t|\bar Z_t) 
    = \exp\!\left\{ - \int_0^t \lambda_C(u \mid T, \bar Z_T)\,du \right\}
    = \PP(C>t \mid T, \bar Z_T), \quad t>0. \label{eq:G_Sc}
    \end{align}
    Recall $\Delta = \ind(T<C)$ and $X = \min(T,C)$. When $\Delta = 1$, we have $X = T$ and $\bar Z_X = \bar Z_T$. Therefore,
    \begin{align*}
        U(\theta;G,\Ac) 
        & = \frac{\Delta \ind(X >\tau)}{G(X|\bar Z_X)} \left[\ind\{X \in \Cc_{\tau,\theta}(\bar Z_\tau; \Ac)\} - (1-\alpha) \right] \\
        & = \frac{\Delta \ind(T >\tau)}{G(T|\bar Z_T)} \left[\ind\{T \in \Cc_{\tau,\theta}(\bar Z_\tau; \Ac)\} - (1-\alpha) \right] \\
        & = \frac{\ind(T<C)\ind(T >\tau)}{G(T|\bar Z_T)} \left[\ind\{T \in \Cc_{\tau,\theta}(\bar Z_\tau; \Ac)\} - (1-\alpha) \right]
    \end{align*}
    By tower property of expectations,
    \begin{align}
        \E\left\{ U(\theta;G,\Ac) \right\} 
        & =  \E\left(\frac{\ind(T<C)\ind(T >\tau)}{G(T|\bar Z_T)} \left[\ind\{T \in \Cc_{\tau,\theta}(\bar Z_\tau; \Ac)\} - (1-\alpha) \right] \right) \nonumber \\
        & =  \E\left\{\E\left(\left. \frac{\ind(T<C)\ind(T >\tau)}{G(T|\bar Z_T)} \left[\ind\{T \in \Cc_{\tau,\theta}(\bar Z_\tau; \Ac)\} - (1-\alpha) \right] \right| T,\bar Z_T \right)\right\} \nonumber\\
        & =  \E\left(\frac{\E\left\{\left.\ind(T<C) ~\right|~ T,\bar Z_T\right\} }{G(T|\bar Z_T)}  \cdot \ind(T >\tau) \left[\ind\{T \in \Cc_{\tau,\theta}(\bar Z_\tau; \Ac)\} - (1-\alpha) \right] \right) \nonumber \\
        & =  \E\left(\frac{\PP\left(\left.C>T ~\right|~ T,\bar Z_T\right) }{G(T|\bar Z_T)}  \cdot \ind(T >\tau) \left[\ind\{T \in \Cc_{\tau,\theta}(\bar Z_\tau; \Ac)\} - (1-\alpha) \right] \right) \nonumber \\
        & = \E\left(\ind(T >\tau) \left[\ind\{T \in \Cc_{\tau,\theta}(\bar Z_\tau; \Ac)\} - (1-\alpha) \right] \right) \label{eq:identification_proof1} \\
        & = \E\{U^*(\theta;\Ac)\} \nonumber,
    \end{align}
    where \eqref{eq:identification_proof1}  holds by \eqref{eq:G_Sc}. 
    
    Therefore, 
    \[
    \theta^* = \inf\{\theta\in\Theta:\E\{U^*(\theta;\Ac)\}\geq 0\}
    =  \inf\{\theta\in\Theta:\E\{U(\theta;G,\Ac)\}\geq 0\}. 
    \]
    
\end{proof}

\section{Additional detials and proofs for the coverage guarantees}\label{app:proof_coverage}

\subsection{Proof of Theorem \ref{thm:IPCW_coverage}}

We first prove the following two lemmas, which will be used to prove the coverage guarantee in Theorem \ref{thm:IPCW_coverage}. 

Recall that $\Dc$ denotes the observed sample, which is split into two disjoint sets $\Dctr$ and $\Dccal$. The set $\Dctr$ is used to fit the prediction model $\Ac$ and estimate the nuisance parameter $G$; the set $\Dccal$ is used to obtain $\hat\theta$. 

\begin{lemma}\label{lem:thm_IPCW_proof1}
     Suppose that  Assumptions \ref{ass:cen_ind2} and \ref{ass:strict_positivity} hold, and that $\hat G(X\mid\bar Z_X)>\eta$ almost surely.
     Then for any $\epsilon \in(0,1)$, there exists a constant $\consta>0$ such that with probability at least $1-\epsilon$,
    \begin{align*}
        \sup_{\theta} \left| \frac{1}{|\Dccal|} \sum_{i\in \Iccal} U_{i}(\theta;\hat G,\hat\Ac)  - \E\left\{\left.  U(\theta;\hat G,\hat\Ac) \right| \Dc \right\} \right|
        \leq \eta^{-1} \left( \sqrt{\frac{1}{2}\log \frac{1}{\epsilon}} + \consta \right)\frac{1}{\sqrt{|\Dccal|}}.
    \end{align*}
\end{lemma}

\begin{proof}[Proof of Lemma \ref{lem:thm_IPCW_proof1}]
    The conclusion follows by applying McDiarmid's inequality for empirical processes and Lemma 4 of \citet{yang2024doubly}, using similar arguments as the proof of Lemma S.2.1 in the supplement of \citet{farina2025doubly}. 
\end{proof}

\begin{lemma}\label{lem:thm_IPCW_proof2}
     Suppose that Assumptions \ref{ass:cen_ind2} and \ref{ass:strict_positivity} hold, and that $\hat G(X\mid\bar Z_X)>\eta$ almost surely. Then we have 
     \begin{align*}
         \sup_{\theta} \left| \E\left\{\left.  U(\theta;\hat G,\hat\Ac) \right| \Dc \right\} - \E\left\{\left.  U(\theta;G,\hat\Ac) \right| \Dc \right\} \right|
         \leq \eta^{-2} \|\hat G-G\|_2. 
     \end{align*}
\end{lemma}

\begin{proof}[Proof of Lemma \ref{lem:thm_IPCW_proof2}]
    The conclusion follows by applying Jensen's inequality and Cauchy-Schwarz inequality, using similar arguments as the proof of Lemma S.2.2 in the supplement of \citet{farina2025doubly}.
\end{proof}

\begin{proof}[Proof of Theorem \ref{thm:IPCW_coverage}]

    Recall from Section \ref{sec:preliminary} that $(T_{n+1}, \bar Z_{\tau, n+1})$ are from the new test individual which is independent from $\Dc$. 
    Recall from \eqref{eq:D(theta)} the definition of $U^*$, and from Algorithm \ref{alg:split_dynamicCP} that $\hat \Cc(\bar Z_{\tau,n+1}) = \Cc_{\tau, \hat\theta}(\bar Z_{\tau,n+1};\hat\Ac)$. 
    Since $\hat\Ac$ and $\hat\theta$ are estimated from $\Dc$, they are fixed after conditional on $\Dc$.
    Therefore, by Lemma \ref{lem:identification_IPCW}, 
    \begin{align}
        &\quad \E\left\{\left. U_{ n+1}(\hat\theta;G,\hat A) \right| \Dc \right\} \\
        & = \E\left\{\left. U^*_{n+1}(\hat\theta;\hat A) \right| \Dc \right\} \nonumber \\
        & = \E\left(\left. \ind(T_{n+1} > \tau) \cdot \left[ \ind\{T_{n+1}\in \hat\Cc(\bar Z_{\tau,n+1})\} - (1-\alpha) \right] \right| \Dc\right) \nonumber \\
        & = \PP\left\{\left. T_{n+1}>\tau,~ T_{n+1}\in \hat\Cc(\bar Z_{\tau,n+1}) \right| \Dc \right\} - (1-\alpha)\cdot \PP(T_{n+1}>\tau) \nonumber \\
        & = \PP\left\{\left. T_{n+1}\in \hat\Cc(\bar Z_{\tau,n+1}) \right| T_{n+1}>\tau,\Dc \right\}\cdot \PP(T_{n+1}>\tau) - (1-\alpha)\cdot \PP(T_{n+1}>\tau) \nonumber \\
        & = \left[ \PP\left\{\left. T_{n+1}\in \hat\Cc(\bar Z_{\tau,n+1}) \right| T_{n+1}>\tau, \Dc \right\} - (1-\alpha) \right] \cdot \sttau. \label{eq:coverage_proof_E1}
    \end{align}
    On the other hand, 
    \begin{align}
        &\quad \E\left\{\left. U_{ n+1}(\hat\theta;G,\hat A) \right| \Dc \right\} \nonumber\\
        & = \frac{1}{|\Dccal|}\sum_{i\in\Iccal}  U_{i}(\hat\theta;\hat G,\hat\Ac)  - \left[\frac{1}{|\Dccal|}\sum_{i\in\Iccal}  U_{i}(\hat\theta;\hat G,\hat\Ac)  - \E\left\{\left. U_{ n+1}(\hat\theta;G,\hat A) \right| \Dc \right\} \right] \nonumber\\
        & \geq \frac{1}{|\Dccal|} \sum_{i\in\Iccal}  U_{i}(\hat\theta;\hat G,\hat\Ac) 
        - \left| \frac{1}{|\Dccal|}\sum_{i\in\Iccal}  U_{i}(\hat\theta;\hat G,\hat\Ac) - \E\left\{\left. U_{ n+1}(\hat\theta;G,\hat A) \right| \Dc \right\} \right| \nonumber \\
        &\geq - \sup_{\theta}\left| \frac{1}{|\Dccal|}\sum_{i\in\Iccal}  U_{i}(\theta;\hat G,\hat\Ac) - \E\left\{\left. U_{ n+1}(\theta;G,\hat A) \right| \Dc \right\} \right|,  \label{eq:coverage_proof_E2}
    \end{align}
    where \eqref{eq:coverage_proof_E2} holds because $ \sum_{i\in\Iccal}  U_{i}(\hat\theta;\hat G,\hat\Ac) \geq 0$ (by the definition of $\hat\theta$ in Algorithm \ref{alg:split_dynamicCP}).
    
    Combining \eqref{eq:coverage_proof_E1} and \eqref{eq:coverage_proof_E2}, we have 
    \begin{align}
         &\quad \PP\left\{\left. T_{n+1}\in \hat\Cc(\bar Z_{\tau,n+1}) \right| T_{n+1}>\tau \right\} \nonumber\\
         & \geq 1-\alpha - \sttau^{-1}\cdot \sup_{\theta}\left| \frac{1}{|\Dccal|}\sum_{i\in\Iccal}  U_{i}(\hat\theta;\hat G,\hat\Ac) - \E\left\{\left. U_{ n+1}(\hat\theta;G,\hat A) \right| \Dc \right\} \right|. \label{eq:coverage_proof_3}
    \end{align}
    
    We now bound the above supremum term. 
    By Lemma \ref{lem:thm_IPCW_proof1} and Lemma \ref{lem:thm_IPCW_proof2}, we have that for any $\epsilon\in(0,1)$, there exists a constant $\consta>0$ such that with probability at least $1-\epsilon$,
    \begin{align}
        &\quad \sup_\theta\left| \frac{1}{|\Dccal|}\sum_{i\in\Iccal}  U_{i}(\theta;\hat G,\hat\Ac) - \E\{U(\theta; G,\hat\Ac) \mid \Dc\} \right| \nonumber \\
        & \leq \sup_\theta\left| \frac{1}{|\Dccal|}\sum_{i\in\Iccal}  U_{i}(\theta; \hat G,\hat\Ac) - \E\{U(\theta; \hat G,\hat\Ac)\mid \Dc\} \right| \nonumber\\
        &\quad + \sup_\theta\left|\E\{U(\theta; \hat G,\hat\Ac)\mid\Dc\} - \E\{U(\theta;G,\hat\Ac)\mid\Dc\} \right| \nonumber\\
        &\leq \eta^{-1} \left( \sqrt{\frac{1}{2}\log \frac{1}{\epsilon}} + \consta \right)\frac{1}{\sqrt{|\Dccal|}} 
        + \eta^{-2} \|\hat G-G\|_2.  \label{eq:coverage_proof_4}
    \end{align}
    Combining \eqref{eq:coverage_proof_3} and \eqref{eq:coverage_proof_4} concludes the proof.
\end{proof}

\subsection{Additional details and proof for the coverage guarantee of HAPS-A}
\label{app:proof_AIPCW}

As in the main paper, we assume that $\hat G$ admits a nonnegative censoring hazard function $\hat\lambda_C$, with corresponding cumulative hazard
\[
\hat\Lambda_C(t\mid\bar Z_t)
= \int_0^t\hat\lambda_C(u\mid\bar Z_u)\,du
= -\log\hat G(t\mid\bar Z_t).
\]
Let $R(u)=\ind(X\geq u)$ and let
$\hat M_C(t)
= N_C(t)-\int_0^t R(u)\hat\lambda_C(u\mid\bar Z_u)\,du$.
For a function $f(u,\bar z_u)$, possibly depending on $\Dctr$,
define the following norms using a new observation independent of $\Dc$:
\begin{align*}
\|f(u,\bar Z_u)\|_2
&=\E\{R(u)|f(u,\bar Z_u)|^2\mid\Dc\}^{1/2},\\
\|f(X,\bar Z_X)\|_2
&=\E\{|f(X,\bar Z_X)|^2\mid\Dc\}^{1/2},\\
\|f\|_{\sup,2}
&=\E\Big\{\sup_{0\leq u\leq X}|f(u,\bar Z_u)|^2
\,\Big|\,\Dc\Big\}^{1/2}.
\end{align*}

To simplify notation, we use the following shorthand whenever there is no ambiguity:
\begin{align*}
Y_\theta
&=U^*(\theta;\hat\Ac)
=\ind(T>\tau)\big[
\ind\{T\in\Cc_{\tau,\theta}(\bar Z_\tau;\hat\Ac)\}
-(1-\alpha)\big],\\
h_\tau(u;\theta)
&=h_\tau(u,\bar Z_u;\theta,\hat\Ac)
=\E(Y_\theta\mid\bar Z_u,T\geq u,\Dc),\\
\hat h_\tau(u;\theta)
&=\hat h_\tau(u,\bar Z_u;\theta,\hat\Ac).
\end{align*}
From their definitions, $Y_\theta\in\{0,\alpha,-(1-\alpha)\}$ and
$h_\tau(u;\theta)\in[-(1-\alpha),\alpha]$.
Using this shorthand and the same norm notation for these processes, we consider the following product errors:
\begin{align}
r_1(\theta)
&=\int_0^\infty
\|\hat h_\tau(u;\theta)-h_\tau(u;\theta)\|_2\,
\|\hat\lambda_C(u\mid\bar Z_u)
-\lambda_C(u\mid\bar Z_u)\|_2\,du,
\label{eq:r1}\\
r_2(\theta)
&=\|\hat h_\tau(X;\theta)-h_\tau(X;\theta)\|_2\,
\|\hat\Lambda_C(X\mid\bar Z_X)
-\Lambda_C(X\mid\bar Z_X)\|_2
\nonumber\\
&\quad+\int_0^\infty
\|\Lambda_C(u\mid\bar Z_u)
-\hat\Lambda_C(u\mid\bar Z_u)\|_2\,
\left\|\frac{d}{du}
\{\hat h_\tau(u;\theta)-h_\tau(u;\theta)\}\right\|_2\,du,
\label{eq:r2}\\
r_3(\theta)
&=\|\hat\Lambda_C-\Lambda_C\|_{\sup,2}\,
\|\hat h_\tau(\cdot;\theta)-h_\tau(\cdot;\theta)\|_{\sup,2}.
\label{eq:r3}
\end{align}
The derivative in \eqref{eq:r2} is taken along the realized covariate history.
We set $r_2(\theta)=\infty$ if the function
$u\mapsto\hat h_\tau(u;\theta)-h_\tau(u;\theta)$
fails to be absolutely continuous on $[0,X]$ with positive conditional probability given $\Dc$.

\begin{lemma}\label{lem:thm_AIPCW_proof1}
Suppose that Assumptions~\ref{ass:cen_ind2} and~\ref{ass:strict_positivity}
hold, that $\hat G(X\mid\bar Z_X)>\eta$ almost surely, and that
$\hat h_\tau$ takes values in $[-(1-\alpha),\alpha]$.
Then for any $\epsilon\in(0,1)$, with probability at least
$1-\epsilon$ over $\Dc$,
\begin{align*}
&\sup_{\theta\in\Theta}\left|
\frac{1}{|\Dccal|}\sum_{i\in\Iccal}
U_{\AIPCW,i}(\theta;\hat G,\hat h_\tau,\hat\Ac)
-\E\{U_{\AIPCW}(\theta;\hat G,\hat h_\tau,\hat\Ac)\mid\Dc\}
\right|\\
&\qquad\leq
\frac{c_\eta}{\sqrt{|\Dccal|}}
\left(\sqrt{\frac{1}{2}\log\frac{1}{\epsilon}}+\consta\right),
\end{align*}
where $c_\eta=2\max(\alpha,1-\alpha)(2\eta^{-1}-1)$
and $\consta=\sqrt{\frac{1}{2}\log(2|\Theta|)}$.
\end{lemma}

\begin{proof}[Proof of Lemma~\ref{lem:thm_AIPCW_proof1}]
Conditional on $\Dctr$, the fitted prediction model and nuisance estimates
are fixed, and the calibration observations are i.i.d.
For each fixed $\theta$, the conditional mean of
$U_{\AIPCW,i}(\theta;\hat G,\hat h_\tau,\hat\Ac)$ given $\Dctr$
equals the population expectation
$\E\{U_{\AIPCW}(\theta;\hat G,\hat h_\tau,\hat\Ac)\mid\Dc\}$
in the lemma. Indeed, the latter expectation is taken over a new
observation independent of $\Dc$, and all fitted functions depend
only on $\Dctr$.

Expanding \eqref{eq:U_AIPCW} gives
\begin{align*}
U_{\AIPCW}(\theta;\hat G,\hat h_\tau,\hat\Ac)
&=\frac{\Delta Y_\theta+(1-\Delta)\hat h_\tau(X;\theta)}
{\hat G(X\mid\bar Z_X)}\\
&\quad-\int_0^X
\hat h_\tau(u;\theta)
\frac{\hat\lambda_C(u\mid\bar Z_u)}
{\hat G(u\mid\bar Z_u)}\,du.
\end{align*}
The numerator of the first term is bounded in absolute value by
$\max(\alpha,1-\alpha)$.
For the integral term, the representation of $\hat G$ gives
\[
\int_0^X
\frac{\hat\lambda_C(u\mid\bar Z_u)}
{\hat G(u\mid\bar Z_u)}\,du
=\frac{1}{\hat G(X\mid\bar Z_X)}-1.
\]
By boundedness of $\hat h_\tau$ and positivity of $\hat G$,
\begin{align*}
\left|U_{\AIPCW}(\theta;\hat G,\hat h_\tau,\hat\Ac)\right|
&\leq \max(\alpha,1-\alpha)
\left\{\frac{2}{\hat G(X\mid\bar Z_X)}-1\right\}\\
&\leq \frac{c_\eta}{2},
\end{align*}
uniformly over $\theta\in\Theta$.

Let $m=|\Dccal|$ and let $W$ denote the supremum in the lemma.
For each fixed $\theta$, the estimating-function values lie in
$[-c_\eta/2,c_\eta/2]$, an interval of length $c_\eta$.
Applying Hoeffding's inequality conditional on $\Dctr$ and taking
a union bound over $\Theta$ gives
\[
\PP(W>t\mid\Dctr)
\leq 2|\Theta|\exp\left(-\frac{2mt^2}{c_\eta^2}\right),
\qquad t>0.
\]
Taking $t=c_\eta\sqrt{\log(2|\Theta|/\epsilon)/(2m)}$ and using
$\sqrt{a+b}\leq\sqrt a+\sqrt b$ gives the stated bound with
conditional probability at least $1-\epsilon$ given $\Dctr$.
By the law of total probability, the bound holds with probability
at least $1-\epsilon$ over $\Dc$.
\end{proof}

\begin{lemma}\label{lem:thm_AIPCW_proof2}
Suppose that Assumptions~\ref{ass:cen_ind2} and~\ref{ass:strict_positivity}
hold, that $\hat G(X\mid\bar Z_X)>\eta$ almost surely, and that
$\hat h_\tau$ is bounded.
Then, almost surely, for each $\theta\in\Theta$,
\[
\E\{U_{\AIPCW}(\theta;G,h_\tau,\hat\Ac)\mid\Dc\}
=\E(Y_\theta\mid\Dc),
\]
and
\begin{align}
&\left|
\E\{U_{\AIPCW}(\theta;\hat G,\hat h_\tau,\hat\Ac)\mid\Dc\}
-\E\{U_{\AIPCW}(\theta;G,h_\tau,\hat\Ac)\mid\Dc\}
\right|\\
&\qquad\leq
\min\{\eta^{-1}r_1(\theta),\,
\eta^{-1}r_2(\theta)+(\eta^{-1}-1)r_3(\theta)\}. \label{eq:lemmaE4_LHS}
\end{align}
\end{lemma}

\begin{proof}[Proof of Lemma~\ref{lem:thm_AIPCW_proof2}]

(1) We first show the first claim.
Fix $\theta$ and condition on $\Dc$, so the fitted functions are fixed
and the expectations are over a new observation. With some algebra,
\eqref{eq:U_AIPCW} can be written as
\[
U_{\AIPCW}(\theta;\hat G,\hat h_\tau,\hat\Ac)
=
Y_\theta-\int_0^X
\frac{Y_\theta-\hat h_\tau(u;\theta)}
{\hat G(u\mid\bar Z_u)}\,d\hat M_C(u).
\]

Consider the filtration
$\mathcal F_t^*=\sigma(\Dc,T,\bar Z_T,\{N_C(s):0\leq s\leq t\})$.
By Assumption~\ref{ass:cen_ind2}, the intensity of $N_C(u)$ in this
filtration is $R(u)\lambda_C(u\mid\bar Z_u)$, so $M_C$ is a
martingale. The fitted functions, $Y_\theta$, and the covariate
history through $T$ are all known at time zero. Thus, extending
the integrand by zero for $u>X$ gives a predictable process,
since $R(u)=\ind(X\geq u)$ is predictable. The integrand is bounded
on $[0,X]$ because $Y_\theta$ and $\hat h_\tau$ are bounded and
$\hat G(u\mid\bar Z_u)\geq\hat G(X\mid\bar Z_X)>\eta$ for $u\leq X$.
Moreover, the compensator has finite conditional expectation
because $N_C$ has at most one jump. It follows that
\[
\E\left\{
\left.
\int_0^X
\frac{Y_\theta-\hat h_\tau(u;\theta)}
{\hat G(u\mid\bar Z_u)}\,dM_C(u)
\right|\Dc
\right\}=0.
\]
By Assumption~\ref{ass:cen_ind2} and the definition of $h_\tau$, we have
\[
\E(Y_\theta\mid\bar Z_u,X\geq u,\Dc)
=\E(Y_\theta\mid\bar Z_u,T\geq u,\Dc)
=h_\tau(u;\theta).
\]
In addition, we have
\[
d\hat M_C(u)=dM_C(u)
+R(u)\{\lambda_C(u\mid\bar Z_u)
-\hat\lambda_C(u\mid\bar Z_u)\}\,du.
\]
Substituting this decomposition and applying iterated expectations gives
\begin{align}
&\E\{U_{\AIPCW}(\theta;\hat G,\hat h_\tau,\hat\Ac)\mid\Dc\}
-\E(Y_\theta\mid\Dc)
\nonumber\\
&\qquad=
\E\left\{
\left.
\int_0^X
\frac{\hat h_\tau(u;\theta)-h_\tau(u;\theta)}
{\hat G(u\mid\bar Z_u)}\,
d\{\Lambda_C(u\mid\bar Z_u)-\hat\Lambda_C(u\mid\bar Z_u)\}
\right|\Dc\right\}.
\label{eq:AIPCW_mixed_bias}
\end{align}
Setting $(\hat G,\hat h_\tau)=(G,h_\tau)$ makes the right-hand side
of \eqref{eq:AIPCW_mixed_bias} zero, proving the first claim.

(2) We now show the second claim, and we will show that the left hand side (L.H.S.) of \eqref{eq:lemmaE4_LHS} can be upper bounded by both $\eta^{-1}r_1(\theta)$ and $\eta^{-1}r_2(\theta)+(\eta^{-1}-1)r_3(\theta)$, which concludes the proof.

To obtain the first bound, write the cumulative-hazard differential
as the hazard difference times $du$. Since
$\hat G(u\mid\bar Z_u)>\eta$ for $u\leq X$, the absolute value of
the right-hand side of \eqref{eq:AIPCW_mixed_bias} is bounded by
\begin{align*}
&\eta^{-1}\int_0^\infty
\E\big[
R(u)|\hat h_\tau(u;\theta)-h_\tau(u;\theta)|
\,|\lambda_C(u\mid\bar Z_u)-\hat\lambda_C(u\mid\bar Z_u)|
\mid\Dc
\big]\,du\\
&\qquad\leq\eta^{-1}r_1(\theta),
\end{align*}
where the last inequality follows from the Cauchy--Schwarz inequality
conditional on $\Dc$ with the at-risk weight $R(u)$.

For the second bound, suppose that $r_2(\theta)<\infty$.
Suppressing the covariate history in the arguments of $\hat G$,
$\Lambda_C$, and $\hat\Lambda_C$, integration by parts and the
product rule give
\begin{align}
&\int_0^X
\frac{\hat h_\tau(u;\theta)-h_\tau(u;\theta)}{\hat G(u)}
\,d\{\Lambda_C(u)-\hat\Lambda_C(u)\}\nonumber \\
&\quad=
\frac{\{\hat h_\tau(X;\theta)-h_\tau(X;\theta)\}
\{\Lambda_C(X)-\hat\Lambda_C(X)\}}{\hat G(X)}\nonumber\\
&\qquad-
\int_0^X
\frac{\Lambda_C(u)-\hat\Lambda_C(u)}{\hat G(u)}
\frac{d}{du}\{\hat h_\tau(u;\theta)-h_\tau(u;\theta)\}\,du \nonumber\\
&\qquad-
\int_0^X
\{\Lambda_C(u)-\hat\Lambda_C(u)\}
\{\hat h_\tau(u;\theta)-h_\tau(u;\theta)\}
\,d\{1/\hat G(u)\}. \label{eq:lemE4_proof_bound1}
\end{align}
By positivity and the Cauchy--Schwarz inequality
conditional on $\Dc$, the sum of the conditional expected absolute
values of the first two terms in \eqref{eq:lemE4_proof_bound1} is at most $\eta^{-1}r_2(\theta)$.

Since $1/\hat G(u)$ is nondecreasing
and satisfies
\[
\int_0^X d\{1/\hat G(u)\}
=\frac{1}{\hat G(X)}-1
\leq\eta^{-1}-1.
\]
Therefore, the absolute value of the last integral in
\eqref{eq:lemE4_proof_bound1} is bounded by

\begin{align*}
&\left|
\int_0^X
\{\Lambda_C(u)-\hat\Lambda_C(u)\}
\{\hat h_\tau(u;\theta)-h_\tau(u;\theta)\}
\,d\{1/\hat G(u)\}
\right| \nonumber\\
&\quad\leq
\sup_{0\leq u\leq X}|\Lambda_C(u)-\hat\Lambda_C(u)|
\sup_{0\leq u\leq X}|\hat h_\tau(u;\theta)-h_\tau(u;\theta)|
\int_0^X d\{1/\hat G(u)\} \nonumber\\
&\quad\leq
(\eta^{-1}-1)
\sup_{0\leq u\leq X}|\Lambda_C(u)-\hat\Lambda_C(u)|
\sup_{0\leq u\leq X}|\hat h_\tau(u;\theta)-h_\tau(u;\theta)| 
\end{align*}
Taking expectations conditional on $\Dc$ and applying the
Cauchy--Schwarz inequality to the product of the two suprema
bounds the conditional expected absolute value of last integral in \eqref{eq:lemE4_proof_bound1} by
\begin{align*}
(\eta^{-1}-1)
\|\Lambda_C-\hat\Lambda_C\|_{\sup,2}
\|\hat h_\tau(\cdot;\theta)-h_\tau(\cdot;\theta)\|_{\sup,2}
=(\eta^{-1}-1)r_3(\theta).
\end{align*}

Combining the above gives the second bound. 
Taking the minimum of the two bounds proves the result. 

When
$r_2(\theta)=\infty$, the first bound alone gives the stated minimum.
\end{proof}

\begin{proof}[Proof for Theorem~\ref{thm:AIPCW_coverage}]
For fixed $\theta\in\Theta$, write
\[
\bar U_{\AIPCW}(\theta)
=\frac{1}{|\Dccal|}\sum_{i\in\Iccal}
U_{\AIPCW,i}(\theta;\hat G,\hat h_\tau,\hat\Ac),
\qquad
\mu_0(\theta)=\E(Y_\theta\mid\Dc).
\]
By feasibility of the augmented calibration problem and finiteness
of $\Theta$, the selected parameter $\hat\theta_{\AIPCW}$ satisfies
$\bar U_{\AIPCW}(\hat\theta_{\AIPCW})\geq 0$. Hence
\begin{align*}
\mu_0(\hat\theta_{\AIPCW})
&\geq \bar U_{\AIPCW}(\hat\theta_{\AIPCW})
-\left|\bar U_{\AIPCW}(\hat\theta_{\AIPCW})
-\mu_0(\hat\theta_{\AIPCW})\right|\\
&\geq -\left|\bar U_{\AIPCW}(\hat\theta_{\AIPCW})
-\mu_0(\hat\theta_{\AIPCW})\right|\\
&\geq -\sup_{\theta\in\Theta}
\left|\bar U_{\AIPCW}(\theta)-\mu_0(\theta)\right|.
\end{align*}
The triangle inequality bounds this supremum by the sum of the
empirical deviation and the population bias.
Lemma~\ref{lem:thm_AIPCW_proof1} bounds the empirical deviation.
Lemma~\ref{lem:thm_AIPCW_proof2} bounds the population bias.
Thus, with probability at least $1-\epsilon$ over $\Dc$,
\begin{align*}
\mu_0(\hat\theta_{\AIPCW})
&\geq
-\frac{c_\eta}{\sqrt{|\Dccal|}}
\left(\sqrt{\frac12\log\frac1\epsilon}+\consta\right)\\
&\quad
-\eta^{-1}\sup_{\theta\in\Theta}
\min\{r_1(\theta),\,r_2(\theta)+(1-\eta)r_3(\theta)\}.
\end{align*}

Conditional on $\Dc$, $\hat\theta_{\AIPCW}$ is fixed. 
Since the test individual is independent of $\Dc$, evaluating $\mu_0$ at $\hat\theta_{\AIPCW}$ and using the definition of $Y_\theta$ gives
\[
\mu_0(\hat\theta_{\AIPCW})
=\sttau\left[
\PP\{T_{n+1}\in\hat\Cc_{\AIPCW}(\bar Z_{\tau,n+1})
\mid T_{n+1}>\tau,\Dc\}-(1-\alpha)
\right].
\]
Substituting the lower bound on $\mu_0(\hat\theta_{\AIPCW})$
and dividing by $\sttau>0$ proves the theorem.
\end{proof}

\newpage
\section{Additional details and results for simulation}\label{app:simu}

\subsection{Data generating mechanisms}\label{app:simu_DGM}

\subsubsection{DGM1}
We generate a covariate process that is piecewise constant over prespecified time
intervals defined by
$0 = t_1 < \cdots < t_K < t_{K+1} = \infty$, where $K \ge 2$.
Specifically, $Z_t = Z_{t_k}$ for all $t \in [t_k, t_{k+1})$, $k = 1, \ldots, K$.
The covariate values $\{Z_{t_k}\}_{k=1}^K$ are generated from an autoregressive model
of order 1:
\begin{align*}
    Z_{t_{k+1}} = \rho Z_{t_k} + \sqrt{1 - \rho^2}\,\epsilon_k,
    \qquad k = 1, \ldots, K,
\end{align*}
where $Z_0 \sim N(0,1)$, $\epsilon_1,\ldots,\epsilon_K \stackrel{\text{i.i.d.}}{\sim} N(0,1)$
are independent of $Z_0$, and $\rho \in [-1,1]$ controls the temporal dependence of the
covariate process.
In particular, $\rho = \mathrm{Corr}(Z_{t_k}, Z_{t_{k-1}})$.
The setting with $\rho = 0$ corresponds to i.i.d.\ covariates $Z_{t_k} \sim N(0,1)$,
whereas $\rho = 1$ corresponds to time-invariant covariates with
$Z_{t_k} = Z_0$ for all $k \in \{1,\ldots,K\}$.

Conditional on the covariate history $\bar Z_{t_k}$, the event time increment $T_k$
and the censoring time increment $C_k$ for the interval $[t_k, t_{k+1})$ are generated independently from Weibull distributions:
\begin{align*}
    T_k &\sim \mathrm{Weibull}\!\left(a_{T,k}, \sigma_{T,k}(\bar Z_{t_k})\right), \\
    C_k &\sim \mathrm{Weibull}\!\left(a_{C,k}, \sigma_{C,k}(\bar Z_{t_k})\right),
\end{align*}
where $a_{T,k}, a_{C,k} > 0$ are shape parameters, and the scale parameters are given by
\begin{align}
    \sigma_{T,k}(\bar Z_{t_k})
    & = \exp\!\left\{
        - \frac{\beta_{T,0}^{(k)} + \beta_{T}^{(k)} Z_{t_k}}{a_{T,k}}
      \right\}, \label{eq:sigma_T} \\
    \sigma_{C,k}(\bar Z_{t_k})
    & = \exp\!\left\{
        - \frac{\beta_{C,0}^{(k)} + \beta_{C}^{(k)} Z_{t_k}}{a_{C,k}}
      \right\}. \label{eq:sigma_C}
\end{align}
Let $\zeta_1 = 1$ and $\zeta_k = \zeta_{k-1}\cdot \ind(T_{k-1} > t_{k} - t_{k-1})$ for $k = 2,...,K$ denote the indicators for whether the individual survived to $t_k$.
Let $\zeta_1^c = 1$ and $\zeta_k^c = \zeta_{k-1}^c\cdot \ind(C_{k-1} > t_{k} - t_{k-1})$ for $k = 2,...,K$ denote the indicators that the subject hasn't experienced censoring yet till $t_k$.
The latent event time $T$ and censoring time $C$ are obtained by: 
\[
T = \sum_{k=1}^K \zeta_k \cdot \min(T_k, ~ {t_{k+1} - t_{k}}), \quad
C = \sum_{k=1}^K \zeta_k^c \cdot \min(C_k, ~{t_{k+1} - t_{k}}).
\]

{Note that \eqref{eq:sigma_T} and \eqref{eq:sigma_C} are used only for illustration, under which $T_k$ and $C_k$ depend only on the most recent covariates $Z_{t_k}$, resulting in a Markov data-generating mechanism. Our framework does not rely on the Markov assumption and applies more generally to settings where $T$ and $C$ may depend on the entire covariate history.}

Denote $\alpha_* = (\alpha_*^{(1)}, \ldots, \alpha_*^{(K)})$ and
$\beta_* = (\beta_*^{(1)}, \ldots, \beta_*^{(K)})$, where $*$ is a place holder. 
We take $K = 3$, $(t_2, t_3) = (3,6)$, $\rho = 0.3$, 
and 
\[
\alpha_T = (4,\,4,\,5), \quad
\beta_{T,0} = (-8,\,-8,\,-5), \quad
\beta_T = (1,\,2,\,3),
\]
\[
\alpha_C = (3,\,3,\,3), \quad
\beta_{C,0} = (-6,\,-6,\,-5), \quad
\beta_C = (2,\,2,\,2).
\]
The resulting censoring rate is approximately $36.2\%$.
The survival probabilities are approximately
$\PP(T>3)=95.8\%$ and $\PP(T>6)=86.3\%$, and the corresponding at-risk
probabilities are approximately
$\PP(X>3)=79.6\%$ and $\PP(X>6)=63.9\%$.


\subsubsection{DGM2}

DGM2 uses the same piecewise-constant time-varying covariate process as DGM1.
$[t_k,t_{k+1})$. 
In addition, we generate a baseline binary covariate
$B \sim \mathrm{Bernoulli}(0.5)$, independent of the time-varying covariate process.

Conditional on $(B,\bar Z_{t_k})$, the event time is generated
interval-by-interval from Weibull distributions
\[
T_k \sim \mathrm{Weibull}\{a_{T,k}, \sigma_{T,k}(B,Z_{t_k})\},
\]
where
\[
\sigma_{T,k}(B,Z_{t_k})
=
\exp\left\{
-\frac{
\beta_{T,0}^{(k)} + \beta_{TL}^{(k)} Z_{t_k} + \beta_{TB}^{(k)} B
}{a_{T,k}}
\right\}.
\]

The censoring mechanism is a subgroup mixture. For subjects with $B=0$, the
latent censoring time is generated as
\[
C \sim \mathrm{Unif}(0,15).
\]
For subjects with $B=1$, censoring is generated interval-by-interval from
Weibull waiting times
\[
C_k \sim \mathrm{Weibull}\{a_{C,k}, \sigma_{C,k}(Z_{t_k})\},
\]
where
\[
\sigma_{C,k}(Z_{t_k})
=
\exp\left\{
-\frac{
\beta_{C,0}^{(k)}
+ \beta_{CL}^{(k)} Z_{t_k}
+ \beta_{CL2}^{(k)} Z_{t_k}^2
}{a_{C,k}}
\right\}.
\]

We take
\[
a_T = (4,4,5), \quad
\beta_{T,0} = (-8,-8,-5), \quad
\beta_{TL} = (1,2,3), \quad
\beta_{TB} = (0.7,0.9,1.1),
\]
\[
a_C = (3,3,3), \quad
\beta_{C,0} = (-8.0,-7.5,-7.0), \quad
\beta_{CL} = (0.5,0.6,0.7), \quad
\beta_{CL2} = (0.4,0.5,0.6).
\]
The resulting censoring rate is approximately $32.1\%$.
The survival probabilities are approximately
$\PP(T>3)=93.9\%$ and $\PP(T>6)=81.1\%$, and the corresponding at-risk
probabilities are approximately
$\PP(X>3)=83.5\%$ and $\PP(X>6)=62.3\%$.

Under this data generating mechanism, the Cox model is misspecified for the censoring mechanism:
for $B=0$ the censoring time follows a uniform distribution, whereas for $B=1$ the
censoring hazard depends nonlinearly on $Z_{t_k}$ through both $Z_{t_k}$ and $Z_{t_k}^2$.

\subsubsection{DGM3}

DGM3 uses the same structure as DGM2 except with different parameters for generating the baseline covariate and the censoring time. 
The baseline covariate $B \sim \mathrm{Bernoulli}(0.4)$, and we take coefficients
\[
a_C = (1,1,1), \quad
\beta_{C,0} = (-2.8,-2.6,-2.4), \quad
\beta_{CL} = (-0.8,-1,-1.2), \quad
\beta_{CL2} = (0,0,0).
\]
The resulting censoring rate is approximately $53.7\%$.
The survival probabilities are approximately
$\PP(T>3)=94.2\%$ and $\PP(T>6)=82.1\%$, and the corresponding at-risk
probabilities are approximately
$\PP(X>3)=74.9\%$ and $\PP(X>6)=47.8\%$.

\newpage
\subsection{Details of the baseline IPCW method}\label{app:IPCW0}

The baseline IPCW method is an extension of the IPCW approach in \citet{farina2025doubly} to constructing two-sided prediction intervals and can also be viewed as a generalization of the weighting approach in \citet{yi2025survival} by allowing arbitrary conformity scores and flexible estimation of the censoring weights.
Algorithm \ref{alg:split_IPCW0} provides the details of this baseline.

\begin{algorithm}[h!]
\caption{Baseline IPCW method}
\label{alg:split_IPCW0}
\begin{algorithmic}
\STATE \textbf{Input:} observed data $\Dc = \{O_i\}_{i=1}^n$; coverage level $1-\alpha$; prediction model for the conditional distribution of $T|Z_0$; censoring model for the conditional distribution of $C|Z_0$; baseline covariates $Z_{0,n+1}$ from a new test individual.
\begin{enumerate}
    \item Randomly split $\Dc$ into a training set $\Dctr$ and a calibration set $\Dccal$ with almost equal size.
    
    \item Fit the prediction model on $\Dctr$, and let $\tilde q_\theta(Z_{0,n+1})$ denote the estimated conditional $\theta$-quantile function of $T_{n+1}$ given $Z_{0,n+1}$. Consider the following class of candidate intervals:  
    \begin{align*}
        \left\{ \big(\tilde q_{1-\theta}(Z_{0,n+1}),~ \tilde q_{\theta}(Z_{0,n+1}) \big): ~ \theta\in\Theta \right\}, 
    \end{align*}
    where $\Theta$ is taken as a grid on $[0.5,1]$ with increments 0.01.
    
    \item Fit the censoring model on $\Dctr$, and let $\hat S_c(t|z_0)$ denote the estimator for $\PP(C>t\mid Z_0 = z_0)$. \\
    
    \item Using a grid search over $\Theta$, compute 
    \begin{align}
        \tilde\theta
        = \inf\left\{ \theta\in\Theta: \frac{1}{|\Dccal|}\sum_{i \in \Dccal} \frac{\Delta_i \ind\left\{\tilde q_{1 - \theta}(Z_{0,i}) \leq T_i \leq \tilde q_{\theta}(Z_{0,i}) \right\} }{\hat S_c(X_i|Z_{0,i})}\geq 1-\alpha \right\}. \label{eq:theta_tilde_IPCW0}
    \end{align}
\end{enumerate}
\STATE \textbf{Output:} prediction interval $\left(\tilde q_{1-\tilde\theta}(Z_{0,n+1}), ~ \tilde q_{\tilde\theta}(Z_{0,n+1})\right)$.
\end{algorithmic}
\end{algorithm}

\subsection{Implementation details for simulation}\label{app:simu_prediction_model}

\subsubsection{Prediction models}\label{app:pred-models}

\paragraph{Landmarking prediction models.} At each prediction time $\tau$, prediction models for $T\mid \bar Z_\tau, T>\tau$ are fitted using the landmarking approach, in which we consider the following models:
\begin{itemize}
    \item `Cox': proportional hazards model fitted using the \texttt{coxph} function in the \texttt{survival} R package. 
    \item `RSF': random survival forests fitted using the \texttt{rfsrc} function in the \texttt{randomForestSRC} R package, with \texttt{ntree = 1000} and \texttt{nodesize = 15}. 
\end{itemize}
For a given landmarking time $\tau$, the above models are fitted with individuals that are at risk at time $\tau$, i.e., those with $X>\tau$. The covariates included in the proportional hazards model and random survival forests are different for different conformal methods: for HAPS and its two extensions, $Z_\tau$ is included; whereas for baseline methods, only baseline covariates $Z_0$ are included.


\paragraph{Neural-network prediction model.}
As an additional prediction model, we implemented a single-event neural-network dynamic prediction model inspired by Dynamic-DeepHit \citep{lee2019dynamic}, denoted by `DDH'. At each landmark time $\tau$, the model takes as input the observed covariate history up to $\tau$ for subjects still at risk. A gated recurrent unit (GRU) is used to encode the history into a hidden representation, and is then used to output the estimated discrete hazards.

The residual time from $\tau$ is discretized using quantiles of the observed residual follow-up times in the training set, with at most 40 time bins. The GRU representation is then used to predict discrete hazards in each bin, which is combined to output the survival curve estimate. 
The network is trained by minimizing the discrete-time right-censored survival negative log-likelihood. Unlike the original Dynamic-DeepHit implementation, we do not include the ranking loss or competing-risk architecture, since our setting involves a single event type and the model is used to estimate conditional quantiles for conformal prediction.

For the neural network structure and tuning parameters, we used a one-layer GRU with hidden size 64, dropout 0.1, learning rate $10^{-3}$, weight decay $10^{-5}$, batch size 64, 40 residual-time bins, and 50 training epochs. Training was performed on CPU. In degenerate landmark splits with fewer than 8 at-risk training subjects or no observed events, the implementation falls back to a Kaplan-Meier estimate.

\subsubsection{Nuisance models}

For the censoring survival function $G$, we consider the following working models:
\begin{itemize}
    \item Cox: a proportional hazards model implemented using the \texttt{coxph} function in the \texttt{survival} R package.
    \item RSF: a random survival forest implemented using the \texttt{rfsrc} function in the \texttt{randomForestSRC} R package, with \texttt{ntree = 1000},~ \texttt{nodesize = 30}, and \texttt{nsplit = 10}.
\end{itemize}

For the event-survival nuisance functions $S_k$, we consider the following working models:
\begin{itemize}
    \item Cox: a proportional hazards model implemented using the \texttt{coxph} function in the \texttt{survival} R package.
    \item RSF: a random survival forest implemented using \texttt{rfsrc} in the \texttt{randomForestSRC} R package, with \texttt{ntree = 1000}, ~\texttt{nodesize = 30}, and \texttt{nsplit = 10}.
    \item XGB: a gradient-boosted Cox model implemented using the \texttt{xgboost} R package with objective \texttt{survival:cox}, ~ \texttt{nrounds = 200}, ~ \texttt{eta = 0.05}, ~\texttt{max\_depth = 4}, ~\texttt{min\_child\_weight = 1}, ~\texttt{subsample = 0.8}, ~\texttt{colsample\_bytree = 0.8}, ~\texttt{lambda = 1}, and \texttt{gamma = 0}.
\end{itemize}

For the augmentation nuisance functions $\xi_k$, we use the following working model:
\begin{itemize}
    \item XGB: a gradient-boosted multiclass classifier implemented using the \texttt{xgboost} R package with objective \texttt{multi:softprob}. The three classes correspond to whether the pseudo-outcome is zero, negative, or positive. We use \texttt{num\_class = 3}, ~ \texttt{eval\_metric = mlogloss}, ~\texttt{nrounds = 200}, ~\texttt{eta = 0.05}, ~\texttt{max\_depth = 4}, \texttt{min\_child\_weight = 1}, ~\texttt{subsample = 0.8}, ~\texttt{colsample\_bytree = 0.8}, ~\texttt{lambda = 1}, and \texttt{gamma = 0}. Then the conditional expectation $\xi_k$ is computed with the estimated probabilities for the three classes.
\end{itemize}

For HAPS and HAPS-DR, the censoring survival function $G$ is estimated using the global approach described in Section~\ref{app:G_est} when the model for $G$ is Cox. When the model for $G$ is RSF, we use the piecewise estimation approach described in Section~\ref{app:G_est_piecewise}. For HAPS-A, we use the piecewise estimation approach in Section~\ref{app:G_est_piecewise} for estimating $G$ under both Cox and RSF models.

\subsection{Additional details and simulation results for the two simulation setups in the main paper}

\paragraph{Simulation setups used in the main paper.}
The simulation figures in the main paper report two representative settings, summarized in Table~\ref{tab:main_sim_estimation_models}. 

In Setup A, data are generated from DGM1, and a correctly specified Cox model is used to estimate $G$. This represents a setting in which a semiparametric model is used to estimate the censoring distribution and the model is correctly specified.

In Setup B, data are generated from DGM2, where the conditional censoring distribution is more complex, and flexible machine learning methods are used to estimate the censoring distribution.

\begin{table}[!ht]
\centering
\caption{DGMs and estimation models used in the two main simulation settings.}
\label{tab:main_sim_estimation_models}
{\fontsize{9}{10}\selectfont
\begin{tabular}{lcccccc}
\toprule
Setups & DGM & Prediction model & $G$ & $S_k$ & $\xi_k$ & Baseline methods \\
\midrule
Setup A & DGM1 & Cox & Cox & Cox & XGB & Cox prediction, Cox model for $G$ \\
Setup B & DGM2 & DDH & RSF & XGB & XGB & RSF prediction, RSF model for $G$ \\
\bottomrule
\end{tabular}
}
\end{table}

\paragraph{Additional results on upper and lower endpoints.}
Figure~\ref{fig:simu_setupA_4panel} and Figure~\ref{fig:simu_setupB_4panel} summarize the 90\% prediction intervals in simulation Setups A and B, including the lower and upper endpoints of the intervals. We observe that, for HAPS and its extensions, the lower endpoints tend to increase and the upper endpoints tend to decrease at later prediction times. This suggests that the prediction intervals become more informative over time as additional covariate history is incorporated.

\begin{figure}[ht!]
    \centering

    \begin{subfigure}{0.49\linewidth}
        \centering
        \includegraphics[width=\linewidth]{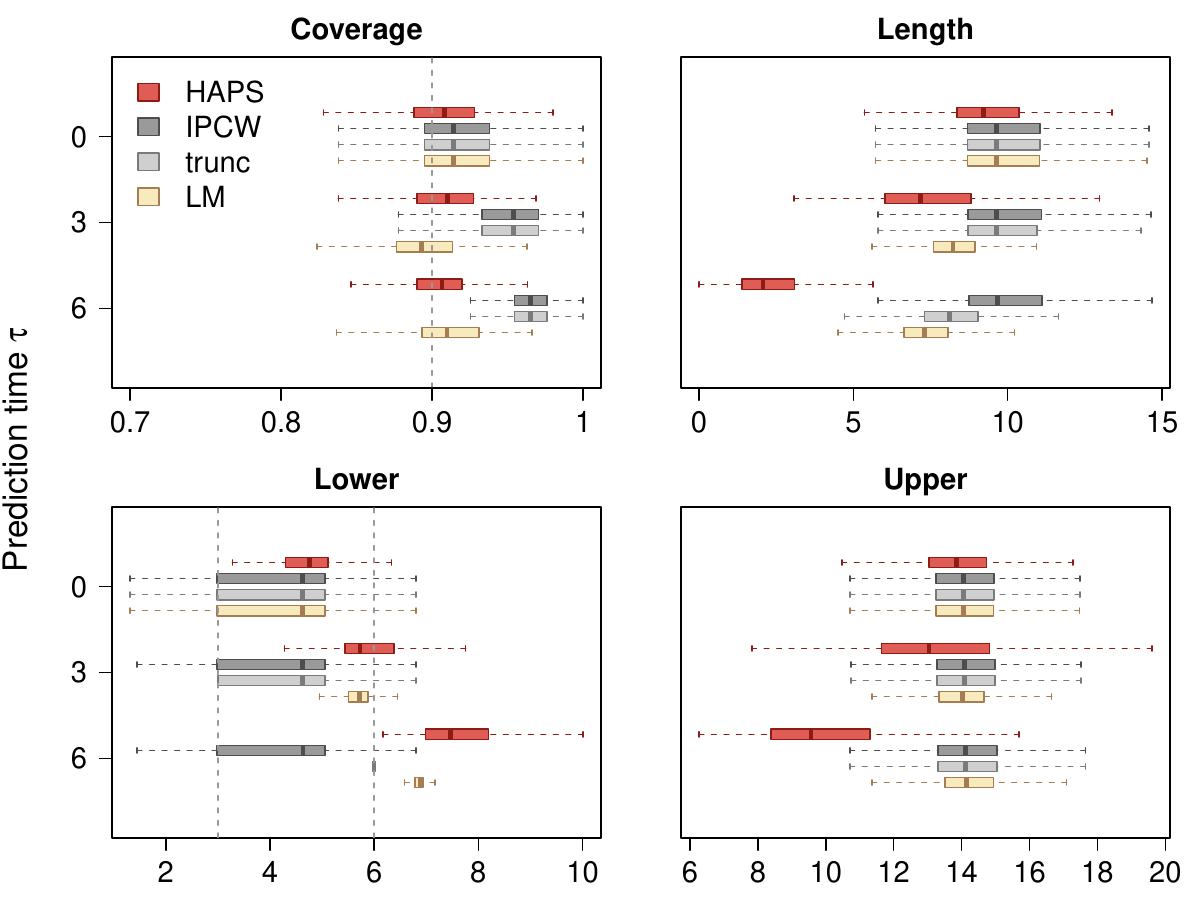}
        \caption{HAPS versus baseline methods.}
    \end{subfigure}
    \hfill
    \begin{subfigure}{0.49\linewidth}
        \centering
        \includegraphics[width=\linewidth]{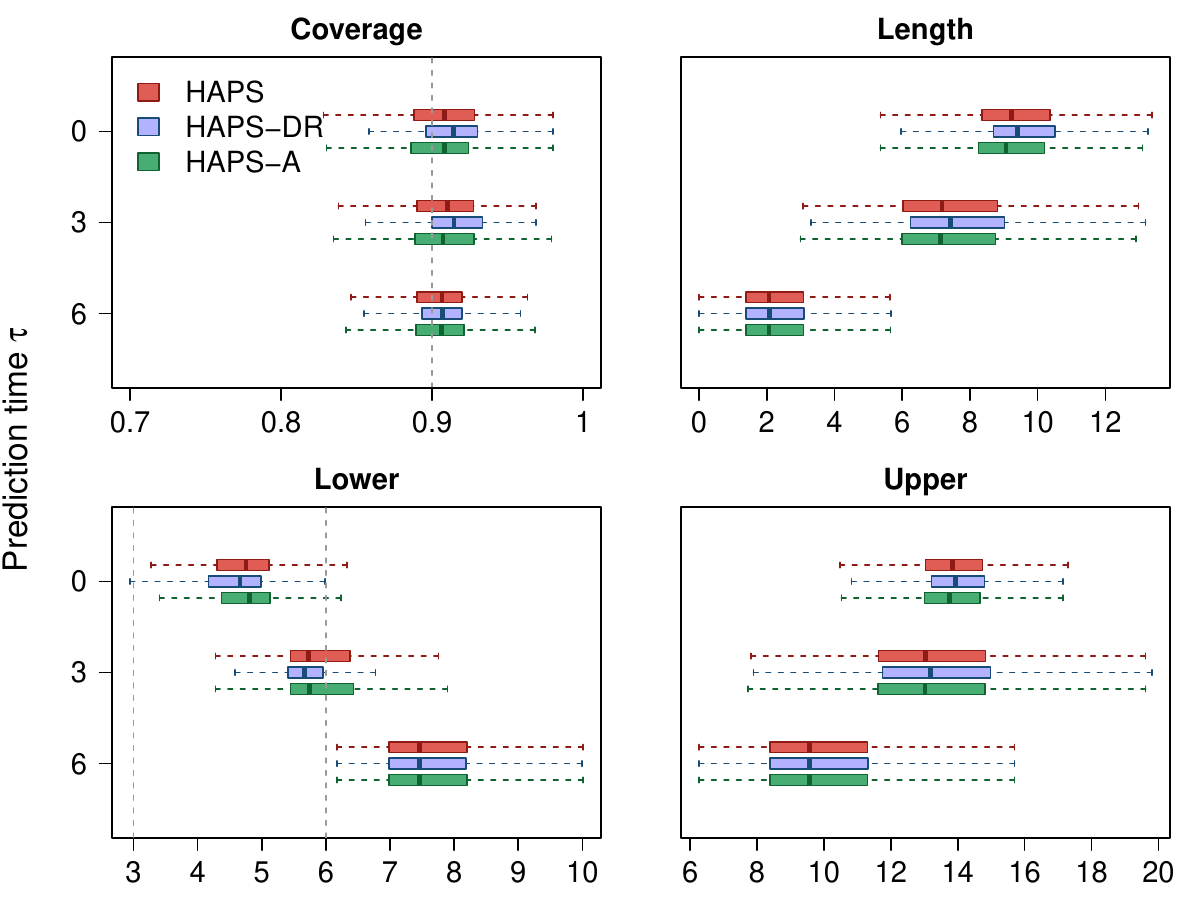}
        \caption{HAPS and its two extensions.}
    
    \end{subfigure}


    \caption{Summaries for the 90\% prediction intervals in Setup A of the main paper.}
    \label{fig:simu_setupA_4panel}
\end{figure}

\begin{figure}[h!]
    \centering

    \begin{subfigure}{0.49\linewidth}
        \centering
        \includegraphics[width=\linewidth]{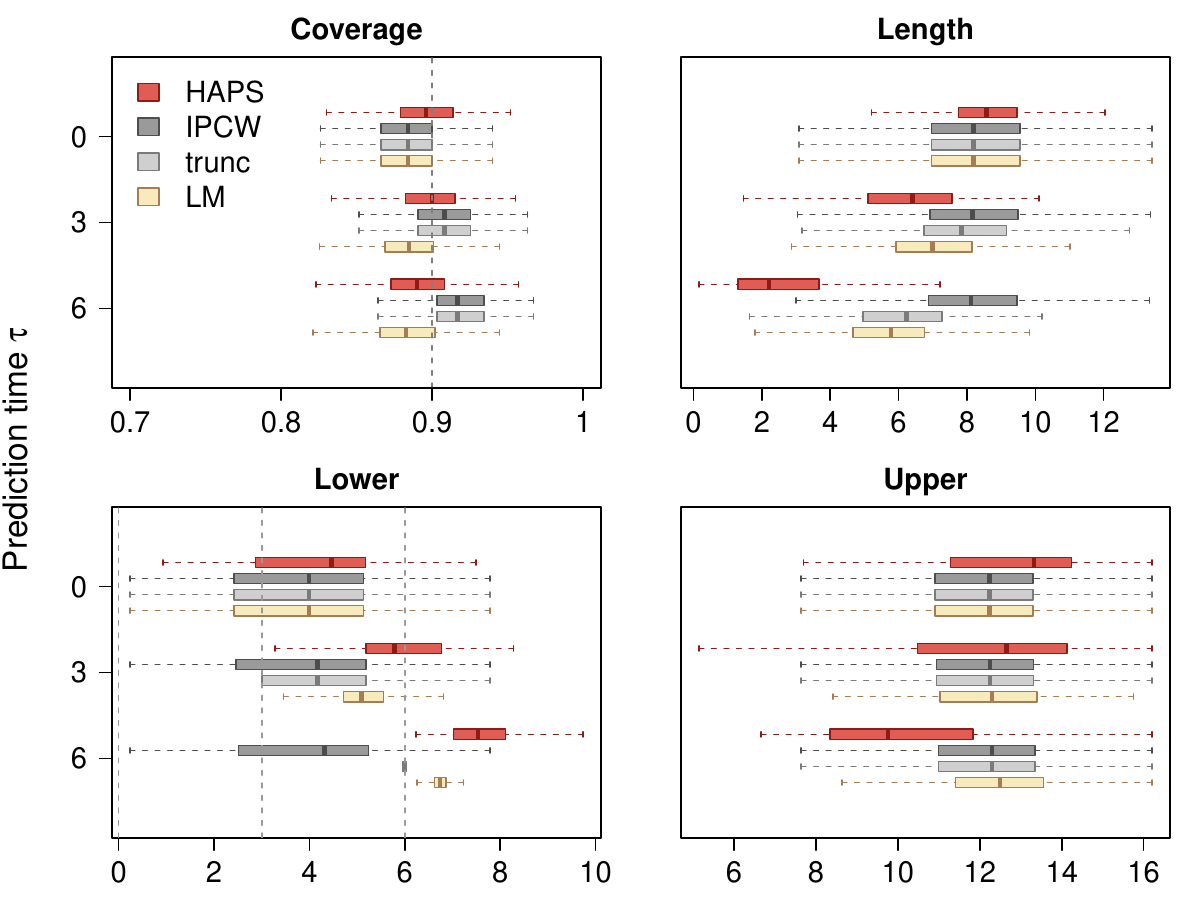}
        \caption{HAPS versus baseline methods.}
    \end{subfigure}
    \hfill
    \begin{subfigure}{0.49\linewidth}
        \centering
        \includegraphics[width=\linewidth]{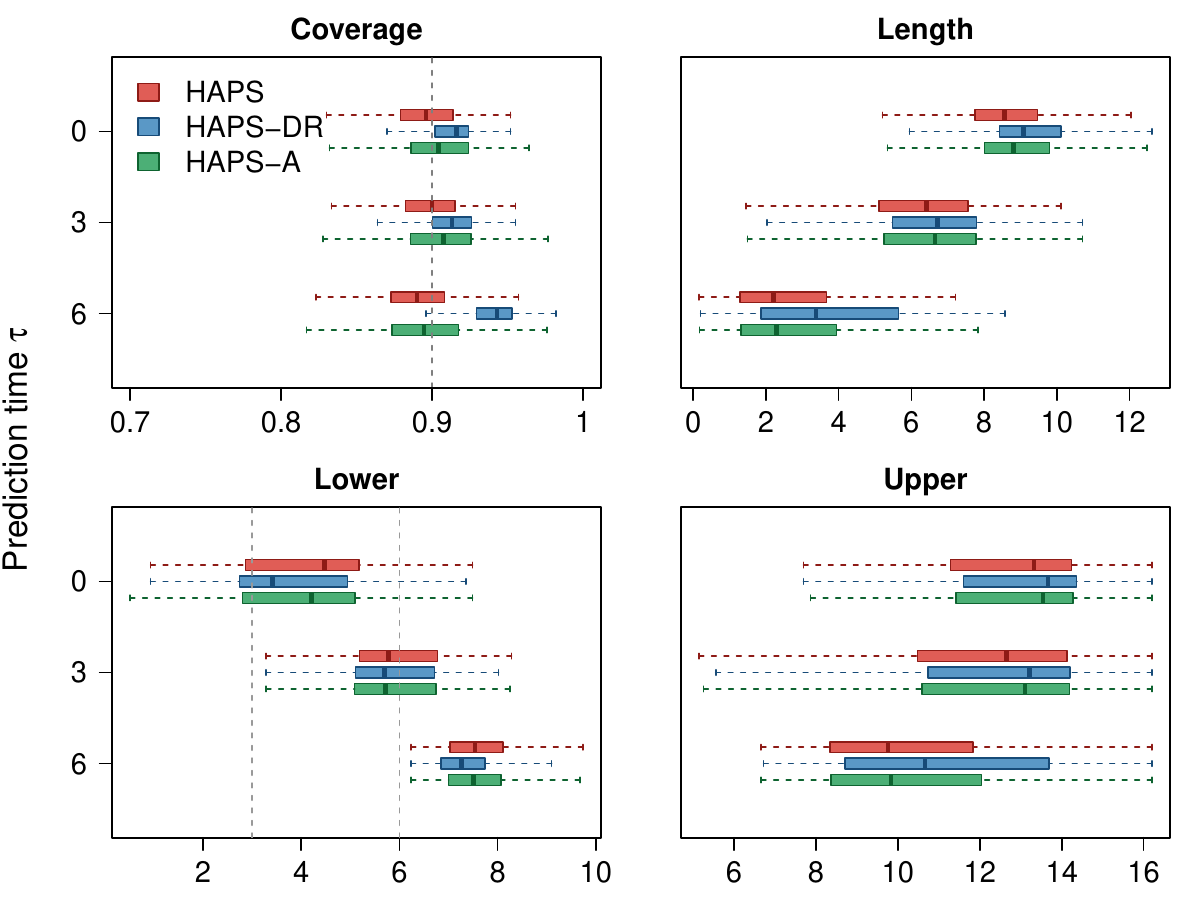}
        \caption{HAPS and its two extensions.}
    
    \end{subfigure}
    
    \caption{Summaries for the 90\% prediction intervals in Setup B of the main paper.}
    \label{fig:simu_setupB_4panel}
\end{figure}

\newpage
\subsection{Assessing robustness to censoring estimation errors} \label{app:HAPS_robustness}

Additional simulation studies are performed to investigate the robustness of HAPS and its extensions (HAPS-DR, HAPS-A) to censoring estimation errors. 

Appendix \ref{sec:simu_robustness_cen_misspecification} studies the performance of HAPS and its extensions when misspecified semiparametric models are used to estimate $G$, and Appendix \ref{sec:simu_robustness_cen_nonparametric}  studies the performance of HAPS and its extensions when flexible machine learning methods are used to estimate the censoring distribution.

\subsubsection{Under censoring model misspecification}\label{sec:simu_robustness_cen_misspecification}

Figure~\ref{fig:simu_cen_misspcification} summarizes the prediction intervals from HAPS and its two extensions (HAPS-DR, HAPS-A) under DGM2 and DGM3 when Cox model is used to estimate the censoring distribution. In both DGMs, censoring follows a mixture distribution, so a Cox model is misspecified for the conditional censoring distribution.

The results show that HAPS noticeably undercovers in DGM3, suggesting that it can be sensitive to censoring model misspecification when a parametric or semiparametric censoring model is used. In contrast, the two extensions, HAPS-A and HAPS-DR, improve coverage, indicating  robustness to censoring model misspecification. We also observe that HAPS-DR tends to overcover and produces wider intervals than HAPS and HAPS-A, which is expected from its post-processing construction.

\begin{figure}[ht!]
    \centering
    \begin{subfigure}{0.49\linewidth}
        \centering
    \includegraphics[width=\linewidth]{figures/mixC2_DGM1/h_HAPS_versions_ddh_rsf_Sxgb_cox_Xixgb_multiclass_tunewide64_epochs50_rho0.3_R200_n1000_ntest500_alpha0.1_loc0_Gtauone_d0.pdf}
    \caption{DGM2}
    \end{subfigure}
    \hfill
    \begin{subfigure}{0.49\linewidth}
        \centering
    \includegraphics[width=\linewidth]{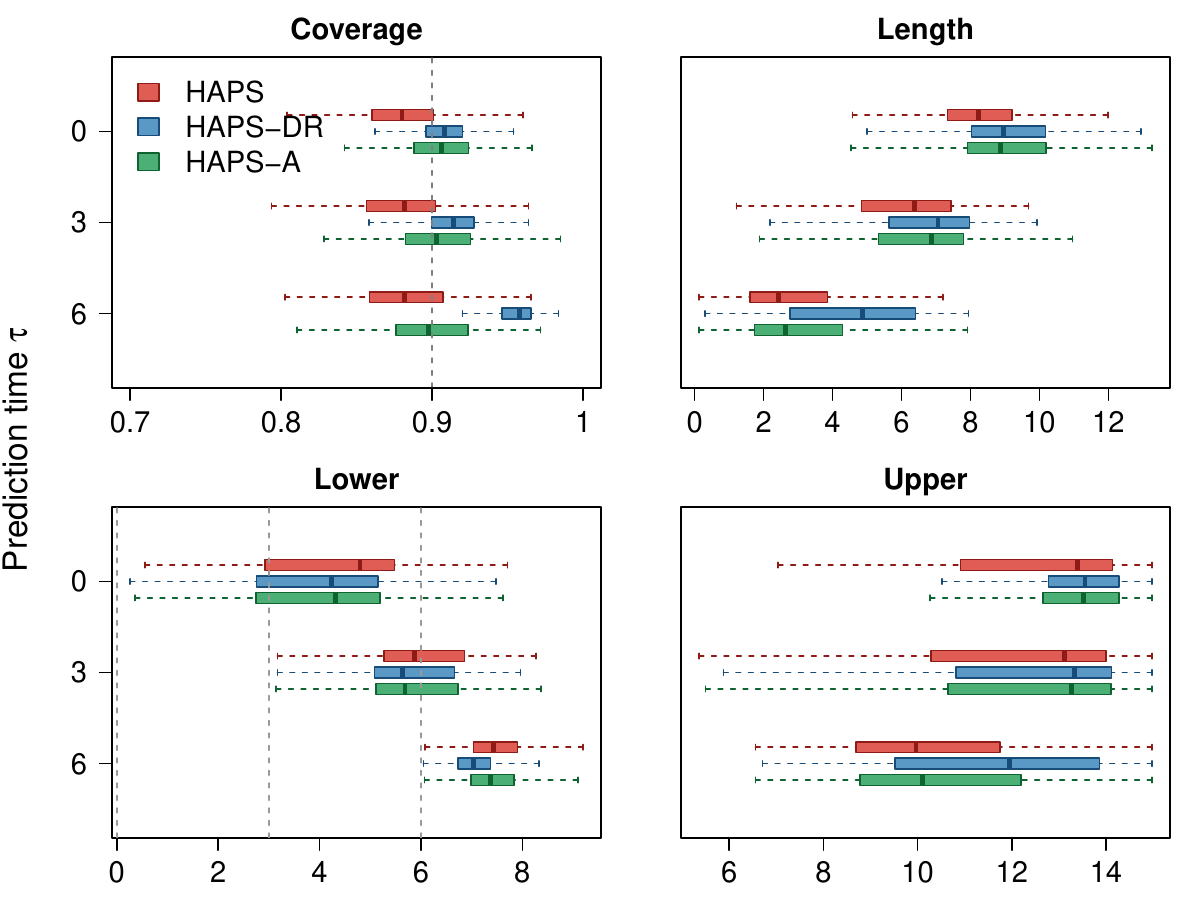}
    \caption{DGM3}
    \end{subfigure}
    
    \caption{Summary of 90\% prediction intervals with DDH prediction model and a misspecified Cox model for censoring. HAPS uses DDH for the prediction model, Cox for the censoring model; HAPS-A uses XGB to estimate additional nuisance estimators $S_k$ and $\xi_k$ in the augmentation term.}
    \label{fig:simu_cen_misspcification}
\end{figure}

\newpage
\subsubsection{Robustness when using machine learning methods to estimate nuisance parameters}\label{sec:simu_robustness_cen_nonparametric} 

Figure \ref{fig:simu_G_nonparametric_versus_n} summarizes the 80\% prediction intervals from HAPS and its two extensions (HAPS-DR, HAPS-A) that use RSF for the prediction model, RSF for the censoring model, and XGB for the additional augmentation nuisance parameters $S_k$ and $\xi_k$ for HAPS-A, under DGM1-DGM3 and various sample sizes. 

We observe that all three methods achieve close to nominal coverage in large sample sizes, reflecting the robustness of HAPS and its extensions when using flexible machine learning methods to estimate the nuisance parameters.

\begin{figure}[ht!]
    \centering
    \includegraphics[width=0.8\linewidth]{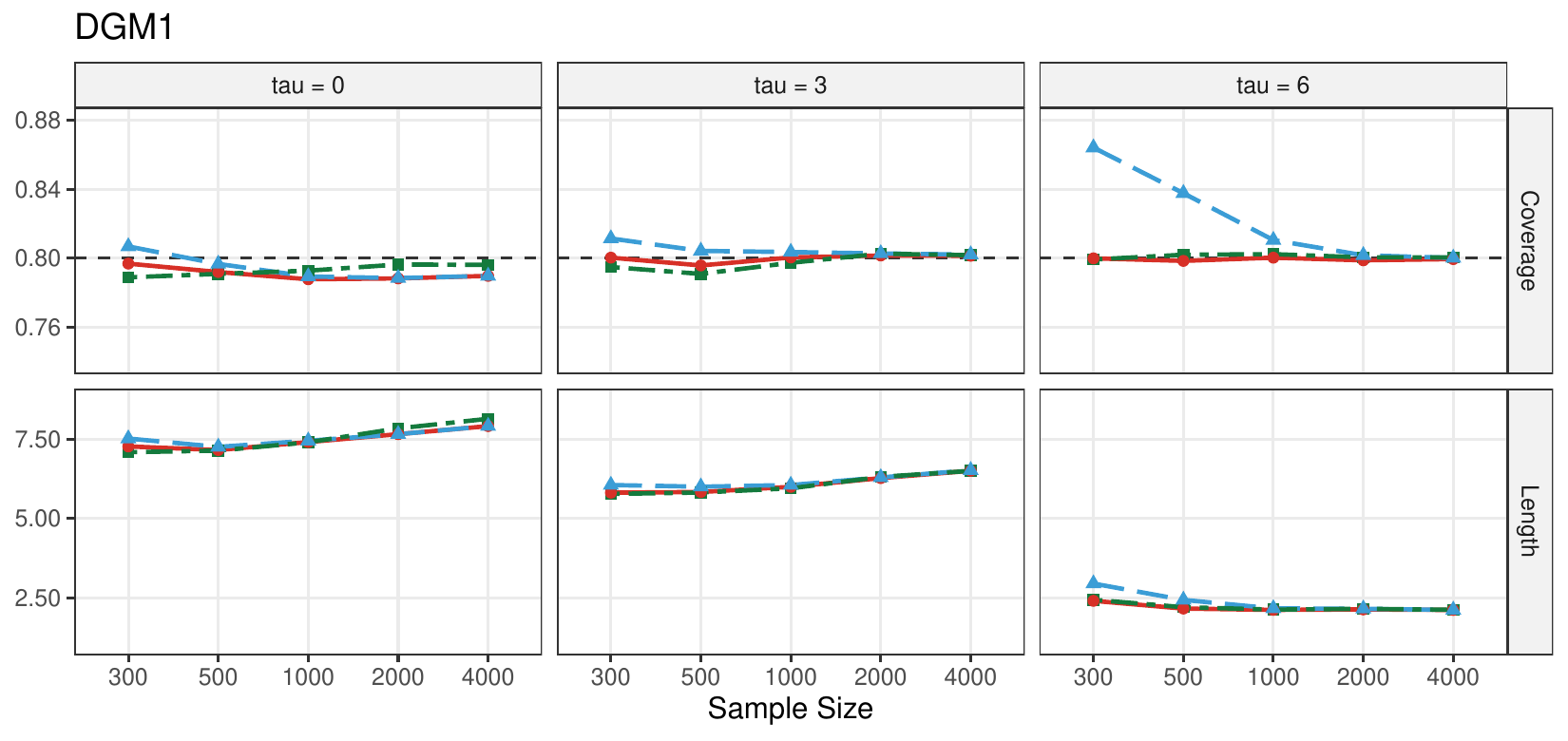}
    \includegraphics[width=0.8\linewidth]{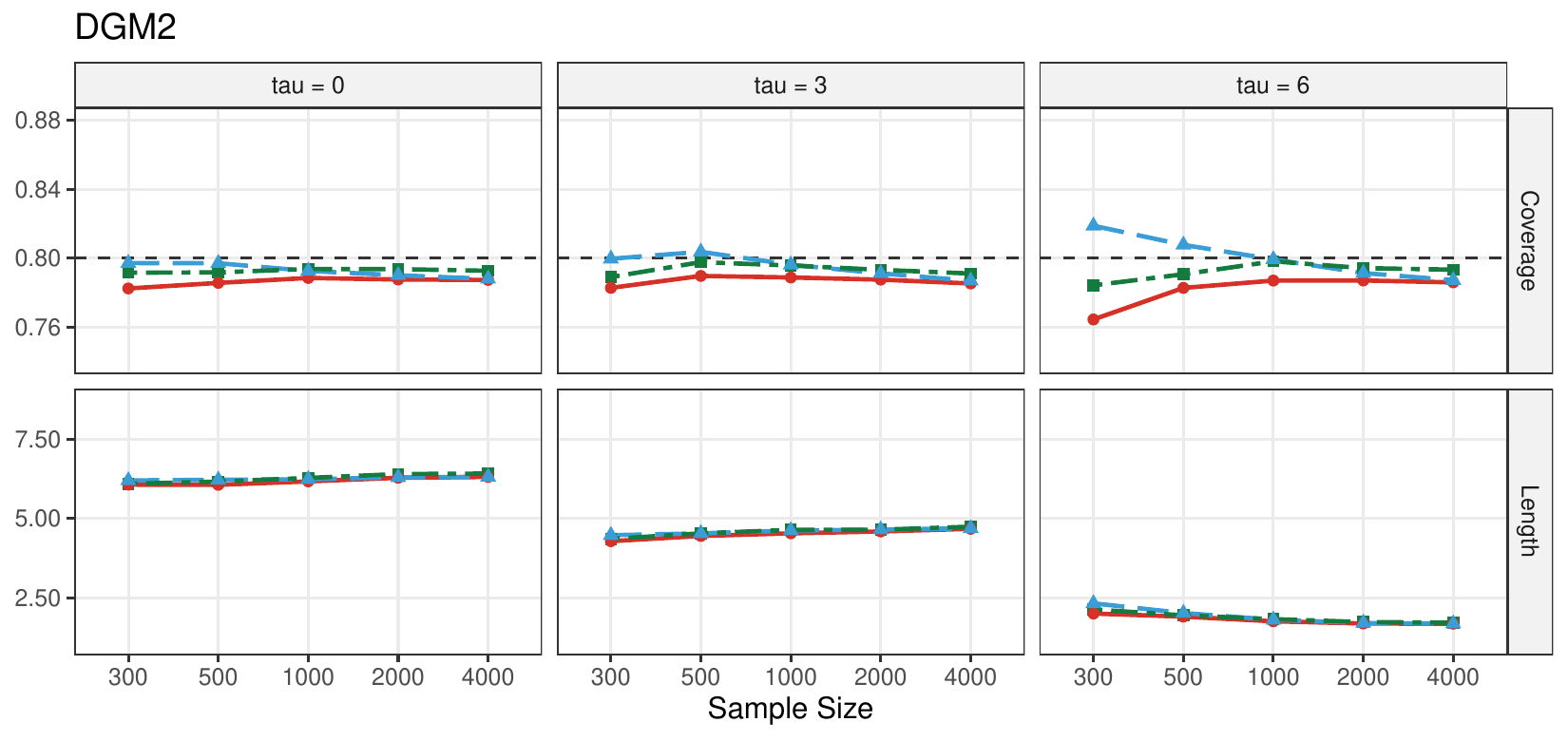}
    \includegraphics[width=0.8\linewidth]{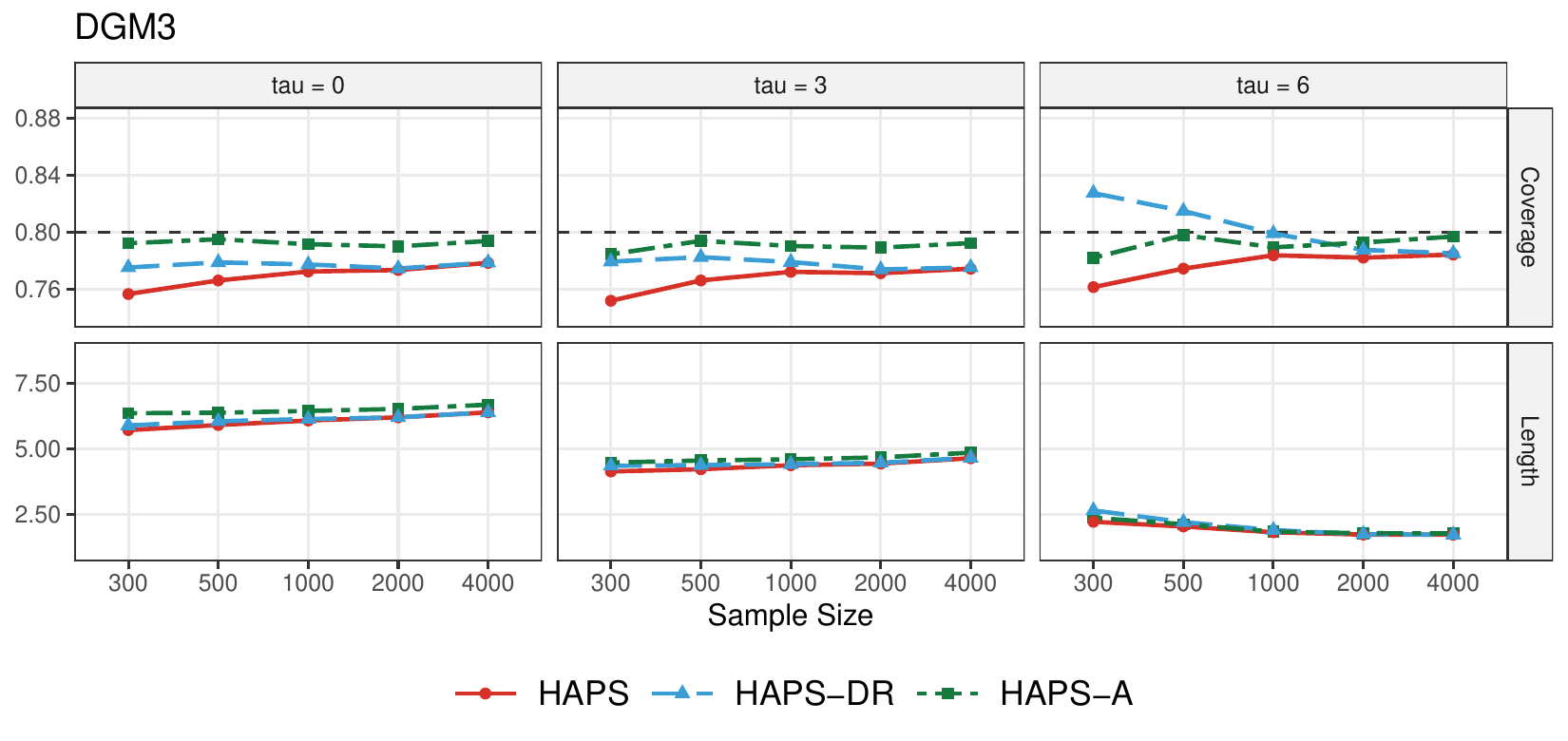}
    \caption{Mean coverage and median interval length of 80\% prediction intervals from HAPS, HAPS-DR, and HAPS-A across sample sizes for DGM1--DGM3, across 200 Monte Carlo replications. 
    The dashed horizontal line marks the nominal 80\% coverage level.}
    \label{fig:simu_G_nonparametric_versus_n}
\end{figure}

\clearpage
\subsection{Impact of prediction model accuracy on prediction intervals}\label{app:simu_impact_pred_accuracy} 

In this section, we study the impact of prediction model accuracy on prediction intervals from HAPS and its two extensions, HAPS-DR and HAPS-A. Specifically, we compare the prediction intervals obtained using the DDH prediction model described in Appendix~\ref{app:pred-models} with those obtained using a smaller DDH model with hidden size 32 and 20 training epochs.

Figure~\ref{fig:simu_impact_pred} summarizes the 90\% prediction intervals under DGM2, where an RSF model is used to estimate the censoring distribution and XGB models are used to estimate the additional augmentation nuisance parameters $S_k$ and $\xi_k$.

We observe that in both settings, the prediction intervals achieve close to nominal or above nominal coverage. However, when the smaller DDH model is used, the prediction intervals tend to be wider, especially at $\tau = 3$ and $\tau = 6$, reflecting the poorer prediction performance of this model. We also observe that, compared with HAPS and HAPS-A, the overcoverage of HAPS-DR becomes more pronounced when the smaller DDH model is used; and in this case, the prediction intervals from HAPS-DR are substantially wider than those from HAPS and HAPS-A.

\begin{figure}[h!]
    \centering

    \begin{subfigure}{1\linewidth}
        \centering
        \includegraphics[width=0.55\linewidth]{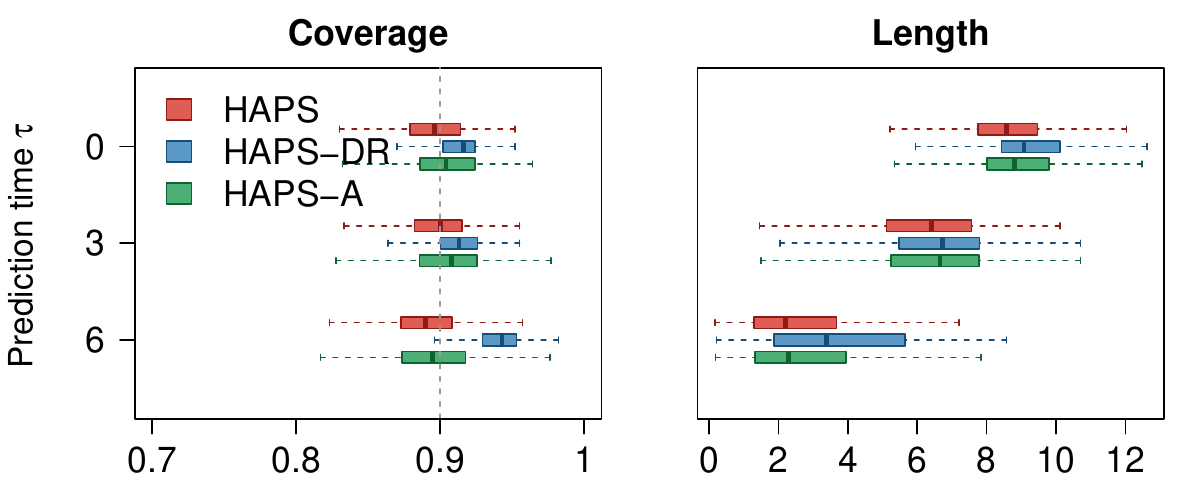}
        \caption{DDH model described in Section \ref{app:pred-models}.
        with hidden size 64 and 50 training epochs.
        }
    \end{subfigure}

    \vspace{1em}
    \begin{subfigure}{1\linewidth}
        \centering
        \includegraphics[width=0.55\linewidth]{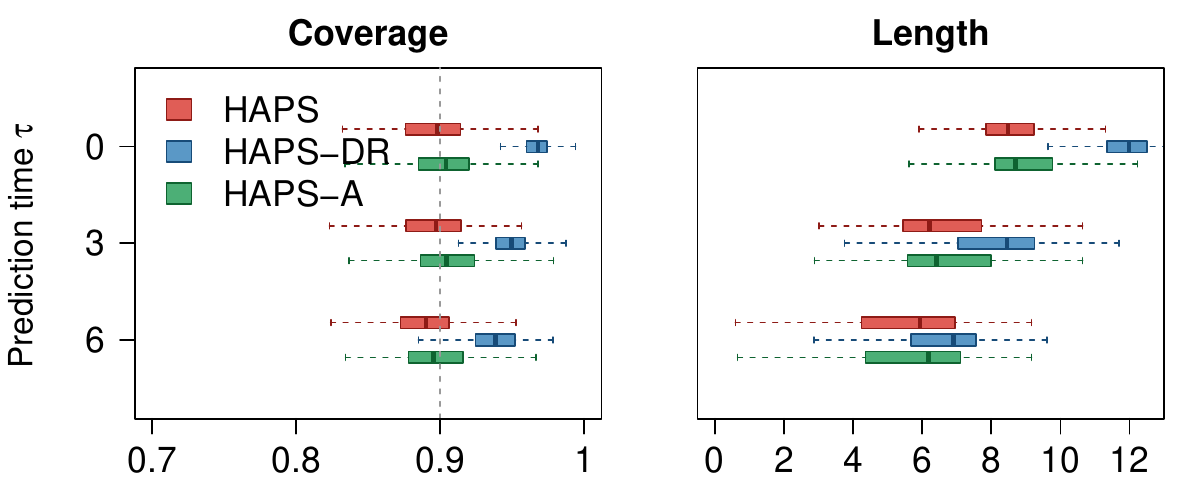}
        \caption{Smaller DDH model with hidden size 32 and  20 training epochs.}
    \end{subfigure}

    \caption{Summary of 90\% prediction intervals for HAPS and its two extensions in DGM2 under two DDH prediction model tuning choices. All nuisance models are held fixed: RSF is used to estimate the censoring distribution, and XGB models are used to estimate additional augmentation nuisance parameters $S_k$ and $\xi_k$.}
    \label{fig:simu_impact_pred}
\end{figure}

\FloatBarrier
\subsection{Additional baselines}\label{app:add-baselines}

To disentangle the two roles of covariate history in the empirical improvements of HAPS, we consider three additional baselines:
\begin{itemize}
    \item LM-$Z_\tau$: fits the same prediction model as HAPS but applies static IPCW calibration on the landmark sample at time $\tau$, using only baseline covariates and time-varying covariates measured at $\tau$ to handle censoring.
    \item HAPS-$Z_0$: uses the covariate history for handling censoring in calibration (same as HAPS) and uses only $Z_0$ to fit the prediction model.
    \item Quantile: the interval formed by the $\alpha/2$ and $(1-\alpha/2)$ quantiles from the fitted prediction model (uncalibrated).
\end{itemize}
The LM-$Z_\tau$ method uses the same prediction model as HAPS, so any difference between the performance of HAPS and LM-$Z_\tau$ is due to using time-varying covariates for censoring adjustments. The HAPS-$Z_0$ method uses only baseline covariates for fitting the prediction model and uses the same censoring adjustments as HAPS, so any difference between HAPS and HAPS-$Z_0$ reflects how including time-varying covariates in the prediction model affects performance.

Table~\ref{tab:add-baselines} shows the coverage (Cov) and median interval length (Len) under the two simulation setups in Section~\ref{sec:simulation}, where the nominal coverage probability is 0.9. We observe that HAPS and HAPS-$Z_0$ achieve close to the nominal 90\% coverage for all $\tau$ and in both Setup A and Setup B, while LM-$Z_\tau$ over-covers at $\tau = 3$ in Setup A, indicating that using covariate history for handling censoring in calibration ensures close to nominal coverage. In addition, the median interval length from HAPS-$Z_0$ is greater than that from HAPS, indicating that incorporating covariate history in the prediction model leads to shorter intervals. LM-$Z_\tau$ shows close to 90\% coverage at $\tau = 6$ because, under the data generating mechanisms in both setups, $Z_\tau$ can properly address censoring in the last time interval. Finally, we observe that the uncalibrated quantile intervals can have large coverage gaps, which shows the importance of conformal calibration.

For Setup B, we also observe that HAPS-$Z_0$ at $\tau = 3$ and $\tau = 6$ and HAPS at $\tau = 6$ appear to be slightly undercovered, possibly due to the slow convergence rate of the RSF model for censoring. The results for HAPS in Setup B also agree with Theorem~\ref{thm:IPCW_coverage}, in that the coverage gap can be larger at later $\tau$, where $\sttau$ tends to be smaller.

\begin{table}[ht]
\centering
\small
\caption{Simulation results for Setups A and B with additional baselines at different $\tau$. Cov: mean coverage (Monte Carlo SE); Len: median interval length. The nominal coverage probability is 0.9.}
\label{tab:add-baselines}
\begin{tabular}{lcccc}
\toprule
 & \multicolumn{2}{c}{$\tau = 3$} & \multicolumn{2}{c}{$\tau = 6$} \\
\cmidrule(lr){2-3} \cmidrule(lr){4-5}
Method & Cov (SE) & Len & Cov (SE) & Len \\
\midrule
\multicolumn{5}{l}{\textit{Setup A}} \\
HAPS           & .899 (.0020) & 7.31 & .899 (.0017) & 2.35 \\
LM-$Z_\tau$    & \textbf{.917} (.0024) & 7.81 & .903 (.0018) & 2.37 \\
HAPS-$Z_0$     & .895 (.0022) & \textbf{8.32} & .900 (.0019) & \textbf{7.04} \\
Quantile (Cox) & .894 (.0018) & 7.24 & \textbf{.883} (.0016) & 2.25 \\
\midrule
\multicolumn{5}{l}{\textit{Setup B}} \\
HAPS           & .891 (.0018) & 6.17 & \textbf{.885} (.0020) & 2.57 \\
LM-$Z_\tau$    & \textbf{.887} (.0018) & 6.10 & \textbf{.886} (.0020) & 2.58 \\
HAPS-$Z_0$     & \textbf{.883} (.0021) & \textbf{7.00} & \textbf{.885} (.0021) & \textbf{5.61} \\
Quantile (DDH) & .905 (.0014) & 6.50 & \textbf{.940} (.0013) & \textbf{3.67} \\
\bottomrule
\end{tabular}
\end{table}

\FloatBarrier
\subsection{Simulations under smaller sample sizes}\label{app:small-n}

In both data applications, 20\% of individuals have been used as test data. The training and calibration sample size is 250 (125 training / 125 calibration) for the PBC data example and 744 (372 training / 372 calibration) for the colon cancer data example.

Table~\ref{tab:small-n} reports additional simulation results under simulation Setup A with sample sizes $n = 250$ and $n = 750$ to mimic the sample sizes in the two applications. We report the mean coverage (Cov) and its SE across 200 Monte Carlo replications and the median length of the prediction intervals (Len). We also report the rate (Triv) at which the estimating inequality does not admit a solution, which can happen especially when the sample size is small. In such cases, the algorithm returns a trivial interval $(\tau, 50)$, where 50 is a large enough upper bound for $T$ under this paper's data-generating mechanisms.

We observe that under the smaller sample sizes $n = 250$ and $n = 750$, HAPS still achieves close to nominal coverage, and its median interval length is much smaller than that of baselines using only baseline covariates. The conclusions match those of the simulation in Section~\ref{sec:simulation} (sample size 1000; 500 training + 500 calibration).

\begin{table}[ht]
\centering
\small
\caption{Simulation results under Setup A with smaller sample sizes: $n = 250$ (125 training / 125 calibration, same size as the PBC application) and $n = 750$ (375 training / 375 calibration, similar to the size of the colon cancer application). Cov: mean coverage (SE); Len: median interval length; Triv: rate of trivial intervals; --: not applicable.}
\label{tab:small-n}
\begin{tabular}{lcccccc}
\toprule
 & \multicolumn{3}{c}{$\tau = 3$} & \multicolumn{3}{c}{$\tau = 6$} \\
\cmidrule(lr){2-4} \cmidrule(lr){5-7}
Method & Cov (SE) & Len & Triv & Cov (SE) & Len & Triv \\
\midrule
\multicolumn{7}{l}{\textit{$n = 250$}} \\
HAPS           & .899 (.0038) & 7.69 & .060 & .901 (.0036) & 2.55 & .095 \\
HAPS-A         & .890 (.0041) & 7.40 & .065 & .886 (.0055) & 2.50 & .065 \\
HAPS-DR        & \textbf{.919} (.0028) & 8.08 & .060   & .909 (.0028) & 2.60 & .095-   \\
HAPS-$Z_0$     & .892 (.0039) & 8.54 & .055 & .903 (.0034) & \textbf{7.25} & .050 \\
LM-$Z_\tau$    & .909 (.0036) & 8.02 & .080 & .900 (.0038) & 2.54 & .120 \\
LM             & \textbf{.884} (.0033) & 8.23 & .035 & .901 (.0031) & \textbf{7.22} & .050 \\
IPCW           & \textbf{.938} (.0029) & \textbf{9.86} & .046 & \textbf{.959} (.0017) & \textbf{9.91} & .046 \\
trunc          & \textbf{.938} (.0029) & \textbf{9.79} & .046 & \textbf{.959} (.0017) & \textbf{8.14} & .046 \\
Quantile (Cox) & \textbf{.881} (.0028) & 7.10 & --   & \textbf{.856} (.0029) & 2.24 & --   \\
\midrule
\multicolumn{7}{l}{\textit{$n = 750$}} \\
HAPS           & .899 (.0022) & 7.41 & .000 & .898 (.0021) & 2.37 & .005 \\
HAPS-A         & .898 (.0022) & 7.39 & .000 & .898 (.0021) & 2.37 & .000 \\
HAPS-DR        & \textbf{.915} (.0019) & 7.90 & .000   & .909 (.0018) & 2.44 & .005   \\
HAPS-$Z_0$     & .895 (.0025) & \textbf{8.48} & .000 & .900 (.0022) & \textbf{7.15} & .000 \\
LM-$Z_\tau$    & \textbf{.914} (.0026) & 7.83 & .015 & .899 (.0017) & 2.38 & .000 \\
LM             & \textbf{.884} (.0022) & 8.03 & .000 & .904 (.0020) & \textbf{7.13} & .000 \\
IPCW           & \textbf{.949} (.0020) & \textbf{9.74} & .000 & \textbf{.963} (.0012) & \textbf{9.76} & .000 \\
trunc          & \textbf{.949} (.0020) & \textbf{9.71} & .000 & \textbf{.963} (.0012) & \textbf{8.17} & .000 \\
Quantile (Cox) & .893 (.0018) & 7.25 & --   & \textbf{.879} (.0018) & 2.25 & --   \\
\bottomrule
\end{tabular}
\end{table}

\clearpage
\section{Additional details and results for application}\label{app:application}

\subsection{Details of the data sets}\label{app:dataset_details}

\paragraph{Colon cancer data.}
We analyze data from one of the first successful trials of adjuvant chemotherapy for colon cancer \citep{laurie1989surgical, moertel1990levamisole}, available as the \texttt{colon} data set in the \texttt{survival} R package and consisting of 929 colon cancer patients. We focus on predicting overall survival, i.e., time to death. The censoring rate is around 51.3\%.
We include the following baseline covariates: age (years), sex (male/female), treatment group (observation, Levamisole, or Levamisole+5-FU), and number of positive lymph nodes (more than 4 vs.\ not); as well as a time-varying covariate for recurrence, recorded as a 0–1 process indicating whether recurrence has occurred by each time.

\paragraph{PBC data.}
We analyze data from the Mayo Clinic trial in primary biliary cholangitis (PBC) conducted between 1974 and 1984, available as the \texttt{pbcseq} data set in \texttt{survival} R package and consisting of 312 patients with longitudinal laboratory measurements in addition to baseline covariates. Again we focus on predicting overall survival. The censoring rate is approximately 55.1\%.
Following \citet{murtaugh1994primary}, we consider the following (transformed) covariates: age, $\log$(bilirubin), $\log$(prothrombin time), $\log$(albumin), and edema score (0/0.5/1). Among these, bilirubin, prothrombin time, albumin, and edema score are time-varying.

\subsection{IPCW estimator for the survivor-conditional coverage} \label{app:justification_IPCW_coverage}

As mentioned before, the event times in real data may be censored. As a result, the empirical survivor-conditional coverage cannot be computed directly on test data. We therefore consider its IPCW estimator in \eqref{eq:beta_hat} below. Such an estimator for coverage probability has also been reported in \citet{farina2025doubly} in static conformal settings. 
It is a consistent estimator of the coverage probability when the censoring model is correctly specified. 
The details of the estimator are provided below.

For a given prediction set $\Cc(\bar Z_\tau)$, Let $\beta = \PP\{T\in\Cc(\bar Z_\tau)\mid T>\tau\}$ denote the coverage probability among survivors at time $\tau$. 
Recall from Section \ref{sec:method_C} the definition of $G$. 
With an estimator of $G$, denoted by $\hat G$, and a test sample indexed by $\Ic_{\text{test}}$, the IPCW estimator for $\beta$ is 
\begin{align}
    \hat\beta = \frac{\sum_{i\in \Ic_{\text{test}}, X_i>\tau} \ind\{X_i\in \hat\Cc(\bar Z_\tau)\}\cdot \Delta_i / \hat G(X_i|\bar Z_{X_i, i})}{\sum_{i\in \Ic_{\text{test}}, X_i>\tau} \Delta_i / \hat G(X_i|\bar Z_{X_i, i})}. \label{eq:beta_hat}
\end{align}
The justification of the estimator is provided below.

\paragraph{Justification.}
Let 
\begin{align*}
    D^*(\beta) = \ind(T>\tau)\left[  \ind\{T\in\Cc(\bar Z_\tau)\} - \beta \right]
\end{align*}
By definition, we have that $D^*(\beta)$ is an estimating function for $\beta$ in the sense that $\E\{D^*(\beta)\} = 0$. 
Under right censoring, $T$ and $\bar Z_\tau$ are not always completely observed. So we apply IPCW \citep{robins1992recovery, rotnitzky2005inverse}, resulting in the following estimating function for $\beta$ in the observed data:
\begin{align}
    D(\beta;G) = \frac{\Delta}{G(X|\bar Z_X)} \cdot \ind(X>\tau)\left[  \ind\{X\in\Cc(\bar Z_\tau)\} - \beta \right]. \label{eq:U(beta)}
\end{align}
Therefore, with a random test sample indexed by $\Ic_{\text{test}}$, if we have an estimator of $G$, denoted by $\hat G$, then $\beta$ can be estimated by solving the following estimating equation:
\begin{align}
    \sum_{i\in\Ic_{\text{test}}} D_i(\beta;\hat G) = 0. \label{eq:EE_U(beta)}
\end{align}
By plugging in the expression \eqref{eq:U(beta)} into \eqref{eq:EE_U(beta)} and solve for $\beta$, we get a closed-form solution 
\begin{align*}
    \hat\beta 
    & = \frac{\sum_{i\in\Ic_{\text{test}}}  \ind(X_i>\tau) \ind\{X_i\in\Cc(\bar Z_{\tau,i})\} \cdot \Delta_i/\hat G(X_i|\bar Z_{X_i,i})}{\sum_{i\in\Ic_{\text{test}}}  \ind(X_i>\tau)\cdot \Delta_i/\hat G(X_i|\bar Z_{X_i,i})} \\
    & = \frac{\sum_{i\in\Ic_{\text{test}}, X_i>\tau}  \ind\{X_i\in\Cc(\bar Z_{\tau,i})\} \cdot \Delta_i/\hat G(X_i|\bar Z_{X_i,i})}{\sum_{i\in\Ic_{\text{test}}, X_i>\tau}  \Delta_i/\hat G(X_i|\bar Z_{X_i,i})}. 
\end{align*}

\paragraph{Remark.}
Evaluating the coverage of prediction intervals with right censored data is challenging. 
The evaluation strategy we use here is based on IPCW and is valid only when the censoring model is correctly specified.

In addition to IPCW, \citet{farina2025doubly} also considered a regression-based estimator, which requires correct modeling of the conditional event-time distribution given covariates, and an augmented IPCW (AIPCW) estimator, which is consistent if at least one of the censoring model or the event-time model is correctly specified.
Extending the regression-based or AIPCW-based evaluation strategies to survivor-conditional coverage in our setting with time-varying covariates is challenging because they require modeling not only the conditional event time distribution but also the covariate process.

Alternatively, \citet{sesia2025doubly} reported lower and upper bounds on the true miscoverage rate and used their average as a point estimate. However, this midpoint estimator is a heuristic, without a theoretical guarantee for convergence to the true miscoverage rate.

\subsection{Additional results for colon cancer data application}\label{app:application_colon}

\begin{figure}[!ht]
    \centering
    \includegraphics[width=0.65\linewidth]{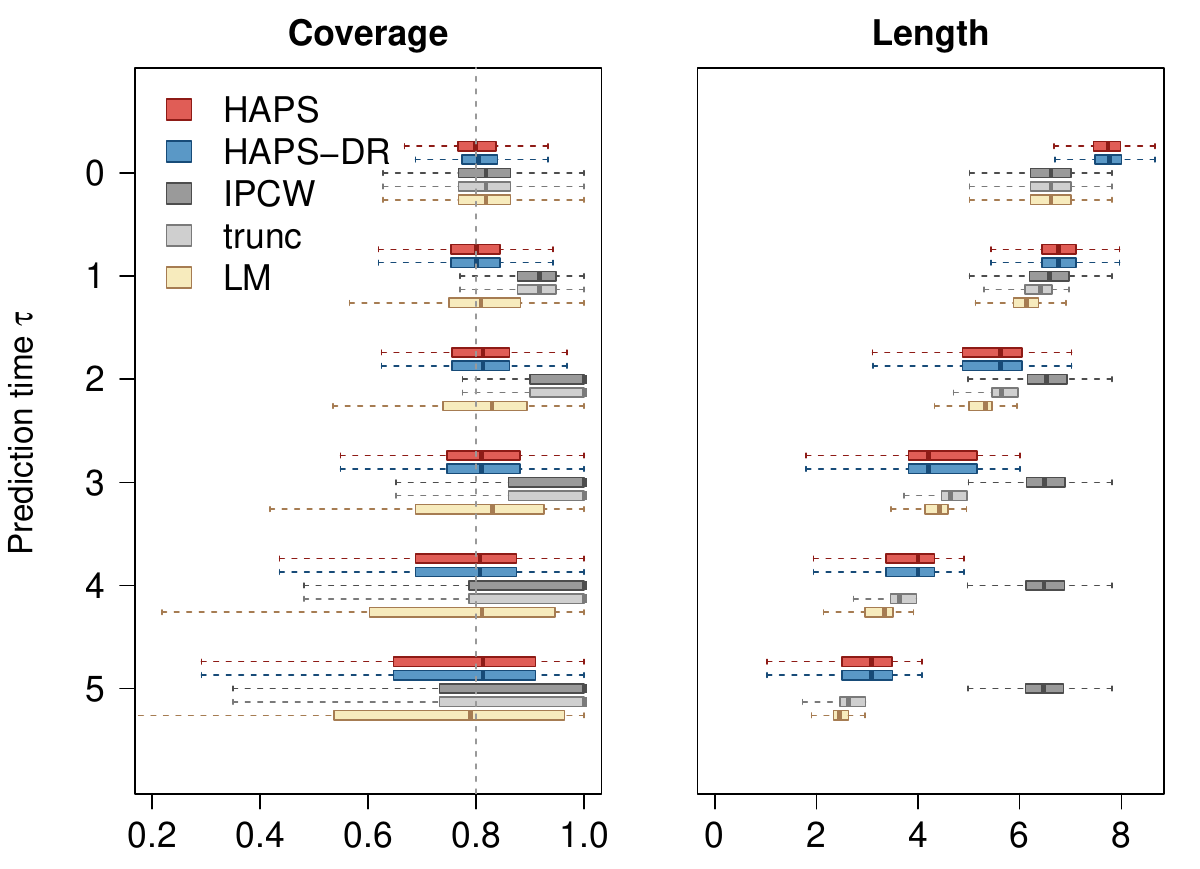}
    \caption{Summary of 80\% prediction intervals across 200 random splits for the colon cancer data; coverage is IPCW-estimated.}
\end{figure}

\subsection{Additional results for PBC data application}\label{app:application_pbc}

\begin{figure}[!ht]
    \centering
    \includegraphics[width=0.65\linewidth]{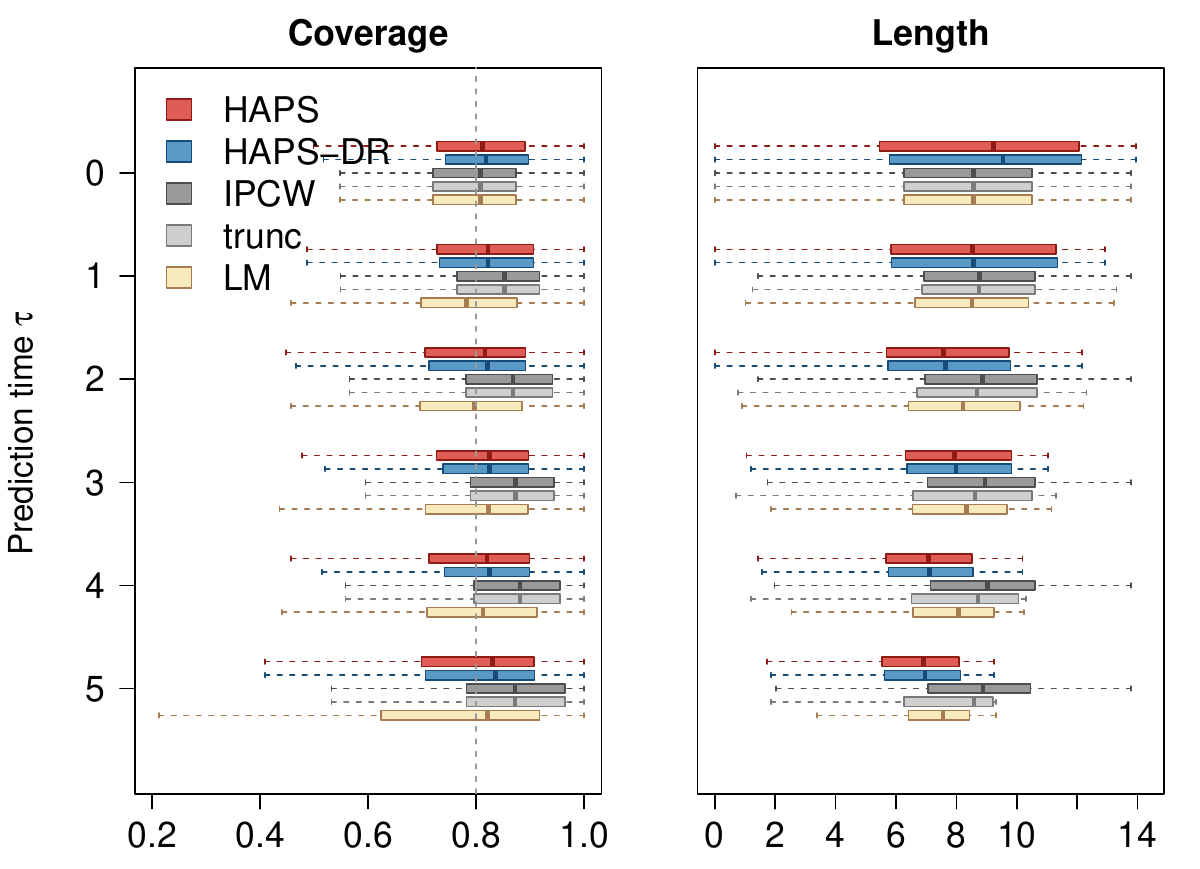}
    \caption{Summary of 80\% prediction intervals across 200 random splits for the PBC data; coverage is IPCW-estimated.}
\end{figure}

\subsubsection{Ratio of prediction interval lengths from HAPS and HAPS-DR relative to baseline methods}

\begin{table}[!ht]
\centering
\caption{Ratio of median interval length for HAPS and HAPS-DR relative to baseline methods in the two data applications.}
\label{tab:app_length_ratio}
{\fontsize{8}{9}\selectfont
\begin{tabular}{llccccccccc}
\toprule
&  & \multicolumn{3}{c}{$\tau=3$} & \multicolumn{3}{c}{$\tau=5$} \\
\cmidrule(lr){3-5} \cmidrule(lr){6-8} 
Data & Method  & IPCW & trunc & LM & IPCW & trunc & LM \\
\midrule
PBC & HAPS    & 0.90 & 0.94 & 0.98 & 0.78 & 0.87 & 0.91 \\
    & HAPS-DR & 0.90 & 0.94 & 0.98 & 0.78 & 0.87 & 0.91 \\
\midrule
Colon & HAPS    & 0.56 & 0.71 & 0.78 & 0.40 & 0.74 & 0.87 \\
      & HAPS-DR & 0.56 & 0.71 & 0.78 & 0.40 & 0.74 & 0.87 \\
\bottomrule
\end{tabular}
}
\end{table}


\FloatBarrier
\subsection{Sensitivity of the coverage estimates to the censoring model}\label{app:app-eval-sens}

The coverages reported in the data applications are IPCW estimates, and they approximate the true coverage only when the censoring weights are close to the truth. Here we analyze the sensitivity of the coverage estimates to the censoring models used.

\paragraph{Coverage estimates under different censoring models.}
To evaluate how sensitive the coverage estimate is to the censoring model used in the IPCW evaluation weights, we compute additional IPCW coverage estimates under two censoring models:
\begin{itemize}
    \item KM: the Kaplan--Meier estimator for the censoring survival curve, which does not incorporate covariates.
    \item RSF: a random survival forest with baseline covariates (since existing RSF software does not handle time-varying covariates).
\end{itemize}
Table~\ref{tab:app-eval-sens} summarizes the changes in the IPCW coverage estimates for the two data applications under these evaluation censoring models, relative to the estimates reported in Section~\ref{sec:application}. Almost all of the differences are no larger than .02 in absolute value, indicating that the coverage estimates are robust to the censoring model used for evaluation. This sensitivity analysis cannot rule out potential misspecification of the censoring model used for evaluation; however, it demonstrates the stability of the IPCW coverage estimate under different censoring models.

\paragraph{Additional comparison to empirical coverage in the uncensored subgroup.}
We also report, in Table~\ref{tab:app-cov-estimators}, the empirical coverage among the uncensored subgroup, for which the true event times are observed. However, we note that this empirical coverage is for the subgroup $\{T > \tau, T < C\}$, which is different from the survivor subgroup $\{T > \tau\}$ targeted by the proposed method. Table~\ref{tab:app-cov-estimators} summarizes the coverage estimates for HAPS under different estimators, as well as the number of uncensored subjects that the estimates are based on. We observe that for both applications, at $\tau = 5$, only about 10 uncensored subjects are used in the coverage estimate, reflecting that the coverage estimates may be unreliable.

\begin{table}[ht]
\centering
\small
\caption{Changes in the IPCW estimate for coverage in the data applications, relative to the estimate with the Cox model in Section~\ref{sec:application}.}
\label{tab:app-eval-sens}
\begin{tabular}{lllccc}
\toprule
Data & Method & Censoring model for Eval. & $\tau = 0$ & $\tau = 3$ & $\tau = 5$ \\
\midrule
PBC   & HAPS & KM  & $-.003$ & $-.002$ & $-.001$ \\
PBC   & HAPS & RSF & $-.019$ & $-.010$ & $-.018$ \\
PBC   & IPCW & KM  & $+.007$ & $+.005$ & $+.005$ \\
PBC   & IPCW & RSF & $-.011$ & $-.015$ & $-.020$ \\
PBC   & LM   & KM  & $+.007$ & $+.020$ & $+.033$ \\
PBC   & LM   & RSF & $-.011$ & $+.002$ & $+.006$ \\
\midrule
colon & HAPS & KM  & $+.001$ & $+.001$ & $+.007$ \\
colon & HAPS & RSF & $+.002$ & $+.002$ & $+.002$ \\
colon & IPCW & KM  & $-.002$ & $-.000$ & $+.000$ \\
colon & IPCW & RSF & $-.001$ & $+.000$ & $+.000$ \\
colon & LM   & KM  & $-.002$ & $-.004$ & $-.011$ \\
colon & LM   & RSF & $-.001$ & $-.002$ & $-.005$ \\
\bottomrule
\end{tabular}
\end{table}

\begin{table}[ht]
\centering
\small
\setlength{\tabcolsep}{4pt}
\caption{The mean coverage estimate (Monte Carlo SE) of HAPS over 200 splits, $\alpha = 0.2$. ``\% uncens.'' and ``\# uncens.'' refer to the uncensored percentage and the number of uncensored subjects within the at-risk subgroup $X > \tau$ in a single test split (recall that $X = \min(T, C)$), averaged over 200 splits.}
\label{tab:app-cov-estimators}
\begin{tabular}{llccccrr}
\toprule
Dataset & $\tau$ & IPCW-Cox & IPCW-KM & IPCW-RSF & Complete-case & \% uncens. & \# uncens. \\
\midrule
colon & 0 & .800 (.004) & .801 (.004) & .802 (.004) & .780 (.004) & 48.8 & 90.2 \\
      & 3 & .803 (.007) & .804 (.007) & .805 (.007) & .761 (.008) & 23.9 & 29.8 \\
      & 5 & .767 (.013) & .774 (.013) & .770 (.013) & .684 (.013) &  9.6 &  9.7 \\
\midrule
PBC   & 0 & .799 (.008) & .796 (.008) & .780 (.008) & .761 (.008) & 44.8 & 27.7 \\
      & 3 & .794 (.010) & .792 (.010) & .784 (.010) & .721 (.011) & 33.0 & 16.1 \\
      & 5 & .781 (.012) & .780 (.012) & .763 (.013) & .698 (.013) & 26.1 & 10.6 \\
\bottomrule
\end{tabular}
\end{table}

\FloatBarrier
\subsection{Effective sample size, weight truncation, and positivity}\label{app:app-alpha}

From the coverage bound in Theorem~\ref{thm:IPCW_coverage}, the coverage gap for HAPS can be large when $\eta$ is small (i.e., under positivity violations for censoring) or when $\sttau$ is small (which happens when $\tau$ is very large, i.e., prediction at the tail).

Miscoverage under positivity violations can also be viewed from the effective sample size perspective. Recall that the estimating function in \eqref{eq:IPCW} involves the weights $\Delta \ind(X > \tau)/G(X \mid \bar Z_X)$, so only uncensored subjects with $X > \tau$ contribute to the calibration inequality. Therefore, when $\tau$ is large or when $G(X \mid \bar Z_X)$ is close to zero for some individuals, the effective sample size used for calibration is small. Truncating the weights can improve the effective sample size, but it can introduce bias when estimating $\theta$, which can lead to bias in coverage. (See also Table~\ref{tab:app-cov-estimators} and the discussion in Appendix~\ref{app:app-eval-sens} on the unreliability of the coverage estimates for the two applications at $\tau = 5$ due to the small number of uncensored subjects.)

With extreme weights, and especially for small $\alpha$, the algorithm may return very wide or even trivial intervals. Any subject whose standardized weight exceeds $\alpha$ must be covered for the weighted calibration inequality to hold. Since $G(X \mid \bar Z_X)$ is small for large $X$, such subjects tend to be those with large $X$, and covering them forces the candidate set in \eqref{eq:Cc_condi_quantile} toward wide intervals. In small samples, the inequality may even be infeasible, in which case we return the largest candidate interval, i.e., $\theta = 1$. The IPCW coverage estimate is likewise unreliable here, as it is dominated by the few subjects with large weights.

Table~\ref{tab:app-alpha-trivial} shows the fraction of the 200 replications that selected $\hat\theta = 1$, i.e., the largest interval in the candidate set. We observe that this fraction is larger with smaller $\alpha$, indicating that the prediction intervals are not very informative in such settings. In Table~\ref{tab:app-alpha-cov}, we also report the IPCW-estimated coverage and median interval length of HAPS for the two data applications under different $\alpha$.

\begin{table}[ht]
\centering
\small
\caption{Fraction of the 200 replications with $\hat\theta = 1$ in HAPS.}
\label{tab:app-alpha-trivial}
\begin{tabular}{llcccc}
\toprule
Data & $\tau$ & $\alpha = 0.2$ & $\alpha = 0.1$ & $\alpha = 0.05$ & $\alpha = 0.01$ \\
\midrule
PBC   & 0 & .065 & .540 & .910 & 1.000 \\
PBC   & 3 & .095 & .400 & .725 & .960  \\
PBC   & 5 & .070 & .280 & .605 & .915  \\
colon & 0 & .000 & .030 & .470 & .995  \\
colon & 3 & .065 & .540 & .860 & .995  \\
colon & 5 & .090 & .390 & .730 & .950  \\
\bottomrule
\end{tabular}
\end{table}

\begin{table}[ht]
\centering
\small
\caption{Mean IPCW-estimated coverage and median interval length (in parentheses) of HAPS among the 200 replications.}
\label{tab:app-alpha-cov}
\begin{tabular}{llcccc}
\toprule
Data & $\tau$ & $\alpha = 0.2$ & $\alpha = 0.1$ & $\alpha = 0.05$ & $\alpha = 0.01$ \\
\midrule
PBC   & 0 & .799 (8.36) & .904 (13.33) & .942 (13.68) & .950 (13.72) \\
PBC   & 3 & .794 (7.67) & .886 (9.14)  & .926 (10.54) & .936 (10.70) \\
PBC   & 5 & .781 (6.62) & .897 (7.82)  & .939 (8.82)  & .955 (8.98)  \\
colon & 0 & .800 (7.75) & .902 (8.29)  & .942 (8.54)  & .946 (8.59)  \\
colon & 3 & .803 (4.33) & .887 (5.66)  & .908 (5.88)  & .915 (5.91)  \\
colon & 5 & .767 (3.03) & .878 (3.57)  & .915 (3.96)  & .926 (3.99)  \\
\bottomrule
\end{tabular}
\end{table}
\FloatBarrier

\section{Additional information}

\subsection{Computational resources for experiments}
\label{app:compute}

All experiments were run on CPUs; no GPU acceleration was used. Smaller-scale
experiments and real-data analyses were conducted on a MacBook Air, 13-inch,
2024, with an Apple M3 chip and 16 GB memory. For the main laptop-based
simulation experiments with $R=200$, $n=1000$, and $n_{\mathrm{test}}=500$,
the Cox-based DGM1 experiments took a few minutes per method. The DDH experiment
under DGM2 with RSF censoring estimation and XGB nuisance-parameter estimation
was the most computationally intensive local run, requiring approximately
4.7 hours in total, or about 84 seconds per Monte Carlo replication, while
constructing HAPS, HAPS-DR, HAPS-A, and the raw DDH intervals in the same run.
Analogous DDH runs with a Cox censoring model required approximately 1 hour.
The two real-data DDH analyses with $R=200$ required approximately 6--20 minutes,
depending on the dataset and censoring model.

The larger sample-size experiments were run on CPU through Open
Science Grid (OSG) \citep{osg07, osg09, https://doi.org/10.21231/906p-4d78, https://doi.org/10.21231/0kvz-ve57}
as parallel array jobs. For the appendix sample-size
experiments at nominal 80\% coverage, the median runtime per prediction time
was approximately 32--44 seconds for $n=300$, 55--58 seconds for $n=500$,
117--125 seconds for $n=1000$, 389--444 seconds for $n=2000$, and
1243--1548 seconds for $n=4000$. Equivalently, the median runtime for one
Monte Carlo replication across the three prediction times was approximately
2.2 minutes for $n=300$, 3 minutes for $n=500$, 6 minutes for $n=1000$,
20--22 minutes for $n=2000$, and 63--79 minutes for $n=4000$. 
These OSG jobs were run in parallel, so the serial runtimes should not be interpreted as the
wall-clock time required for the full experiment.

\subsection{Existing assets and licenses}\label{app:existing_assets}

The real-data applications use the publicly available \texttt{colon} and
\texttt{pbcseq} datasets distributed with the \texttt{survival} R package.
We cite the original studies associated with these datasets and use the datasets
in accordance with the \texttt{survival} package license.

Our implementation builds on existing open-source software. The R code uses
\texttt{survival}, \texttt{randomForestSRC}, \texttt{xgboost}, \texttt{dplyr},
\texttt{tidyr}, \texttt{ggplot2}, \texttt{scales}, and \texttt{reticulate}.
The neural-network prediction model is implemented in Python through
\texttt{reticulate} and uses \texttt{torch}, \texttt{numpy}, and \texttt{pandas}.
These packages are used under their corresponding open-source licenses. The
versions of all R and Python packages used in our experiments are listed in the
\texttt{README.md} file included with the supplementary code.

The implementation of the DDH prediction model was adapted from the publicly available Dynamic-DeepHit implementation by Lee, Yoon, and van der Schaar
(\url{https://github.com/chl8856/Dynamic-DeepHit}), associated with \citet{lee2019dynamic}. The repository did not specify an explicit software license at the time of writing.



\end{document}